\documentclass[a4paper, 11pt]{article}
\usepackage[utf8]{inputenc}

\usepackage{mathtools}
\usepackage{bbm}
\usepackage{amsmath}
    \allowdisplaybreaks
\DeclareMathOperator{\unif}{Unif}
    \DeclareMathOperator{\tr}{Tr}
    \DeclareMathOperator{\cov}{Cov}
    \DeclareMathOperator{\bvn}{BvN}
    \DeclareMathOperator{\esjd}{ESJD}
    \DeclareMathOperator{\ejc}{EJC}
    \DeclareMathOperator{\argmin}{argmin}
    \DeclareMathOperator{\argmax}{argmax}
    
    \DeclareMathOperator{\norm}{Nor}
    \DeclareMathOperator{\Sbar}{\bar{\mathcal{S}}}
    \DeclareMathOperator{\diag}{Diagonal}
    \DeclareMathOperator{\bigO}{\mathcal{O}}
    \DeclareMathOperator{\ecc}{\varepsilon}
    \newcommand{\wass}[1]{\mathcal{W}_{#1}}
    \DeclareMathOperator{\tvd}{\mathcal{TV}}
    
    \DeclareMathOperator{\ceil}{Ceiling}
    \DeclareMathOperator{\iid}{\overset{\smash{\textup{iid}}}{\sim}}
\usepackage{amssymb}
\usepackage{amsthm}
    
    \newtheorem{example}{Example}
    \newtheorem{theorem}{Theorem}
    \newtheorem{proposition}{Proposition}
    
    \newtheorem{remark}{Remark}
    \newtheorem{lemma}{Lemma}
    
    \newtheorem{assumption}{Assumption}
\usepackage[margin = 2.5cm]{geometry}

\usepackage{xcolor}
    \definecolor{wongblue}{HTML}{0072B2}
    \definecolor{wongorange}{HTML}{D55E00}

\usepackage{hyperref}
\hypersetup{
    hidelinks=true,
    pdfencoding=auto,
    colorlinks=true,
    citecolor=wongblue,
    linkcolor=wongorange,
    urlcolor=wongorange
}
\usepackage{graphicx}

\usepackage{authblk}
\usepackage[round]{natbib}
\usepackage{bibunits}
\defaultbibliographystyle{plainnat-minimal}
\defaultbibliography{refs}

\usepackage{algpseudocode,algorithm,algorithmicx}

\title{Scalable couplings for the random walk Metropolis algorithm}

\author[1]{Tam\'as P. Papp\thanks{\href{mailto:t.papp@lancaster.ac.uk}{\texttt{t.papp@lancaster.ac.uk}}}}
\author[2]{Chris Sherlock}
\affil[1]{\small{STOR-i Centre for Doctoral Training, Lancaster University, Lancaster, UK.}}
\affil[2]{\small{School of Mathematics, Lancaster University, Lancaster, UK.}}
\date{}

\begin{document}
\maketitle

\begin{bibunit}

\begin{abstract}
	There has been a recent surge of interest in coupling methods for Markov chain Monte Carlo algorithms: they facilitate convergence quantification and unbiased estimation, while exploiting embarrassingly parallel computing capabilities. Motivated by these, we consider the design and analysis of couplings of the random walk Metropolis algorithm which scale well with the dimension of the target measure. Methodologically, we introduce a low-rank modification of the synchronous coupling that is provably optimally contractive in standard high-dimensional asymptotic regimes. We expose a shortcoming of the reflection coupling, the state of the art at the time of writing, and we propose a modification which mitigates the issue. Our analysis bridges the gap to the optimal scaling literature and builds a framework of asymptotic optimality which may be of independent interest. We illustrate the applicability of our proposed couplings, and the potential for extending our ideas, with various numerical experiments.
\end{abstract}


\section{Introduction} \label{sec:intro}

Couplings of Markov chain Monte Carlo (MCMC) algorithms have attracted much interest recently due to their ability to nullify the bias of MCMC estimates \citep{jacob2020unbiased}, conservatively estimate the rate of convergence \citep{biswas2019estimating} of MCMC algorithms and the asymptotic bias of approximate inference procedures \citep{biswas2023bounding}, as well as unbiasedly estimate the asymptotic variance of MCMC algorithms \citep{douc2023solving}. One appeal of these methods is that they are able to exploit parallelism without requiring communication between processors.

In the context of unbiased MCMC and the related convergence quantification methodology, a Markovian coupling of two chains should be designed such that the chains meet after a finite number of iterations. The meeting time acts as a lower bound for the length of the simulation, and as such the efficiency of a coupling can be assessed through the distribution of the meeting time. As remarked in \cite{jacob2020unbiased}, it is paramount to design couplings which have favourable high-dimensional properties, in the sense that the meeting times reflect the true rate of convergence and mixing of the underlying marginal Markov chains. At the same time, the rich design space means that it is challenging to devise efficient couplings, and this design is regarded to be an art form in general \citep{wang2023metropolishastings}.

We focus here on the design and analysis of scalable couplings for the random walk Metropolis (RWM) algorithm with Gaussian proposals. Methodologically (see Section~\ref{sec:prelim-coupling}), we argue for the importance of couplings which are contractive and we design couplings which attempt to optimize for contraction in squared Euclidean distance. Our \emph{Gradient Common Random Numbers} (GCRN) coupling is provably optimally contractive in certain high-dimensional asymptotic regimes, is insensitive to the eccentricity of the target, and is consistently able to contract the chains to within a distance where coalescence in one step is achievable. Once the chains are close, GCRN or other contractive couplings can be swapped for ones which allow for exact meetings; we exemplify and advocate for such two-scale strategies in the sequel.

Our proposed couplings are designed to overcome the shortcomings of those currently available. The reflection-maximal coupling \citep[Section 4.1]{jacob2020unbiased}, arguably the most promising candidate to date, has been seen to scale well with the dimension when the target is spherically symmetric \citep{jacob2020unbiased,oleary2021couplings}. However, for high-dimensional targets which do not possess spherical symmetry this coupling has been seen to perform poorly \citep{papp2022paper}; it does not contract the chains sufficiently unless the step size of the coupled RWM algorithms is chosen to be significantly smaller than is optimal for mixing. The present work validates the favourable behaviour of the reflection coupling in the spherical case, offers an explanation for the issue when spherical symmetry is lacking, and proposes an alternative reflective coupling (Section~\ref{sec:correctedrefl}) which alleviates the problem.

Our analysis bridges the gap between the coupled sampling and the optimal dimensional scaling literatures in MCMC (see the review \citealp{roberts2001optimal} for optimal scaling in the stationarity phase, and \citealp{christensen2005scaling} for the transient phase). The ODE limit in \citet[Theorem 1]{christensen2005scaling} for high-dimensional spherical Gaussian targets, originally developed to explain the transient phase of the RWM algorithm, underpins our approach. We extend this scaling limit to coupled pairs of RWM chains in the spherical Gaussian case (Section~\ref{sec:stdgauss}) and further extend the salient points of this analysis to the elliptical Gaussian case (Section~\ref{sec:elliptgauss}). Related to these scaling limits, we introduce a notion of asymptotic optimality; this extends beyond the Gaussian case (Section~\ref{sec:prod-target}) and may guide the design of effective couplings for other MCMC algorithms.

We conclude with a series of experiments illustrating the practical appeal of our proposed couplings (Section~\ref{sec:applications}) and a discussion of our findings and directions for further work (Section~\ref{sec:discussion}).

\section{Background} \label{sec:background}

This work is motivated by lagged coupling methodology \citep{jacob2020unbiased,biswas2019estimating}, which can be used to estimate the rate of convergence of MCMC algorithms and to obtain unbiased MCMC estimators, and which we briefly recall here. Our set-up differs slightly from the literature in that we start the index of the first chain at $-L$ rather than $0$, however this will add clarity to the sequel.

Consider a $\pi$-invariant Markov kernel $K$, and a joint Markov kernel $\bar{K}((x,y), (\cdot,\cdot))$ with marginals $(K(x, \cdot),K(y, \cdot))$. Throughout this paper, we think of $K$ as a Metropolis-Hastings kernel and the joint kernel $\bar{K}$ as specifying some coupling of the proposals and of the acceptance steps. We construct two Markov chains $(X_t)_{t \ge -L}$ and $(Y_t)_{t \ge 0}$. Each chain evolves marginally according to $K$, and after head-starting the $X$-chain by $L \ge 1$ iterations they evolve jointly according to $\bar{K}$:
\begin{enumerate}
    \item Sample $X_{0}\sim \pi_0 K^L$ and $Y_0\sim \pi_0$.
    \item Sample $(X_{t+1}, Y_{t+1}) \mid (X_t, Y_t) \sim \bar{K}((X_t, Y_t), \cdot) $ for $t \ge 0$.
\end{enumerate}
Furthermore, we design the joint kernel $\bar{K}$ such that there is an almost surely finite meeting time $\tau = \inf\{t \ge 0: X_t = Y_t\}$ and such that $X_t = Y_t$ for all $t \ge \tau$.

This construction can be used to estimate the rate of convergence of Markov chains \citep{biswas2019estimating}. Suppose that we wish to quantify the rate of convergence in a $p$-Wasserstein distance \citep{villani2009optimal} of order $p \ge 1$, defined as $\wass{p}(\mu, \nu) = \inf_{ (X,Y) \in \Gamma(\mu,\nu)}\mathbb{E}[c(X,Y)^p]^{1/p},$ where $c$ is some real-valued ground metric and where $\Gamma(\mu,\nu)$ is the set of all couplings of the distributions $(\mu,\nu)$. Let $\pi_t = \pi_{0}K^t$ be the time-$t$ marginal distribution of the $Y$-chain, so that $Y_t \sim \pi_t$ for all $t \ge 0$. Using firstly the triangle inequality, then the definition of the Wasserstein distance, we have that
\begin{equation} \label{eqn:wass-general-ub}
    \wass{p}(\pi, \pi_t) = \wass{p}(\pi_\infty, \pi_t) \le \sum_{j \ge 0} \wass{p}(\pi_{t + (j + 1) L}, \pi_{t + jL}) \le \sum_{j \ge 0} \mathbb{E}^{1/p}\left[ c(X_{t+jL},Y_{t+jL})^p\right].
\end{equation}
By repeatedly simulating the pair $(X,Y)$ and by replacing expectations with empirical averages, we obtain a consistent estimator of this upper bound. Conveniently, pairs of coupled chains can be simulated in parallel; the meeting times $\tau$ ensure that the estimator is computed in finite time. The lagged coupling framework also allows for the unbiased estimation of expectations of test functions of interest (see \citealp{jacob2020unbiased} and Appendix~\ref{app:unbiased-mcmc}) which enables principled parallel MCMC through averaging across multiple (pairs of) chains simulated in parallel. Such unbiased estimators also facilitate modular inference with cut posterior distributions \citep{jacob2017better} and can be used as part of a wider multilevel Monte Carlo framework in order to unbiasedly estimate functions of expectations \citep{wang2023uniased-multilevel}.

The effectiveness of the lagged coupling framework relies on designing the joint kernel $\bar{K}$ such that the meeting times $\tau$ are small, since large meeting times lengthen the duration of the simulation, loosen the upper bounds on the rate of convergence, and inflate the variance of unbiased estimators (see \citealp{jacob2020unbiased} for this final point). In this paper, we focus on designing kernels $\bar{K}$ whose meeting times scale well with the dimension of the target $\pi$. Since the choice of lag parameter $L$ should be guided by the choice of coupling kernel $\bar{K}$, for the remainder of the paper we focus on the joint Markov chain $(X_t,Y_t)_{t \ge 0}$, where the starting conditions $(X_0,Y_0)$ can be arbitrary.

We illustrate the use of couplings to quantify the bias of approximate sampling algorithms \citep{biswas2023bounding} and to unbiasedly estimate the asymptotic variance of MCMC algorithms \citep{douc2023solving} in the experiments of Section~\ref{sec:applications}.

\section{Couplings of the RWM algorithm} \label{sec:prelim-coupling}

We first restrict our attention to the random walk Metropolis (RWM) kernel $K$ with spherical Gaussian proposals. All coupling kernels $\bar K$ are induced by joint updates of the form
\begin{equation} \label{eqn:rwm-updates}
    X_{t+1} = X_t + h Z_x B_x,\quad
    Y_{t+1} = Y_t + h Z_y B_y,
\end{equation}
where $B_x = \mathbbm{1}\{\log U_x \le  \log\pi(X_t + hZ_x) - \log \pi(X_t)\}$ is the Bernoulli acceptance indicator, $\pi$ is a $d$-dimensional target, and $Z_x\sim\mathcal{N}_d(0_d, I_d)$ and $U_x\sim\unif(0,1)$ are independent. Analogous notation applies to the $Y$-chain. Throughout this paper, we scale the step size as $h = \ell d^{-1/2}$, which ensures that acceptance rates remain stable as the dimension grows \citep{roberts1997weak, christensen2005scaling}.

While the simplicity of spherical proposals is convenient for analysis, in practice better mixing may be obtained with other covariance structures. It is straightforward to extend the couplings considered in this paper to non-spherical Gaussian proposals, see Appendix~\ref{app:precondition}.

\subsection{The importance of contractivity in high dimensions}

To try to make the chains meet quickly, one might be tempted to use a coupling $\bar K$ which maximizes the chance of coalescing the chains at each iteration (see \cite{oleary2021maximal} for examples). However, as the dimension grows, such couplings can perform increasingly poorly for the RWM algorithm.

The issue is that the RWM proposals become increasingly local as the dimension increases, yet the same time the distance between the coupled chains grows. To be able to coalesce the chains, a Markovian coupling of RWM chains must first propose the same state in both chains. The probability of coalescing the chains in one iteration is therefore upper bounded by the volume of overlap of the proposal densities (which is analytically tractable in terms of the standard Gaussian cumulative density function $\Phi(\cdot)$, e.g. \citealp{heng2019unbiased}):
\begin{equation} \label{eqn:prob-meeting}
    \mathbb{P}(X_{t+1} = Y_{t+1} \mid X_t, Y_t) \le 2\Phi\left( -\frac{1}{2h}\lVert X_t-Y_t\rVert \right) \le 2\exp\left(- \frac{d}{8\ell^2}\lVert X_t-Y_t\rVert^2\right),
\end{equation}
where we finally used the Chernoff bound and that $h = \ell d^{1/2}$. Two independent chains will typically start $\lVert X_0-Y_0\rVert^2 = \bigO(d)$ apart, yet the chance of coalescing in one iteration is infinitesimally small unless $\lVert X_t-Y_t\rVert^2= \bigO(d^{-1})$. If the coupling cannot contract the chains to within, say, $\bigO(1)$ squared distance, then one can expect meeting times to grow exponentially with the dimension $d$. Clearly, this is in stark contrast with the $\bigO(d)$ time required for RWM chains to converge \citep{christensen2005scaling,andrieu2024explicit}.

In high dimensions, the coupling should therefore primarily focus on contracting the chains, as opposed to maximizing the probability of coalescing at each iteration. For the coupling to be scalable, we therefore need to design a contractive kernel $\bar K$.

\subsection{Optimizing for contraction}

Setting up the design of $\bar K$ as an optimization problem, a natural objective is the expected contraction in squared Euclidean distance. It is straightforward, following the expansion~\eqref{eqn:innerprod-onestep} below, to show that
\begin{equation*}
    \argmin_{\bar K \in \mathcal{C}}\mathbb{E}\big[\lVert X_{t+1} - Y_{t+1} \rVert^2 \mid X_t, Y_t\big] = \argmax_{\bar K \in \mathcal{C}} \mathbb{E}
    \big[h^2 Z_x^{\top} Z_y B_x B_y \mid X_t, Y_t\big],
\end{equation*}
where $\mathcal{C}$ is any subset of $\mathcal{M}$, the class of all Markovian couplings~\eqref{eqn:rwm-updates}. We call the quantity $\mathbb{E}\big[h^2 Z_x^{\top} Z_y B_x B_y \mid X_t, Y_t\big]$ the \emph{expected jump concordance} ($\ejc$). To optimize the expected contraction of the chains, we therefore need to maximize the $\ejc$.

The form of the $\ejc$ suggests the following coupling strategy: correlate the acceptances $B_x, B_y$ maximally (i.e. try to accept simultaneously in both chains) and correlate the proposal noises $Z_x, Z_y$ maximally. The key observation is that, \emph{as the dimension increases, these two objectives become less and less constrained by each other, so it becomes possible to satisfy both of them simultaneously}. Focusing on the acceptance steps, a Taylor expansion of the log acceptance ratio yields
\begin{equation} \label{eqn:acc-step-Taylor}
    B_x = \mathbbm{1}{\left\{\log U_x \le h Z_x^\top \nabla\log \pi(X_t) + (h^2/2) Z_x^\top \nabla^2 \log \pi(\bar X_t) Z_x \right\}},
\end{equation}
where $\nabla^2$ denotes the Hessian, and $\bar X_t$ is on the line segment from $X_t$ to $X_t + hZ_x$. As a function of $Z_x$, most of the variation in $B_x$ therefore stems from the projection of $Z_x$ onto the logarithmic gradient $\nabla\log \pi(X_t)$. To try to make the chains accept simultaneously as often as possible, it is therefore sensible to maximally correlate $\{Z_x^\top \nabla\log \pi(X_t), Z_y^\top \nabla\log \pi(Y_t)\}$. As each projection only constrains a single coordinate, in high enough dimensions $\{Z_x,Z_y\}$ can still be correlated nearly maximally. This is the motivation behind the gradient-based GCRN coupling, which we introduce in Section~\ref{sec:coupling-list}.

The $\ejc$ is well-defined in standard high-dimensional asymptotic regimes; we will call a coupling \emph{asymptotically optimal} if it maximizes the $\ejc$ in the limit. In such regimes, the variation from the Hessian term of~\eqref{eqn:acc-step-Taylor} vanishes \citep[e.g.][]{sherlock2013optimal}, so we expect couplings like GCRN which account for all of the variation from the gradient term to be close to optimal. The $\ejc$ can in fact be viewed as the coupling-based analogue of the expected squared jumping distance (ESJD; \citealp{sherlock2009optimal}), a widely-used measure of mixing efficiency. The ESJD has the limiting formula $\esjd(\ell) = 2 \ell^2 \Phi(-\ell/2)$ for a standard Gaussian target, which returns the optimal scaling of $\ell=2.38$ and optimal acceptance rate of $23.4 \%$. The $\ejc$ reduces to the $\esjd$ when the the chains are stationary and are identical; it is therefore unsurprising that the same scaling $\ell=2.38$ turns out (see Section~\ref{sec:dim-scaling-stdgauss}) to be close to optimal for our GCRN coupling.

\subsection{The couplings under consideration} \label{sec:coupling-list}

This work focuses on the natural class of \emph{product couplings} $\mathcal{P}$, which contains all couplings of the updates~\eqref{eqn:rwm-updates} such that $(U_x,U_y)$ are independent of $(Z_x,Z_y)$. We introduce three couplings from $\mathcal{P}$, all of which synchronize the acceptance uniforms to $U_y = U_x$, also called a common random numbers (CRN) strategy. The difference is in the coupling of the proposal increments $(Z_x,Z_y)$:
\begin{enumerate}
    \item \textbf{CRN:} $Z_y = Z_x$;
    \item \textbf{Reflection:} $Z_y = Z_x - 2 (e^{\top} Z_x) e$; 
    \item \textbf{GCRN:} $Z_x = Z - (n_x^{\top}Z)n_x + Z_\nabla n_x$ and $Z_y = Z - (n_y^{\top}Z)n_y + Z_\nabla n_y$; 
\end{enumerate}
where: $e = \norm(X_t-Y_t)$, $n_x = \norm(\nabla\log\pi(X_t))$ and $n_y = \norm(\nabla\log\pi(Y_t))$; $\norm(x) = x / \lVert x \rVert$ denotes the normalization operation; $Z_\nabla\sim \mathcal{N}(0,1)$ and $Z\sim \mathcal{N}_d(0_d,I_d)$ are independent. The CRN and reflection couplings serve as baselines and are already established (see e.g. \citealp{oleary2021couplings}). By convention, we default to the CRN coupling when a vector to be normalized is null.

The new GCRN (\emph{Gradient Common Random Numbers}) coupling attempts to synchronize acceptance events by ensuring that $n_x^\top Z_x = n_y^\top Z_y$; it contracts the chains due to synchronized movement towards the mode. By design and in certain high-dimensional asymptotic regimes, GCRN is optimal for contraction within the class $\mathcal{P}$ (see Theorems~\ref{thm:gcrn-opt-stdgauss}, \ref{thm:gcrn-opt-elliptgauss} and~\ref{thm:gcrn-opt-product}). We construct in Appendix~\ref{app:optimal-markovian} an implementable modification of GCRN that is asymptotically optimal over the entire class $\mathcal{M}$; we prefer the simpler GCRN coupling as any additional gain appears very small (see Figures~\ref{fig:spherical-squaredist} and~\ref{fig:elliptical-squaredist}). Other variants of GCRN with similar high-dimensional properties could be devised, for instance synchronizing $n_x^\top Z_x = n_y^\top Z_y$ using an appropriate reflection or rotation. Aiming to combine favourable properties of contractive and reflective couplings, we propose a hybrid between the GCRN and reflection couplings in Section~\ref{sec:correctedrefl}.

In practice, once the chains become close enough to have a reasonable chance of coalescing in one step, we propose to swap from a contractive coupling like GCRN to one that allows for exact meetings. We use such two-scale approaches in the experiments of Section~\ref{sec:applications}; our coalescive coupling of choice is the \emph{reflection-maximal} coupling (\citealp[Section~4.1]{jacob2020unbiased}; see also Appendix~\ref{app:precondition}). Two-scale couplings have previously been considered in e.g. \cite{biswas2022coupling-based,bou-rabee2020coupling}.

As a major component to this work, we develop theory which explains the behaviour of coupled RWM chains in high dimensions. We focus on GCRN and the two baselines; a recurring quantity in our analysis is the joint distribution of $(n_x^{\top}Z_x, n_y^{\top}Z_y)$ which we characterize in Proposition~\ref{prop:corr-grad-proj} below. As with all our results in the main text, this is proved in Appendix~\ref{sec:proofs}.

\begin{proposition} \label{prop:corr-grad-proj} It holds that $(n_x^{\top}Z_x, n_y^{\top}Z_y) \mid \{X_t, Y_t\} \sim \bvn(\rho)$, the bivariate normal distribution with unit coordinate-wise variances and correlation $\rho \in [-1,1]$, where the coupling-specific value $\rho$ is
    \begin{equation*}
       \rho_{\textup{crn}} = n_x^{\top}n_y, \quad \rho_{\textup{refl}} = n_x^{\top}n_y - 2(n_x^{\top}e)(n_y^{\top}e), \quad \rho_{\textup{gcrn}} = 1.
    \end{equation*}
\end{proposition}

\section{Analysis: standard Gaussian case} \label{sec:stdgauss}

We begin our analysis of the couplings of Section~\ref{sec:coupling-list} by considering a standard Gaussian target $\pi^{(d)} = \mathcal{N}_d(0_d,I_d)$ in increasing dimension $d$. Clearly, this setting is very stylized, however it will allow us to obtain limit theorems that cleanly characterize the behaviour of the couplings in high dimensions. Besides, scaling limits for the RWM algorithm often rely on the target behaving asymptotically like a Gaussian \citep{roberts1997weak}, and the conclusions of such analyses have been seen to hold much more widely in practice.

Firstly, we show in Theorem~\ref{thm:gcrn-opt-stdgauss} below that the GCRN coupling is asymptotically optimal for contraction among the class $\mathcal{P}$ of product couplings, in that it optimizes a limiting form of the $\ejc$. This coupling is therefore expected to perform well even in high dimensions.

Thereafter, our analysis centers on the three-dimensional process
\begin{equation*}
    W^{(d)} = \big(W^{(d)}_t\big)_{t \ge 0} = \frac{1}{d}\left( \lVert X_{\lfloor td \rfloor}\rVert^2, \lVert Y_{\lfloor td \rfloor}\rVert^2, X_{\lfloor td \rfloor}^{\top} Y_{\lfloor td \rfloor} \right)_{t \ge 0},
\end{equation*}
where the speed-up factor $d$ corresponds to the natural time-scale under the step size scaling $h = \ell d^{-1/2}$. The form $B_x = \mathbbm{1}{\{\log U_x \le  -h Z_x^\top X_t - (h^2/2) \lVert Z_x \rVert^2 \}}$ of the acceptance step ensures that the process $W^{(d)}$ is Markovian under our couplings (see Appendix~\ref{sec:proofs}). As the target is spherically symmetric, the first and second coordinates of this process are radial components describing the marginal behaviour of each chain. The inner product in the third coordinate captures the remaining joint behaviour of the chains under a given coupling.

We show in Theorem~\ref{thm:ode-limit-stdgauss} that $W^{(d)}$ converges weakly to the solution of an ordinary differential equation (ODE) as the dimension $d$ grows. Our approach follows \cite{christensen2005scaling} and extends this path-breaking work to pairs of Markov chains. By thereafter analyzing the ODE, we shed light on the high-dimensional behaviour of the coupled chains.

\subsection{Asymptotic optimality} \label{sec:asymptotically-optimal-stdgauss}

It is natural to ask which coupling contracts the chains the most in the considered high-dimensional regime. As shown in Section~\ref{sec:prelim-coupling}, this is equivalent to asking which coupling maximizes the $\ejc$, which we take as our efficiency metric. By design, the GCRN coupling optimizes an asymptotic form of the $\ejc$ over the class $\mathcal{P}$ of product couplings, and is therefore asymptotically optimal over this class. In the sequel, we quantify the gap between $\mathcal{P}$ and $\mathcal{M}$ numerically.

\begin{theorem} \label{thm:gcrn-opt-stdgauss}
	Conditionally on $(\lVert X_t \rVert^2, \lVert Y_t \rVert^2, X_t^\top Y_t)/d = (x,y,v)$, it holds that
	\begin{equation*}
	    \adjustlimits\lim_{d\to\infty}\sup_{\bar{K} \in \mathcal{P}} \mathbb{E}\left[ h^2 Z_x^{\top} Z_y B_x B_y \right] = \ell^2\mathbb{E}_{Z\sim \mathcal{N}(0,1)} \left[1\land e^{\ell x^{1/2}Z - \ell^2/2}\land e^{\ell y^{1/2}Z - \ell^2/2}\right]
	\end{equation*}
	and this limit supremum is attained by the GCRN coupling.
\end{theorem}

\subsection{Scaling limits}

Turning to the ODE limit, we need to get a handle on the drift of the process $W^{(d)}$, as well as show that this process does not fluctuate much. In both cases, we work with conditional one-step differences in $(\lVert X_t\rVert^2, \lVert Y_t\rVert^2, X_t^{\top} Y_t)$.

Starting with limiting drift of the process $W^{(d)}$, its first two coordinates are coupling-invariant, are individually Markovian, and are dealt with by prior work \citep{christensen2005scaling}. For the third coordinate, the one-step difference is
\begin{equation} \label{eqn:innerprod-onestep}
    X_{t+1}^{\top} Y_{t+1} - X_t^{\top} Y_t = h Y_t^{\top} Z_x B_x + h X_t^{\top} Z_y B_y + h^2 Z_x^{\top} Z_y B_x B_y.
\end{equation}
The first two terms are coupling-invariant; the final term, the \emph{jump concordance}, is coupling-dependent. We evaluate the drift in Proposition~\ref{prop:drift-stdgauss} below. Finally, we show that the fluctuations of $W^{(d)}$ are negligible in Proposition~\ref{prop:fluctuations-stdgauss} below.

As a technical point, we fix $\bar{x},\bar{y}>0$ and define the set $\mathcal{S} = \{(x,y,v) \mid x \in [0,\bar{x}], y \in [0,\bar{y}], |v| \le \sqrt{xy}\}$. This is an arbitrarily large compact subset of $\bar{\mathcal{S}}$, the set of all feasible values of the process $W^{(d)}$. Our auxiliary results essentially cover all of $\bar{\mathcal{S}}$ as they hold uniformly over $\mathcal{S}$ for any fixed $\bar{x},\bar{y}$. This final detail is important, but for brevity we suppress it from notation.

\begin{proposition} \label{prop:drift-stdgauss}
Under the couplings of Section~\ref{sec:coupling-list} and uniformly over $w = (x,y,v) \in \mathcal{S}$, it holds that
\begin{equation*}
    \lim_{d \to \infty}\mathbb{E} \left[ d\big( W^{(d)}_{(t+1)/d} - W^{(d)}_{t/d}\big) \bigm\lvert W^{(d)}_{t/d} = w \right]  =  c_\ell(w) = \big(a_\ell(x), a_\ell(y), b_\ell(x,y,v)\big),
\end{equation*}
where 
\begin{align*}
a_\ell(x) ={}& \ell^2 (1-2x) e^{\ell^2(x-1)/2} \Phi\left(\frac{\ell}{2x^{1/2}} - \ell x^{1/2}\right) + \ell^2\Phi\left(-\frac{\ell}{2x^{1/2}}\right), \\
\begin{split}
b_\ell(x,y,v) ={}& \ell^2\mathbb{E}_{(Z_1,Z_2)\sim \bvn(\rho)} \left[1\land e^{\ell x^{1/2}Z_1 - \ell^2/2}\land e^{\ell y^{1/2}Z_2 - \ell^2/2}\right] \\
&-\ell^2 v\left[ e^{\ell^2(x-1)/2} \Phi\left(\frac{\ell}{2x^{1/2}} - \ell x^{1/2}\right) + e^{\ell^2(y-1)/2} \Phi\Big(\frac{\ell}{2y^{1/2}} - \ell y^{1/2}\Big) \right],
\end{split}
\end{align*}
and where $\rho = \rho(x,y,v)$ is coupling-specific:
\begin{equation*}
    \rho_\textup{crn} = \frac{v}{(xy)^{1/2}}, \quad  \rho_\textup{refl} = \frac{2xy-(x+y)v}{(xy)^{1/2}(x+y-2v)}, \quad \rho_\textup{gcrn} = 1.
\end{equation*}
\end{proposition}

\begin{proposition} \label{prop:fluctuations-stdgauss}
Under the couplings of Section~\ref{sec:coupling-list}, it holds that
\begin{equation*}
    \lim_{d \to \infty}\sup_{w \in \mathcal{S}} \mathbb{E} \left[ d^2\lVert W^{(d)}_{(t+1)/d} - W^{(d)}_{t/d} \rVert^2 \bigm\lvert W^{(d)}_{t/d} = w \right] < \infty.
\end{equation*}
\end{proposition}

Having obtained the limiting drift of the process $W^{(d)}$, as well as bounded its fluctuations, we can state our main result: the convergence of this process to a deterministic limit.

\begin{theorem} \label{thm:ode-limit-stdgauss} 
Let $W^{(d)}_0 = w_0 \in \bar{\mathcal{S}}$ for all $d$. Then, under the couplings of Section~\ref{sec:coupling-list}, it holds that
\begin{equation*}
	W^{(d)} \implies w \quad \text{as}\quad d \to \infty,
\end{equation*}
where $w:[0, \infty) \to \bar{\mathcal{S}}$ is the solution of the initial value problem
\begin{equation}\label{eqn:ode-limit-stdgauss}
    \mathrm{d}w(t) = c_\ell(w(t))\mathrm{d}t \quad \text{started from } w(0) = w_0,
\end{equation}
and where the drift $c_\ell(\cdot)$ is coupling-specific as in Proposition~\ref{prop:drift-stdgauss}.
\end{theorem}

The analysis of the process $(X_t,Y_t)_{t\ge 0}$ therefore reduces to the analysis of the ODE. We are interested in the squared distance $\lVert X_t - Y_t \rVert^2$: a change of variables leads us to an analogous ODE limit $\bar W^{(d)} \implies \bar w$ for the process 
\begin{equation*}
\left(\bar W^{(d)}_t\right)_{t \ge 0} =\frac{1}{d} \left(\lVert X_{\lfloor td \rfloor}\rVert^2, \lVert Y_{\lfloor td \rfloor}\rVert^2, \lVert X_{\lfloor td \rfloor} - Y_{\lfloor td \rfloor} \rVert^2 \right)_{t \ge 0}
\end{equation*}
to the solution $\bar w = (x,y,s)$ of $\mathrm{d}\bar w(t) = \bar c_\ell(\bar w(t))\mathrm{d}t,$ where $\bar c_\ell(\bar w) = (a_\ell(x), a_\ell(y), \bar b_\ell(x,y,s))$ and $\bar{b}_\ell(x,y,s) = a_\ell(x) + a_\ell(y) - 2b_\ell(x,y,(x+y-s)/2)$. ($a_\ell(\cdot)$ and $b_\ell(\cdot)$ are defined in Proposition~\ref{prop:drift-stdgauss}.) 

The intuition from this result is that after solving the ODE we obtain a function $s(t)$, the solution for the third component, which contracts to $0$ as $t\to \infty$ for the GCRN and reflection couplings but does not contract to $0$ for the CRN coupling, and is such that
\begin{equation*}
\|X_t - Y_t\|^2 \approx s\left( t/d \right) d \quad \text{for large~$d$}.
\end{equation*}

Over one step, the squared-distance changes by roughly $\bar b_\ell(\cdot) / d$; the smaller this value, the more a coupling contracts the chains. Since the contraction efficiency measure $\ejc$ appears in the limiting drift $\bar b_\ell(\cdot)$, by Theorem~\ref{thm:gcrn-opt-stdgauss} the GCRN coupling optimizes this drift point-wise over all couplings in $\mathcal{P}$.

\subsubsection{Numerical illustration} \label{sec:spherical-numeric}

\begin{figure}[htb]
    \centering
    \includegraphics[width=\textwidth]{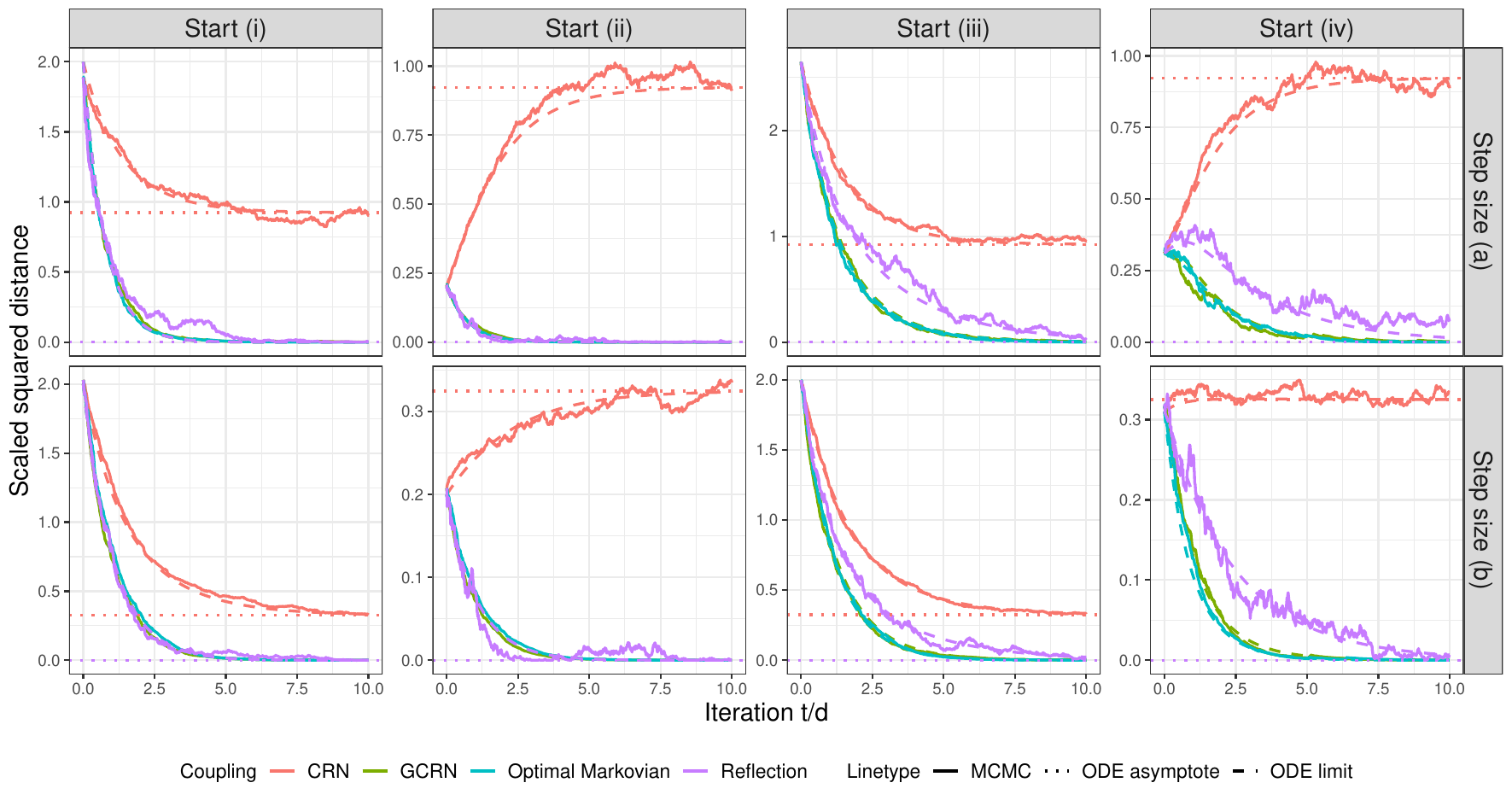}
    \caption{Trace of the scaled squared distance $\lVert X_{t} - Y_{t}\rVert^2/d$ and its ODE limit, for a target $\pi^{(d)} = \mathcal{N}_d(0_d, I_d)$ and various couplings, step sizes, and starting conditions as in Section~\ref{sec:spherical-numeric}.}
    \label{fig:spherical-squaredist}
\end{figure}

We visualize the ODE limits in Figure~\ref{fig:spherical-squaredist}, considering four different starting conditions $(x_0,y_0,\rho_0)$, where $(x_0, y_0)$ correspond to $(\lVert X_0\rVert^2/d, \lVert Y_0\rVert^2/d)$ and $\rho_0$ corresponds to the cosine similarity $X_0^\top Y_0 /(\lVert X_0\rVert \lVert Y_0\rVert)$ between the initial states of the chains. The starting conditions are, and in a limiting sense correspond to chains initialized from:
\begin{enumerate}
    \item[(i)] $(x_0,y_0,\rho_0) = (1,1,0)$, independent draws from the target.
    \item[(ii)] $(x_0,y_0,\rho_0) = (1,1,0.9)$, positively correlated draws from the target.
    \item[(iii)] $(x_0,y_0,\rho_0) = (1.5,0.5,0)$, independent draws from, respectively, over- and under-dispersed versions of the target.
    \item[(iv)] $(x_0,y_0,\rho_0) = (0.4,0.01,-0.5)$, negatively correlated draws from under-dispersed versions of the target.
\end{enumerate}
We also compare ODE limits with MCMC traces in dimension $d = 1{,}000$. Within each scenario, the same starting value $(X_0,Y_0)$ is used for all couplings; its coordinates are sampled independently from $\bvn(x_0, y_0; (x_0 y_0)^{1/2}\rho_0)$, the bivariate normal with coordinate-wise variances $(x_0, y_0)$ and correlation $\rho_0$. We use step sizes (a) $\ell= 2.38$ and (b) $\ell = \sqrt{2}$; the former is optimal at stationarity \citep{roberts1997weak} and both are close to optimal in the transient phase (see \citealp{christensen2005scaling} and Section~\ref{sec:dim-scaling-stdgauss}).

Figure~\ref{fig:spherical-squaredist} confirms that our theory consistently predicts the behaviour of the coupled RWM algorithms in high dimensions $d$. Although stochasticity is clearly present in the traces, this noise will vanish in the limit as $d\to\infty$. The GCRN coupling outperforms the CRN and reflection couplings. Furthermore, GCRN very closely approximates the asymptotically optimal Markovian coupling of Appendix~\ref{app:optimal-markovian}; for this reason, we choose to focus on GCRN, although we will revisit the optimal Markovian coupling in the numerical illustrations of Section~\ref{sec:elliptgauss}.

\subsubsection{Long-time behaviour} \label{sec:fixedpoint-stdgauss}

We turn to the behaviour the chains in the joint long-time and high-dimensional limits. We shall access this through the unique stable fixed point of the limiting ODE~\eqref{eqn:ode-limit-stdgauss}; all fixed points of this process are characterized in Proposition~\ref{prop:fixedpoint-stdgauss} below.

The first two coordinates of the ODE dictate the marginal behaviour of the chains and are autonomous. Their fixed points correspond to the chains being marginally stationary, are stable, and are $(x^*,y^*) = (1,1)$. The third coordinate is more involved: we require the function 
$$h_\ell(\rho) = \mathbb{E}_{(Z_1,Z_2) \sim \bvn(\rho)} \left[1\land e^{\ell Z_1 - \ell^2/2}\land e^{\ell Z_2 - \ell^2/2}\right],$$ which is increasing on its domain $[-1,1]$, is bounded above by $h_\ell(1) = 2 \Phi(-\ell/2)$, and has unbounded derivative as $\rho \to 1$. (See in Lemma~\ref{lemma:h-of-rho} in Appendix~\ref{app:proofs-aux} for properties of $h_\ell$.) Proposition~\ref{prop:drift-stdgauss} indicates that the fixed points $v^*$ are the solutions of $h_\ell(\rho(v)) - v h_\ell(1) = 0$, where the correlation $\rho(\cdot)$ is coupling-specific: $\rho_\textup{crn}(v) = v$ and $\rho_\textup{refl}(v) = \rho_\textup{gcrn}(v) = 1.$ The terms of the fixed-point equation have straightforward interpretations: $v$ is the cosine similarity between the coupled states, $h_\ell(1)= 2 \Phi(-\ell/2)$ is the acceptance rate of the marginal chains, and $h_\ell(\rho(v))$ is the rate at which the chains accept their proposals simultaneously.

A scalable coupling must ensure the stability of the fixed point $v^* = 1$. However, due to the rapid growth of $h_\ell(\rho)$ near $\rho = 1$, the fixed point $v^* = 1$ is highly sensitive to the function $\rho(\cdot)$: stability can essentially only be achieved if $\rho(v) = 1$ in an interval around $v=1$. As we discuss in Section~\ref{sec:sync-accept}, this instability is caused by the distance between the coupled chains increasing by a relatively large amount when one chain accepts its proposal while the other rejects its proposal, and points to the need to use couplings like GCRN which attempt to synchronize acceptance events. As it happens, a spherically symmetric target also allows the reflection coupling to synchronize acceptance events: intuitively, $\rho_\textup{refl}(v) = 1$ here because, when the chains are on the same level set, the reflection is a perfect mapping between the respective logarithmic gradients.

\begin{proposition}\label{prop:fixedpoint-stdgauss} Under the couplings of Section~\ref{sec:coupling-list}, the fixed points of \eqref{eqn:ode-limit-stdgauss} are of the form $w^* = (1,1,v^*)$, where $v^*$ is coupling-specific:
\begin{itemize}
    \item \textup{\textbf{CRN}}: $v_\textup{crn}^* \in (0,1)$, stable and $v_u^*=1$, unstable.
    \item \textup{\textbf{Reflection}}: $v_\textup{refl}^* = 1$, stable.
    \item \textup{\textbf{GCRN}}: $v_\textup{gcrn}^* = 1$, stable.
\end{itemize}
\end{proposition}

We are interested in the limiting squared distance. Proposition~\ref{prop:fixedpoint-stdgauss} suggests that this equals, for the coupling-specific stable value $v^*_\textup{coup}$,
\begin{equation*}
	\lim_{d,t\to\infty}\lVert X_t - Y_t\rVert^2 / d =: s^*_\textup{coup} = 2(1 - v^*_\textup{coup}).
\end{equation*}
For CRN, we plot $s^*_\textup{crn}(\ell)$ in Figure~\ref{fig:asymptote-elliptgauss} below and conclude that $\lVert X_t - Y_t \rVert^2 = \Theta(d)$ for large~$t$ unless $\ell = o_d(1)$. As the dimension increases, the CRN coupling becomes increasingly impractical: to achieve sufficient contraction, it must sacrifice mixing by using an increasingly smaller scaling.

For GCRN and reflection, since $s^*_\textup{gcrn} = s_\textup{refl}^* = 0$ we conclude that $\lVert X_t - Y_t \rVert^2 = o(d)$ for large~$t$. 
However, numerically testing dimensions $d \in [1, 10000]$, we found that the contraction was in fact significantly better. (Recall that, by~\eqref{eqn:prob-meeting}, contraction to $\bigO(h^2)$ is required for a reasonable chance of coalescence.) The GCRN coupling exhibited clear $\bigO(h^2)$ behaviour in all dimensions; remarkably, for $d \ge 4$ the coupling was always able to contract the chains to within numerical precision. The reflection coupling displayed $\bigO(1)$ behaviour, with plots suggesting that it could bring the chains sufficiently close together for a reasonable chance of coalescence when $d \le 100$, but that coalescence would be unlikely in dimensions at least an order of magnitude larger.

\begin{figure}
    \centering
    \includegraphics[width = \textwidth]{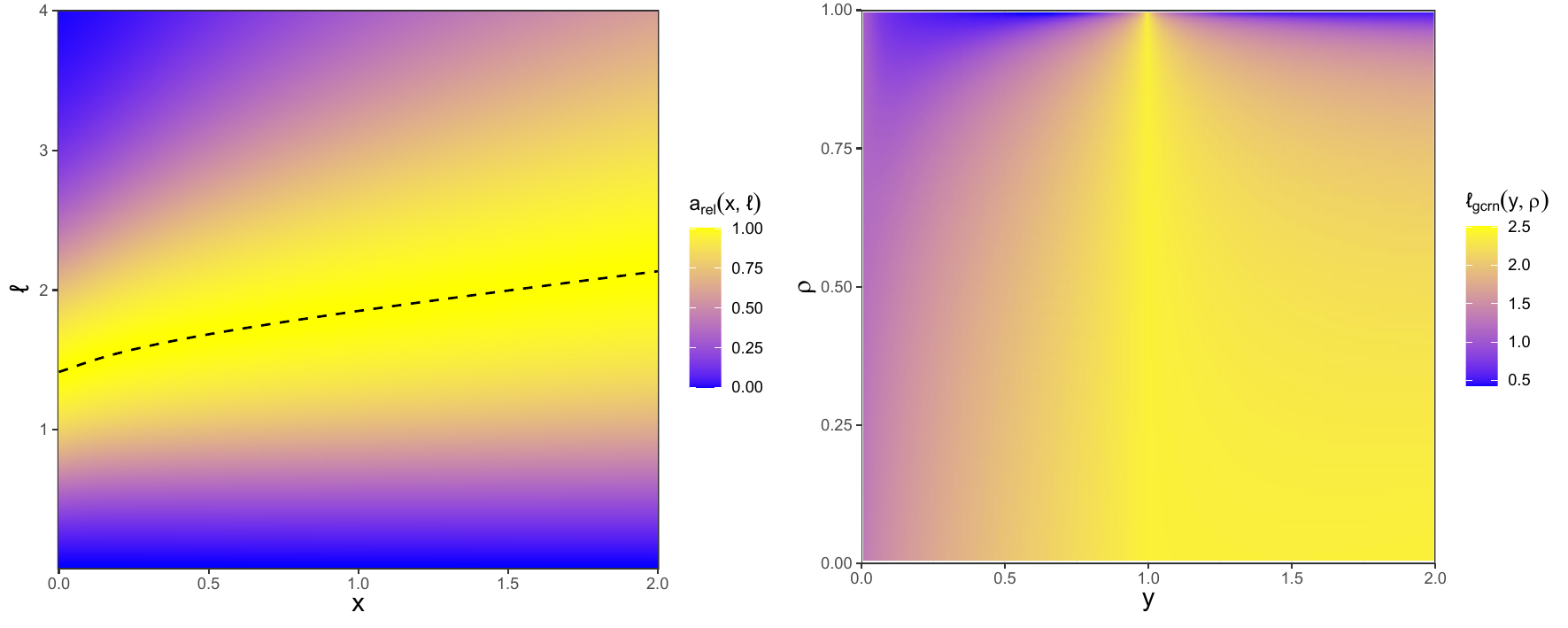}
    \caption{Optimal step size scalings for marginal rate of convergence and for contraction, as in Section~\ref{sec:dim-scaling-stdgauss}. \textbf{Left:} Heatmap of relative drift $a_\textup{rel}(x,\ell) = |a_\ell(x)| / \max_{\ell}\{|a_\ell(x)|\}$; the dashed line traces the optimum point-wise over $x$. \textbf{Right:} Heatmap of optimal step size $\ell_\textup{gcrn}(y,\rho)$ for the GCRN coupling, fixing the $X$-chain stationary with $x=1$.}
    \label{fig:scaling-spherical}
\end{figure}

\subsubsection{Optimal scaling} \label{sec:dim-scaling-stdgauss}

To optimize the rate of convergence of a single RWM chain towards the main target mass, the step size $\ell$ should maximize the absolute drift $|a_\ell(x)|$ point-wise over $x$. We plot this quantity normalized to unit scale in Figure~\ref{fig:scaling-spherical} (left). Echoing \cite{christensen2005scaling}, the speed of convergence is insensitive to the choice of step size and no single step size is uniformly optimal, though the rough trend is that smaller step sizes should be chosen for convergence than for mixing. We single out two step sizes: (a) $\ell = 2.38$, which is optimal for mixing at stationarity, achieves at least 45\% efficiency over the considered range; (b) $\ell = \sqrt{2}$ is optimal for convergence when the chain is at the mode and is at least 86\% efficient over the considered range. The acceptance rate is remarkably stable at the point-wise optimal step size (the dashed line), and hovers around $35\%$.

We next turn to the problem of selecting a step size $\ell$ which optimizes the contraction of the GCRN coupling. In the ODE limit, we should optimize the drift $\bar{b}_\ell(x,y,s)$ point-wise. To obtain sensible guidelines, we fix $x = 1$: this corresponds to a stationary chain coupled with a non-stationary chain and emulates a lagged coupling with a large lag parameter. We reparametrize the state to $(x,y,s) = (1,y, 1 + y - 2y^{1/2}\rho)$, where $\rho \in [-1,1]$ denotes cosine similarity, and we plot the optimal step size $\ell_\textup{gcrn}$ as a function of $(y,\rho)$ in Figure~\ref{fig:scaling-spherical} (right).\footnote{The heatmap for $\rho \in [-1,0)$ is not pictured, as it resembles the lower half of the plot.} Remarkably, the step size $\ell = 2.38$ is close to optimal over much of the range, and in particular when the coupled chains are marginally stationary. We also find (not pictured) that the contraction is insensitive to varying the step size near the optimum. This suggests close to optimal performance for GCRN when the step size is tuned to $\ell = 2.38$ and the acceptance rate to $23.4\%$, though we caution that in practice smaller scalings and higher acceptance rates may be required because it is crucial for a coupling to synchronize acceptance events between the coupled chains (see Section~\ref{sec:sync-accept}). Similar guidelines apply to the reflection coupling, but as we shall see in Section~\ref{sec:elliptgauss} this coupling requires a much smaller scaling $\ell = o_d(1)$ when the target is not spherically symmetric.

The different optimal scalings may seem at odds; the discrepancy can be explained as follows. The rate of contraction of two coupled chains depends on their movement in both the angular and the radial directions. It therefore benefits from better mixing in the angular direction, which for a single stationary chain is optimized by the scaling $\ell = 2.38$. In contrast, the rate of convergence of a single chain towards the main target mass only depends on the movement in the radial direction, which benefits from a smaller scaling.

\section{Analysis: elliptical Gaussian case} \label{sec:elliptgauss}

While we do not see the asymptotic Gaussianity assumption as a major limitation to our theoretical results, the spherical symmetry assumed in Section~\ref{sec:stdgauss} is certainly unrealistic: in practice, even after preconditioning, one cannot expect the RWM proposals to precisely match the structure of the target. In this section, we therefore relax this constraint and consider a more general sequence of targets $\pi^{(d)} = \mathcal{N}_d(0_d,\Sigma_d)$ of increasing dimension $d$. Given the form of the acceptance ratio, it will be convenient to work with the precision matrix $\Omega_d = \Sigma_d^{-1}$. To obtain a transparent asymptotic theory, we fix a positive-valued distribution $\mu$ which captures heterogeneity across the lengthscales of the target and assume that 
\begin{equation}\label{eqn:elliptgauss-prod-assumption}
    \forall d: \quad \Omega_d = \diag(\omega_1^2,\dots,\omega_d^2), \textup{ where } (\omega_i^2)_{i \ge 1} \iid \mu.
\end{equation}

We impose moment conditions on the \emph{spectral distribution} $\mu$ of the precision matrix, see Assumptions~\ref{assumption:ellipt-1} and~\ref{assumption:ellipt-2}. We first show in Theorem~\ref{thm:gcrn-opt-elliptgauss} that the GCRN coupling is asymptotically optimal for contraction within the class $\mathcal{P}$. By distilling the argument of Section~\ref{sec:stdgauss} down to its key part, we then consider a form of limiting drift in Proposition~\ref{prop:drift-elliptgauss}, which tells us about the long-time behaviour of the coupled chains and informs efficient step size scalings. To preempt our main conclusions, we find that the GCRN coupling is robust to the eccentricity of the target and that the same acceptance rate tuning guidelines apply as in the spherical case. The reflection coupling however breaks down for eccentric targets and must sacrifice mixing in order to achieve sufficient contraction. Under stronger structural assumptions, we could arrive at ODE limits, see Remark~\ref{remark:ode-ellipt}.

It will be convenient to define (for all $x,y\in\mathbb{R}^d$ and $k \in \mathbb{Z}$) the inner product $\langle x,y \rangle_{[k]} = x^\top \Omega_d^k y$ and its associated norm $\lVert x \rVert^2_{[k]} = \langle x,x \rangle_{[k]}$. The acceptance step becomes
\begin{equation*}
    B_x = \mathbbm{1}{\left\{\log U_x \le  -h \langle Z_x,X_t \rangle_{[1]} - (h^2/2) \lVert Z_x \rVert^2_{[1]} \right\}}.
\end{equation*}

\subsection{Asymptotic optimality} \label{sec:asympt-opt-ellipt}

We first show that the GCRN coupling is asymptotically optimal for contraction among the class $\mathcal{P}$ of all product couplings, for which we require a law of large numbers assumption.

\begin{assumption} \label{assumption:ellipt-1}
The spectral distribution $\mu$ has finite first moment $z_1= \mathbb{E}[\omega_i^{2}]$.
\end{assumption}

\begin{theorem}\label{thm:gcrn-opt-elliptgauss}
Under Assumption~\ref{assumption:ellipt-1} and conditionally on $(\lVert X_t \rVert^2_{[2]}, \lVert Y_t \rVert^2_{[2]},\langle X_{t},Y_{t} \rangle_{[2]})/(z_1 d) = (x_2,y_2,v_2)$,
it holds that
    \begin{equation*}
        \adjustlimits\lim_{d\to\infty}\sup_{\bar{K} \in \mathcal{P}} \mathbb{E}\left[ h^2 Z_x^{\top} Z_y B_x B_y \right] = \ell^2\mathbb{E}_{Z\sim \mathcal{N}(0,1)} \left[1\land e^{\lambda x_2^{1/2}Z - \lambda^2/2}\land e^{\lambda y_2^{1/2}Z - \lambda^2/2}\right],
    \end{equation*}
where $\lambda = \ell z_1^{1/2}$. Furthermore, the limit supremum is attained by the GCRN coupling.
\end{theorem}

Roughly speaking, Assumption~\ref{assumption:ellipt-1} states that none of the lengthscales of the target are much smaller than the average. Regularity conditions like these are among the weakest required by optimal scaling theory \citep[Condition~1]{sherlock2013optimal} and appear necessary to avoid degeneracies in the efficiency of the RWM algorithm \citep{beskos2018ridged}. Under Assumption~\ref{assumption:ellipt-1}, one should think of 
\begin{equation*}
    \lambda = \ell z_1^{1/2}
\end{equation*}
as the natural step size parameter, with $\lambda = 2.38$ being optimal at stationarity and corresponding to a $23.4\%$ acceptance rate. (We further relate our notation to the optimal scaling literature in Section~\ref{sec:prod-target}.) The key consequence of Assumption~\ref{assumption:ellipt-1} is that the Hessian term in the acceptance ratio tends to a constant: $\lim_{d \to \infty} h^2\lVert Z_x \rVert^2_{[1]} = \ell^2 z_1$ in probability. It is this which enables us to prove Theorem~\ref{thm:gcrn-opt-elliptgauss}. The quantities in Theorem~\ref{thm:gcrn-opt-elliptgauss} are rescaled so that e.g. $x_2 \approx 1$ when the chain is in the main target mass; we will apply a similar scaling to $\lVert X_t \rVert^2_{[k]}$ in the sequel.

\subsection{Scaling limits}

We turn to the problem of characterizing the high-dimensional behaviour of the coupled RWM chains. To be able to handle all couplings of Section~\ref{sec:coupling-list} at once, we impose somewhat stronger assumptions.

\begin{assumption} \label{assumption:ellipt-2}
The spectral distribution $\mu$ has finite $k$-th moment for $k \in \{-2,1\}$.
\end{assumption}

Assumption~\ref{assumption:ellipt-2} guarantees the existence of all intermediate moments $z_{k} = \mathbb{E}[\omega_i^{2k}]$ for all $k \in [-2,1]$. Intuitively, Assumption~\ref{assumption:ellipt-2} asks for none of the lengthscales of the target to be much larger or smaller than the average. In practice, the target is fixed and not randomly generated as in~\eqref{eqn:elliptgauss-prod-assumption}, but it may still arise as a discretization of a limiting target with a given limiting spectral distribution. We exemplify a class of time series models whose limiting spectral distribution satisfies the assumption.

\begin{example}\label{example:ar-p}
Fix $p \in \mathbb{N}$. For all $d \ge p$, let $\pi^{(d)}$ be the distribution of a stationary $\textup{AR}(p)$ process of length $d$, with the same parameters across all dimensions $d$. Then, the spectrum of the precision matrix converges to a spectral distribution $\mu$ which satisfies Assumption~\ref{assumption:ellipt-2}.
\end{example}

In Example~\ref{example:ar-p}, the assumption is satisfied because the limiting target is well-conditioned, however we stress that a uniform control of the condition number is certainly not necessary. More broadly, we expect Assumption~\ref{assumption:ellipt-2} to hold for Gaussian Markov random fields \citep{rue-held2005gmrf-book} whose dependence structure grows suitably slowly with the dimension $d$.

The analysis of the reflection coupling will require the following quantity, which we dub the \emph{limiting eccentricity} of the target
\begin{equation*}
    \ecc = z_1 z_{-1} = \lim_{d \to \infty}\tr(\Omega_d)\tr(\Sigma_d)/d^2.
\end{equation*}
This quantity satisfies $\ecc \ge 1$ and is invariant to a rescaling of the precision matrix by a constant factor. The lower bound is attained when the target is spherical; the greater the imbalance between the average lengthscales of the precision and covariance matrices, the larger $\ecc$ becomes.

\subsubsection{Limiting drift}

We now compute the limiting drift of triplets $W_{[k]} = (\lVert X_t \rVert^2_{[k]}, \lVert Y_t \rVert^2_{[k]},\langle X_{t},Y_{t} \rangle_{[k]})$ for various $k$. For this calculation, we let $\bar{\mathbb{E}}$ denote the expectation conditional on $W_{[j]} / (z_{j-1} d) = (x_j,y_j,v_j)$ for all $j \in \{-1,0,1,2\}$, where the normalizing constants are defined subsequent to Assumption~\ref{assumption:ellipt-2}.

\begin{proposition}\label{prop:drift-elliptgauss} 
Under Assumption~\ref{assumption:ellipt-2} and the couplings of Section~\ref{sec:coupling-list}, for all $k\in\{-1,0,1\}$ it holds that
\begin{equation} \label{eqn:infi-ellipt}
    \begin{aligned}
    &\lim_{d \to \infty}\bar{\mathbb{E}} \left[ \lVert X_{t+1} \rVert_{[k]}^2 - \lVert X_t \rVert_{[k]}^2 \right] = \ell_k^2 \alpha (x_{k+1};x_2), \\
    &\lim_{d \to \infty}\bar{\mathbb{E}}\left[ \lVert Y_{t+1} \rVert_{[k]}^2 - \lVert Y_t \rVert_{[k]}^2 \right] = \ell_k^2 \alpha(y_{k+1};y_2),\\
    &\lim_{d \to \infty}\bar{\mathbb{E}} \left[ \langle X_{t+1}, Y_{t+1}\rangle_{[k]} - \langle X_{t}, Y_{t}\rangle_{[k]} \right] = \ell_k^2 \beta(v_{k+1};x_{2},y_{2},\rho),
    \end{aligned}
\end{equation}
where $\ell_k = \ell z_k^{1/2}$,
\begin{align*}
    \alpha(x_{k+1};x_2) ={}& (1-2x_{k+1}) e^{\lambda^2(x_2-1)/2} \Phi\left(\frac{\lambda}{2x_2^{1/2}} - \lambda x_2^{1/2}\right) + \Phi\left(-\frac{\lambda}{2x_2^{1/2}}\right),\\
\begin{split}
    \beta(v_{k+1};x_{2},y_{2},\rho) ={}& \mathbb{E}_{(Z_1,Z_2)\sim \bvn(\rho)}\left[1\land e^{\lambda x_1^{1/2} Z_1 - \lambda^2/2}\land e^{\lambda x_2^{1/2} Z_2 - \lambda^2/2}\right] \\
    &- v_{k+1}\left[ e^{\lambda^2(x_2-1)/2} \Phi\left(\frac{\lambda}{2x_2^{1/2}} -\lambda x_2^{1/2}\right) + e^{\lambda^2(y_2-1)/2} \Phi\left(\frac{\lambda}{2y_2^{1/2}} -\lambda y_2^{1/2}\right) \right],
\end{split}
\end{align*}
where $\lambda = \ell z_1^{1/2}$, and where $\rho$ is coupling-specific:
\begin{equation*}
    \rho_\textup{crn} = \frac{v_2}{(x_2y_2)^{1/2}}, \quad \rho_\textup{refl} = \frac{v_2}{(x_2y_2)^{1/2}} + \frac{2(x_1 - v_1)(y_1 - v_1)}{\ecc (x_2y_2)^{1/2} (x_0 + y_0 - 2 v_0)}, \quad \rho_\textup{gcrn} = 1.
\end{equation*}
\end{proposition}

Proposition~\ref{prop:drift-elliptgauss} generalizes Proposition~\ref{prop:drift-stdgauss} to the elliptical case, with the notable differences that the natural step size parameter is now $\lambda$ and that the correlation $\rho_{\textup{refl}}$ now depends on the eccentricity $\ecc$. By Theorem~\ref{thm:gcrn-opt-elliptgauss}, the GCRN coupling optimizes point-wise the drift~\eqref{eqn:infi-ellipt} over all couplings in $\mathcal{P}$. With more care, we could additionally bound the fluctuations of $W_{[k]}$ (as in Proposition~\ref{prop:fluctuations-stdgauss}) and we could make our results uniform over compact sets. However, these refinements will not provide any additional insight into the asymptotic behaviour of the coupled chains, so we avoid further technicalities.

Through a change of variables, we can understand the evolution of the squared Mahalanobis distance,
\begin{equation*}
    \lim_{d \to \infty}\bar{\mathbb{E}} \left[ \lVert X_{t+1} - Y_{t+1}\rVert^2_{[k]} - \lVert X_{t} - Y_{t}\rVert^2_{[k]} \right] = \ell_k^2 \gamma(x_{k+1}, y_{k+1}, v_{k+1};x_{2},y_{2},\rho),
\end{equation*}
where $\gamma(x_{k+1}, y_{k+1}, v_{k+1};x_{2},y_{2},\rho) :=  \alpha (x_{k+1};x_2) +  \alpha (y_{k+1};y_2) - 2\beta(v_{k+1};x_{2},y_{2},\rho).$

\begin{remark} \label{remark:ode-ellipt}
Under certain structural assumptions, we can extend Proposition~\ref{prop:drift-elliptgauss} to an ODE limit. For instance, if $\Sigma = \diag(1,\sigma^2,1,\sigma^2,\dots)$, by decomposing the process $W^{(d)}$ considered in Section~\ref{sec:stdgauss} into separate triplets for the odd and even coordinates, we obtain a six-dimensional Markov process. Under moment conditions on $\sigma^2$, this Markov process has an ODE limit as the dimension grows.
\end{remark}

\subsubsection{Long-time behaviour} \label{sec:fixedpoint-elliptgauss}

\begin{figure}[tb]
    \centering
    \includegraphics[width = 0.7\textwidth]{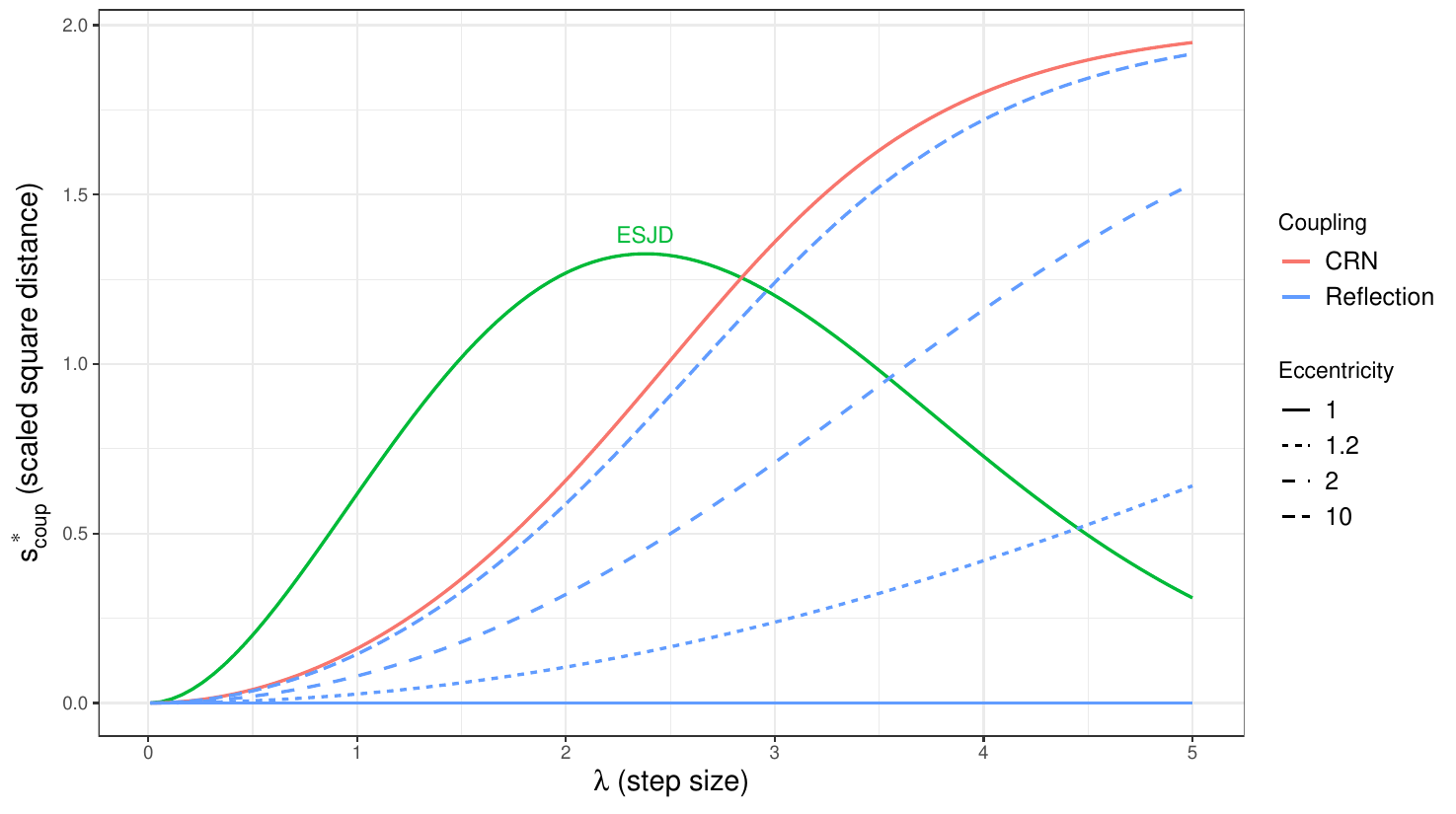}
    \caption{Scaled squared distance $s_\textup{coup}^*$ in the joint long-time and high-dimensional limits, as in Section~\ref{sec:fixedpoint-elliptgauss}, for targets $\pi^{(d)} = \mathcal{N}_d(0_d, \Sigma_d)$ and various couplings and values of the natural step size $\lambda$ and eccentricity $\ecc$. We stress that the only desirable value is $s^*_\textup{coup} = 0$. The efficiency measure $\esjd(\lambda)$ is overlaid for context.}
    \label{fig:asymptote-elliptgauss}
\end{figure}

We infer the behaviour of the coupled chains in the joint long-time and high-dimensional limits from the stable fixed points of the system of equations~\eqref{eqn:infi-ellipt}. This is sensible, as the unconditional expected one-step change in a function of the coupled chains is null when coupled process is stationary.

This analysis mirrors, and assumes familiarity with, the discussion of Section~\ref{sec:fixedpoint-stdgauss}. The fixed points are the solutions of $\alpha(x_{k};x_2) = \alpha(y_{k};y_2) = \beta(v_{k};1,1,\rho) = 0$ for all $k \in\{0,1,2\}$. The solutions $x_k^* = y_k^* = 1$ correspond to marginal stationarity in each chain. This implies that the solutions $v_{k}^* = v^*$ must be constant in $k$, which collapses the fixed-point equations for the remaining joint behaviour of the chains to
\begin{equation} \label{eqn:fixedpoint-elliptgauss}
	h_\lambda(\rho(v)) - vh_\lambda(1) = 0,
\end{equation}
where the correlation $\rho(\cdot)$ is coupling-specific: $\rho_\textup{crn}(v) = v,$ $\rho_\textup{refl}(v) = v + \ecc^{-1} (1-v)$ and $\rho_\textup{gcrn}(v) = 1.$ The terms of the fixed-point equation~\eqref{eqn:fixedpoint-elliptgauss} retain their intuitive interpretations from Section~\ref{sec:fixedpoint-stdgauss} in terms of the cosine similarity $v$, the marginal acceptance rate $h_\lambda(1)$, and the synchronous acceptance rate $h_\lambda(\rho(v))$. We consider a linear stability analysis in Proposition~\ref{prop:fixedpoint-elliptgauss}; for the fixed point $v^* = 1$ to be stable, we essentially require that $\rho(v) = 1$ in an interval around $v = 1$. In other words, we require marginally stationary chains to accept their proposals simultaneously even when they are not coalesced. The reflection coupling is unable to do so in the elliptical case, which ultimately impacts its scalability.

\begin{proposition} \label{prop:fixedpoint-elliptgauss}
Let the target be non-spherical Gaussian (i.e. $\ecc > 1$). Under the couplings of Section~\ref{sec:coupling-list}, the solutions of the fixed-point equation~\eqref{eqn:fixedpoint-elliptgauss} are as follows:
\begin{itemize}
    \item \textup{\textbf{CRN}}: $v_\textup{crn}^* \in (0,1)$, stable and $v_u^* = 1$, unstable. 
    \item \textup{\textbf{Reflection}}: $v_\textup{refl}^* \in (v_\textup{crn}^*,1)$, stable and $v_u^* = 1$, unstable. As a function of $\ecc \in (1, \infty)$, $v^*_\textup{refl}(\ecc)$ is decreasing and tends to $\{1, v_\textup{crn}^*\}$ at the extremes.
    \item \textup{\textbf{GCRN}}: $v_\textup{gcrn}^* =1$, stable.
\end{itemize}
\end{proposition}

We are interested in the limiting behaviour of the squared distance. Proposition~\ref{prop:fixedpoint-elliptgauss} suggests that, for the coupling-specific stable value $v^*_\textup{coup}$,
\begin{equation*}
\lim_{d, t \to \infty}  \lVert X_{t} - Y_{t}\rVert^2 / \tr(\Sigma_d) =: s^*_\textup{coup} = 2(1 - v^*_\textup{coup}).
\end{equation*}
We plot this limiting quantity in Figure~\ref{fig:asymptote-elliptgauss}. For reflection and CRN, since $s^*_\textup{refl},s_\textup{crn}^* > 0$ we conclude that $\lVert X_t - Y_t \rVert^2 = \Theta(d)$ for large~$t$. Both couplings become increasingly impractical as the dimension increases: they must sacrifice mixing for contraction by using an increasingly smaller scaling $\lambda = o_d(1)$. This points to the need to use preconditioning alongside the reflection coupling, as this coupling performs adequately for spherical targets. The CRN coupling is unaffected by the eccentricity of the target but is uniformly worse than the reflection coupling.

For GCRN, since $s^*_\textup{gcrn} = 0$ we conclude that $\lVert X_t - Y_t \rVert^2 = o(d)$ for large~$t$. However, experimentally (we tested dimensions $d \in [1, 10000]$ and targets similar to those of Figure~\ref{fig:elliptical-squaredist}) we found that the GCRN coupling performed significantly better. We observed $\bigO(h^2)$ behaviour in all dimensions; for $d \ge 10$, the GCRN coupling was consistently able to contract the chains to within numerical precision, even for the most eccentric target considered.

In practice, given estimates of the traces of the covariance and precision matrices, we can estimate the eccentricity $\ecc$, compute $v^*_{\textup{coup}}$ by solving~\eqref{eqn:fixedpoint-elliptgauss} numerically, and then predict
\begin{equation}\label{eqn:predict-asymptote-elliptgauss}
	\mathbb{E}\left[ \lVert X_t - Y_t \rVert^2 \right] \approx 2\tr (\Sigma_d)(1 - v^*_{\textup{coup}}) \quad \text{for large } d, t.
\end{equation}

\subsubsection{Optimal scaling} \label{sec:scaling-elliptgauss}

We turn to the problem of step size tuning for the GCRN coupling. When $k=1$, by comparison with Proposition~\ref{prop:drift-stdgauss}, we obtain that
\begin{equation*}
\lim_{d \to \infty}\bar{\mathbb{E}} \left[ \lVert X_{t+1} - Y_{t+1}\rVert_{[1]}^2 - \lVert X_{t} - Y_{t}\rVert_{[1]}^2 \right]  = \bar{b}_{\lambda}(x_2,y_2,s_2),
\end{equation*}
where $\bar{b}_{\lambda}(\cdot)$ is the drift of the squared-distance process which we obtained for a spherical Gaussian target, which now depends on the natural step size parameter $\lambda$. To optimize the contraction of the GCRN coupling, we should therefore optimize the drift $\bar{b}_{\lambda}(\cdot)$ point-wise. This problem was considered in Section~\ref{sec:dim-scaling-stdgauss}: we reiterate that (i) the contraction of the chains is insensitive to the scaling, (ii) the scaling $\lambda = 2.38$ and the acceptance rate of $23.4\%$ are close to optimal for an elliptical Gaussian target and (iii) a somewhat smaller step size may be necessary in practice, as it is crucial for the coupling to ensure that the chains accept their proposals at the same time, see Section~\ref{sec:sync-accept}.

\subsection{Numerical illustration} \label{sec:elliptical-numerics}

\begin{figure}[tb]
    \centering
    \includegraphics[width=\textwidth]{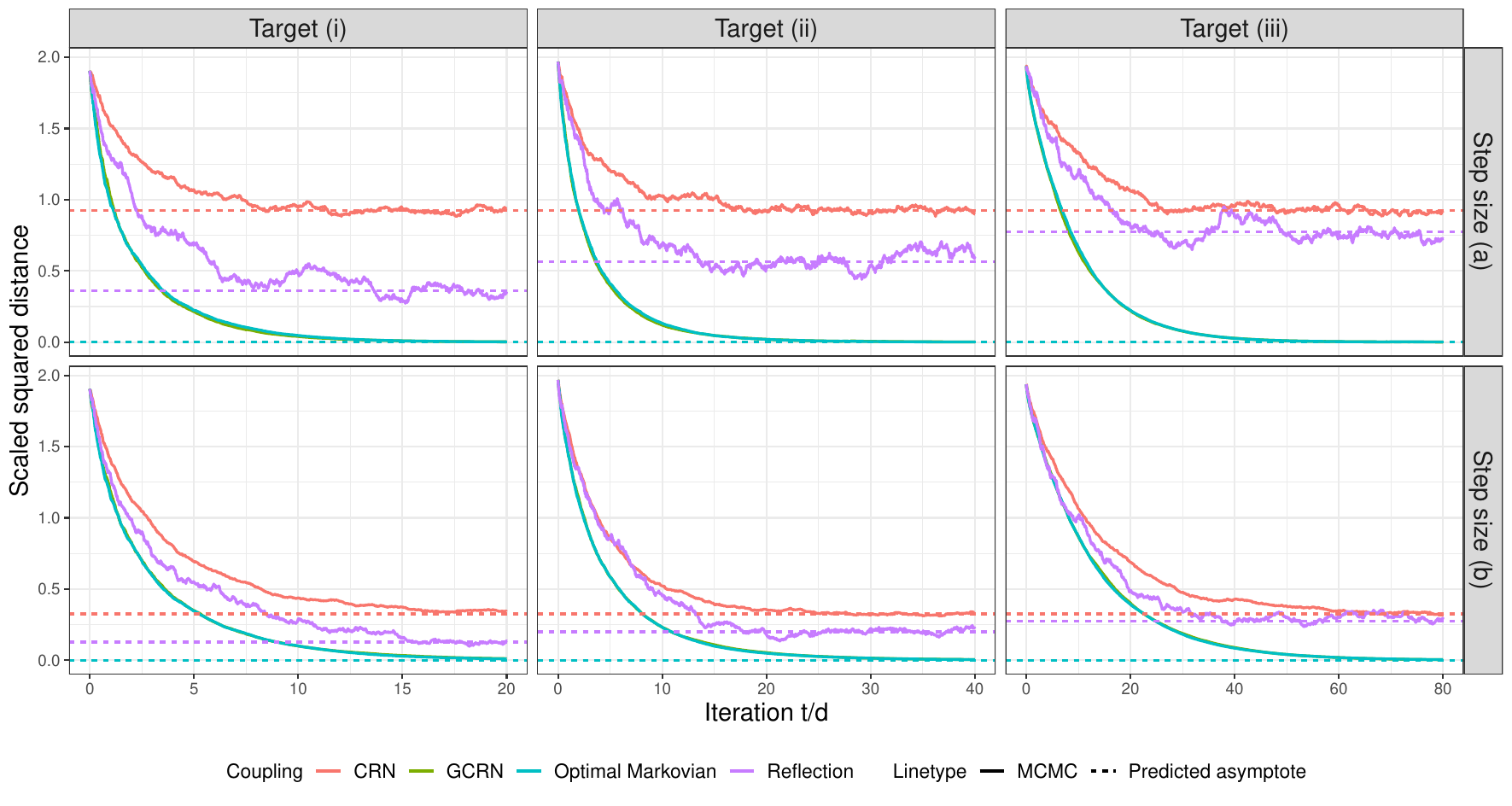}
    \caption{Trace of the scaled squared distance $\lVert X_t - Y_t \rVert^2 / \tr(\Sigma_d)$ and predicted long-time asymptote, for various targets $\pi^{(d)} = \mathcal{N}_d(0_d, \Sigma_d)$, couplings, and step sizes as in Section~\ref{sec:elliptical-numerics}.}
    \label{fig:elliptical-squaredist}
\end{figure}

We verify our findings with three targets $\pi^{(d)} = \mathcal{N}_d(0_d, \Sigma)$ of increasing limiting eccentricity $\ecc$:
\begin{itemize}
    \item[(i)] $\Sigma_{ij} = 0.5^{|i-j|}$ for all $(i,j)$: $\ecc = 5/3$. 
    \item[(ii)] $\Sigma = \diag(\sigma_{1}^2,\dots,\sigma_{d}^2)$ where $(\sigma_{i}^2)_{i\ge 1} \iid \chi^2_3$: $\ecc = 3$.
    \item[(iii)] $\Sigma = \diag(1,24,1,24,\dots)$: $\ecc = 25/4$.
\end{itemize}
Target (i) is an AR(1) process with autocorrelation $0.5$ and unit volatility and corresponds to Example~\ref{example:ar-p}. Target (ii) directly corresponds to Assumption~\ref{assumption:ellipt-2}. Target (iii) corresponds to Remark~\ref{remark:ode-ellipt}, where an ODE limit can be shown. We consider the same two step size scalings in the natural parameter (a) $\lambda = 2.38$ and (b) $\lambda = \sqrt{2}$ as in Section~\ref{sec:spherical-numeric}, as the former is optimal for the stationary phase and we expect the latter to perform well in the transient phase. We fix $d = 2{,}000$ and we always start the coupled chains independently from the target.

Figure~\ref{fig:elliptical-squaredist} shows that the MCMC traces resemble deterministic trajectories, hinting towards an ODE limit. For all couplings, the long-time behaviour of the chains is as predicted by our theory, with GCRN performing well and the two baselines being impractical. Remarkably, GCRN performs identically to the implementable asymptotically optimal Markovian coupling of Appendix~\ref{app:optimal-markovian}: here, as we show for a more general class of targets in Section~\ref{sec:prod-target}, it is because the GCRN coupling itself is asymptotically optimal Markovian when the chains are marginally stationary. As all evidence in the main text and in Appendix~\ref{app:optimal-markovian} indicates that GCRN is very nearly as efficient as the asymptotically optimal Markovian coupling, we favour parsimony and hereafter focus on GCRN.

\section{From theory to practice} \label{sec:in-practice}

In this section, we turn our scaling limit analysis into practical recommendations. Investigating the one-step dynamics of the coupled chains, we expose the necessity for a scalable coupling to synchronize acceptance events between chains. This prompts a study into what conditions are required for the GCRN to work well in practice, as well as a hybrid between the GCRN and reflection couplings, which we expect will coalesce the chains quickly even for eccentric targets.

\subsection{Necessity of synchronizing acceptance events} \label{sec:sync-accept}

Let us interpret the one-step dynamics of coupled RWM chains with step size $h = \bigO(d^{-1/2})$. We have that
\begin{equation} \label{eqn:onestep-dynamics}
	\lVert X_{t+1} - Y_{t+1} \rVert^2 - \lVert X_t - Y_t \rVert^2 = 2 h (X_t - Y_t)^\top (Z_{x} B_x - Z_{y}B_y) + h^2\lVert Z_{x} B_x - Z_{y} B_y \rVert^2.
\end{equation}
One might think of the first term of~\eqref{eqn:onestep-dynamics} as representing the \emph{contractive} part of the dynamics. Its expectation conditional on $(X_t,Y_t)$ is invariant to the coupling; our scaling limits (see e.g. the proof of Proposition~\ref{prop:drift-elliptgauss}) indicate that this expectation is roughly $h^2 (X_t - Y_t)^\top (\nabla \log \pi(X_t) - \nabla \log \pi(Y_t))$ times the acceptance rate, at least when the target is sufficiently regular and each marginal chain is close enough to stationarity. We therefore expect the contraction to be roughly exponential at rate $\bigO(h^2) = \bigO(d^{-1})$, particularly if the target is log-concave.

The second term of~\eqref{eqn:onestep-dynamics} represents the \emph{expansive} part of the dynamics. In particular, a linear $\Theta(1)$ increase is suddenly incurred whenever there is an imbalanced acceptance event, that is acceptance in one chain at the same time as a rejection in the other. This highlights the necessity to use a coupling that synchronizes acceptance events with high probability: for instance, if imbalanced acceptance events occur with $\Theta(1)$ probability, then the equilibrium between the contractive and expansive parts of the dynamics lies at $\lVert X_t - Y_t \rVert^2 = \Theta(d)$, so the coupling cannot scale. Our fixed-point results (Propositions~\ref{prop:fixedpoint-stdgauss} and~\ref{prop:fixedpoint-elliptgauss}) validate this behaviour.

To us, at least, the only principled way of synchronizing acceptance events is by paying careful attention to the acceptance steps, e.g. by exploiting additional local information about the target density, such as its logarithmic gradient as in the GCRN coupling. At the same time, our optimality results concerning GCRN suggest that, in order to construct scalable couplings, judicious inclusion of gradient information is to some extent sufficient. In the next section, we therefore gauge the extent to which the GCRN coupling might perform well in practice.

\subsection{When does GCRN work and when does it not?} \label{sec:prod-target}

To gain insight into what conditions are required for the GCRN coupling to scale well, we consider a general product-target setting with variable lengthscales as in \cite{roberts2001optimal}. To obtain a sensible limit theory, we fix the chains to be marginally stationary; we find that GCRN is asymptotically optimal for contraction among \emph{all Markovian couplings}. We leave extensions to non-stationary chains for further work. 

\begin{assumption}\label{assumption:product-target}
	Let $\pi^{(d)}(x) = \prod_{i=1}^d \omega_i f(\omega_i x)$ with $(\omega_i^2)_{i \ge 1}\iid \mu$, where $\mu$ has finite first moment and where $f:\mathbb{R} \to (0,\infty)$ is twice continuously differentiable with $(\log f)'$ Lipschitz continuous, $\mathbb{E}_{Y \sim f}[(\log f)'(Y)^8] < \infty$ and $\mathbb{E}_{Y \sim f}[(\log f)''(Y)^4]< \infty$.
\end{assumption}

\begin{theorem} \label{thm:gcrn-opt-product}
For all $d \ge 1$, let the target $\pi^{(d)}$ be as in Assumption~\ref{assumption:product-target}, and let the joint distribution of $(X_t,Y_t)$ be in $\Gamma (\pi^{(d)},\pi^{(d)})$, where $\Gamma(\mu,\nu)$ is the set of all couplings of the distributions $(\mu,\nu)$. Then, 
    \begin{equation*}
        \lim_{d\to\infty} \sup_{\bar{K} \in \mathcal{M}} \mathbb{E} \left[ h^2 Z_x^\top Z_y B_x B_y \right] = 2\ell^2\Phi \left(-\ell (bI)^{1/2}/2 \right),
    \end{equation*}
where $I = \mathbb{E}_{Y \sim f} [(\log f)'(Y)^2]$, $b = \mathbb{E}[\omega_i^2]$, and this supremum is attained by the GCRN coupling.
\end{theorem}

The optimality of GCRN is therefore not a purely Gaussian phenomenon: it occurs in Theorem~\ref{thm:gcrn-opt-product} because the variation in the Hessian term of the log-acceptance ratio is, relative to the gradient term, asymptotically negligible. In practice and for a well-tuned RWM algorithm, this condition is approximately satisfied when the logarithmic density of the target does not have any direction in which it varies particularly rapidly \citep[see][]{sherlock2013optimal}; equivalently, when none of the lengthscales of the target are much smaller than the average. We expect the GCRN coupling to perform well when this is the case, particularly if the target is also log-concave, with the ideal behaviour of $\lVert X_t - Y_t \rVert^2$ under GCRN being that of exponential contraction at rate $\bigO(h^2)$ to a small steady-state average.

Conversely, GCRN only explicitly accounts for first-order variation in the acceptance ratio so this coupling may perform poorly when the second-order terms in the acceptance ratio exhibit substantial variability. In practice, the issue can appear when the precision matrix of the target has a few very large eigenvalues. Preconditioning may alleviate the issue, as may reducing the step size; the latter reduces the relative variation from second-order terms in the acceptance ratio, increases the acceptance rate, and improves the contraction by forcing the chains to accept simultaneously more often. One can also harness stochasticity to enhance the contraction of the GCRN coupling, as we show in Section~\ref{sec:correctedrefl}.

\begin{remark}
    Theorem~\ref{thm:gcrn-opt-product} suggests good performance in finite dimensions, but it does not guarantee it. To formally establish the rate of contraction and the long-time behaviour under GCRN in any finite dimension would require a careful nonasymptotic analysis of the contractive and expansive terms of~\eqref{eqn:onestep-dynamics}, which we do not perform here. Instead, we verify the behaviour of GCRN empirically, and we recall the encouraging results in the Gaussian case (Sections~\ref{sec:fixedpoint-stdgauss} and~\ref{sec:fixedpoint-elliptgauss}).
\end{remark}

\begin{remark}
In \cite{roberts2001optimal}, the quantity $I$ of Theorem~\ref{thm:gcrn-opt-product} is interpreted as a marginal roughness measure, whereas the moment $b$ quantifies the variation across the coordinates of the target. In Assumption~\ref{assumption:ellipt-1}, considered in the elliptical Gaussian case, these correspond to $I=1$ and $b = z_1$.
\end{remark}

\subsection{The GCRefl coupling: combining contraction with stochasticity} \label{sec:correctedrefl}

In the diffusion literature, it is well-known that different couplings are suited to different purposes \citep{chenli1989couplings}. When the goal is coalescence, reflection couplings have been seen to be highly effective: in particular, they minimize meeting times in the case of coupled Brownian motions or Ornstein-Uhlenbeck processes \citep[e.g.][Chapter~3.4]{connor2007thesis}. However, we have seen that the reflection coupling of the RWM does not perform well for high-dimensional eccentric targets (Section~\ref{sec:elliptgauss}) as it is no longer able to synchronize acceptance events (Section~\ref{sec:sync-accept}).

We therefore propose a hybrid between the GCRN and reflection couplings, designed to combine the favourable properties of both of these. We call it the \textbf{GCRefl} (\emph{Gradient-Corrected Reflection}) coupling:
\begin{equation*}
\begin{aligned}
    Z_x &= Z - (e_x^{\top}Z)e_x + Z_\nabla e_x,\\
    Z_y &= Z - 2(e^\top Z) e - (e_y^{\top}Z)e_y + Z_\nabla e_y,
\end{aligned}
\end{equation*}
where: $e = \norm(X-Y)$; $e_x = \norm(n_x - (e^\top n_x)e )$ with $n_x = \norm(\nabla\log\pi(X_t))$ and similarly for $e_y$; $Z\sim\mathcal{N}_d(0_d,I_d)$ and $Z_\nabla \sim \mathcal{N}_1(0,1)$ are independent. We default to the reflection coupling when a vector to be normalized is null (this is always the case in dimension $d=1$). The proposed coupling starts from the reflection coupling, then applies a GCRN-like correction in order to increase the rate of simultaneous acceptances.

One can roughly interpret the one-step dynamics of $\lVert X_t -Y_t\rVert$ under the GCRefl coupling as exponential contraction at rate $\bigO(h^2)$ together with a random walk with increments $\bigO(h)$. The exponential contraction dominates when $\lVert X_t - Y_t\rVert = \bigO(d^{1/2})$, whereas the stochasticity dominates when $\lVert X_t - Y_t\rVert = \bigO(1)$. Borrowing intuition from the case of diffusions, this added stochasticity should help drive the chains towards coalescence. Indeed, as we will see in the experiment of Section~\ref{sec:rwm-vs-mala}, the reflection move of GCRefl is particularly helpful in scenarios where GCRN is only able to contract the chains to within $\bigO(1)$ squared distance.

\section{Numerical experiments} \label{sec:applications}

In this section, we illustrate practical applications of our proposed couplings, in particular the estimation of the rate of convergence and asymptotic variance of the RWM, as well as of the bias of an approximate sampling method. We also use a natural extension to the GCRN coupling to devise an effective coupling for the Hug and Hop algorithm \citep{ludkin2022hug} and to quantify the rate of convergence of this algorithm. Finally, we perform a comparative study of coupled RWM and MALA kernels, with a view towards unbiased MCMC. Our discrepancies of choice in this section are the total variation distance and the squared 2-Wasserstein distance
\begin{equation*}
   \tvd(\mu, \nu) = \inf_{ (X,Y) \in \Gamma(\mu,\nu)}\mathbb{E}[\mathbbm{1}\{X = Y\}], \quad \wass{2}^2(\mu, \nu) = \inf_{ (X,Y) \in \Gamma(\mu,\nu)}\mathbb{E}\left[\| X - Y\|^2\right].
\end{equation*}
When estimating the rate of convergence, the bound~\eqref{eqn:wass-general-ub} for the total variation distance simplifies to $\tvd(\pi_t, \pi) \le \mathbb{E}[ 0 \lor \ceil\{ (\tau - t)/ L\} ]$, see \cite{biswas2019estimating}.

Throughout, we have striven to tune the considered algorithms and couplings close to optimally. We defer additional details to Appendix~\ref{app:computation}, including: (i) further algorithmic descriptions, (ii) experiments regarding parameter choice, (iii) further discussion on the contractivity of the couplings used, (iv) alternative coupling strategies.

\subsection{Rate of convergence of the RWM} \label{sec:numeric-conv}

\begin{figure}[tb]
    \centering
    \includegraphics[width=\textwidth]{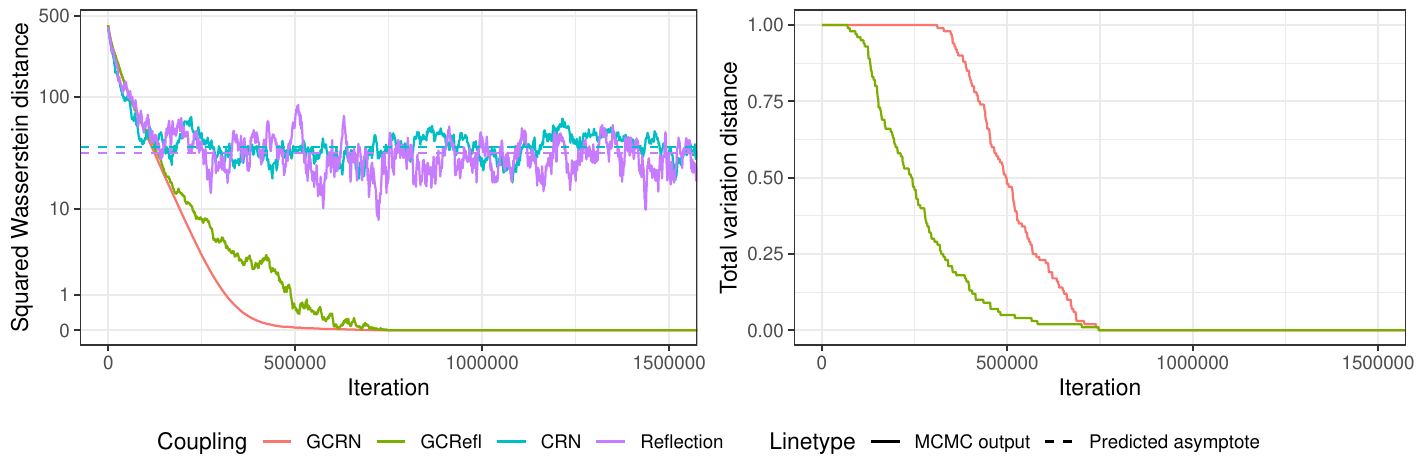}
    \caption{Rate of convergence of the RWM targeting the SVM, when started from the prior. We estimated upper bounds on $\wass{2}^2(\pi_t, \pi)$ and $\tvd(\pi_t, \pi)$ using only the GCRN and GCRefl couplings. The reflection-maximal and CRN couplings were not sufficiently contractive for this problem, as illustrated by traces of the squared distance $\lVert X_{t} - Y_t \rVert^2$ in the left plot.}
    \label{fig:svm-convergence}
\end{figure}

We illustrate the effectiveness of our proposed couplings at quantifying the rate of convergence of the RWM in a challenging high-dimensional setting. We target the posterior distribution~$\pi(x_{1:d} \mid y_{1:d})$ of a stochastic volatility model \citep[SVM;][Section 9.6.2]{liu2001monte}:
\begin{equation*}
\begin{aligned}
    y_i \mid x_i &\sim \mathcal{N}_1(0,\beta^2\exp(x_i)) &&\textup{for } i \in\{1,\dots,d\},\\
    x_{i + 1} \mid x_i &\sim \mathcal{N}_1(\varphi x_i, \sigma^2) &&\textup{for } i \in\{ 1,\dots,d-1\}, \\
    x_1 &\sim \mathcal{N}_1\left(0, \sigma^2/(1 - \varphi^2)\right).
\end{aligned}
\end{equation*}
We fix the dimension to $d = 360$, hyperparameters to $(\beta, \varphi, \sigma) = (0.65,0.98,0.15)$, and generate the data from the model. We fix the starting distribution to be the prior $\pi_0 = \pi(x_{1:d})$. Our goal is to estimate upper bounds on $\tvd(\pi_t, \pi)$ and $\wass{2}^2(\pi_t, \pi)$ as described in Section~\ref{sec:background}.

Before running MCMC, we compute a Laplace approximation $\hat\pi = \mathcal{N}_d(\hat\mu, \hat\Sigma)$. This suggests the step size $h = 2.38/\tr(\hat\Sigma^{-1})^{1/2}$, which we employ and empirically verify as corresponding to a near-optimal acceptance rate of $23\%$. To predict the long-time behaviour of the chains under the CRN and reflection couplings, we plug $\hat\pi$ into Equation~\eqref{eqn:predict-asymptote-elliptgauss}.

To estimate the rate of convergence, we require our couplings to coalesce the chains in finite time. The GCRN and GCRefl couplings cannot produce exact meetings on their own; instead, we opt for a two-scale approach, employing a contractive coupling when the chains are at least $\lVert X_t - Y_t \rVert^2 \ge \delta$ apart, and otherwise (when $\lVert X_t - Y_t \rVert^2 < \delta$) employing the reflection-maximal coupling. Here and in subsequent experiments with two-scale couplings, we select the switching threshold $\delta$ using a grid search on a logarithmic scale: the general trend is that the meeting times are insensitive to choosing $\delta$ smaller than optimal. We leave a formal investigation into the optimal choice of $\delta$ for further work. Our chosen thresholds are $(\delta_\textup{gcrn}, \delta_\textup{gcrefl}) = (0.1, 0.001)$ in this experiment, and we use a large lag of $L = 1.5\times 10^6$ and $100$ independent replicates to compute each bound.

The numerical results are displayed in Figure~\ref{fig:svm-convergence}. Both the two-scale GCRN and GCRefl couplings effectively quantify the rate of convergence of the RWM algorithm. GCRN attains the sharper squared Wasserstein distance bound as it is focused on contraction, while GCRefl attains a sharper total variation distance bound as it is focused on coalescence. In contrast, the reflection-maximal coupling is severely hindered by the eccentricity of the target: the coupling is capable of exact meetings, but as the chains stay far apart the probability of coalescing is effectively null at all iterations. The long-time behaviour of the couplings is accurately predicted by our theory.

\subsection{Bias of approximate sampling}

\begin{figure}[tb]
    \centering
    \includegraphics[width=\textwidth]{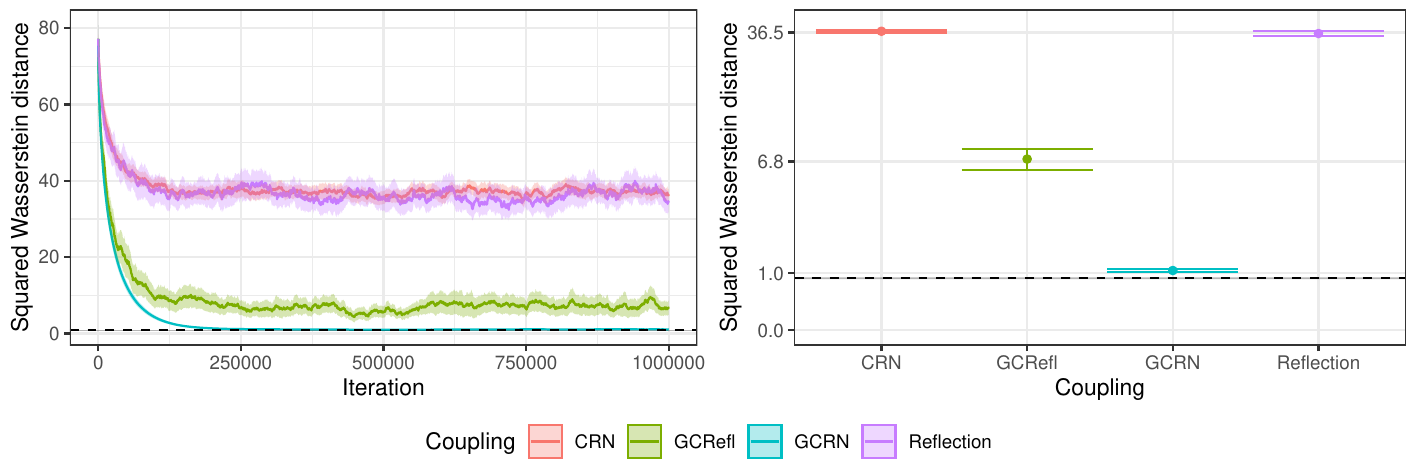}
    \caption{Bias of the Laplace approximation for the SVM. \textbf{Left}: Per-iteration sample average upper bound estimates. \textbf{Right}: Point estimates of upper bounds with $(S,T) = (3.5 \times 10^5,10^6)$. The dashed line is the lower bound of \cite{gelbrich1990on}. All estimates are shown with $\pm2$ standard errors (in either shaded regions or error bars).}
    \label{fig:svm-laplace-bias}
\end{figure}

Couplings can also be used to estimate the bias of approximate sampling procedures \citep{biswas2023bounding}. Motivated by the accuracy of the quantities inferred from the Laplace approximation $\hat\pi$ of the SVM target $\pi$ of Section~\ref{sec:numeric-conv} (such as the optimal step size $h$ and the expected long-time behaviour under the CRN and reflection couplings), we use our proposed couplings to compute upper bounds on $\wass{2}^2(\hat\pi, \pi)$.

The upper bounds are computed with a small modification to the method of Section~\ref{sec:background}. We set the lag to be $L=0$, target the exact distribution $\pi$ with the $X$-chain and its approximation $\hat\pi$ with the $Y$-chain. Assuming that the chains start marginally stationary, irrespective of their coupling it holds \citep{biswas2023bounding} that
\begin{equation*}
	\wass{2}^2(\hat\pi, \pi) \le \sum_{t = S}^T \mathbb{E}\left[\lVert X_t -Y_t \rVert^2 \right],
\end{equation*}
for any integer $T \ge S \ge 0$. In practice, we replace expectations by empirical averages and due to burn-in the bound only holds asymptotically as $T \to \infty$. In our experiment, we start the chains independently from $X_0 \sim \pi$ (using a long MCMC run) and $Y_0 \sim \hat\pi$. To compute upper bounds, we use natural extensions to the GCRN and GCRefl couplings, changing $n_y \leftarrow \hat n_y = \norm(\nabla\log\hat\pi(Y_t))$. As coalescence is not the goal here, none of our couplings attempt to make the chains meet. We use $100$ independent replicates to compute upper bounds, with the same step size $h$ as in Section~\ref{sec:numeric-conv}. We also compute \citep[see][]{gelbrich1990on} the lower bound $\wass{2}^2(\mathcal{N}_d(\hat\mu, \hat\Sigma), \mathcal{N}_d(\mu, \Sigma)) \le \wass{2}^2(\hat\pi, \pi)$, where $(\mu,\Sigma)$ are the mean and covariance of $\pi$.

The numerical results are displayed in Figure~\ref{fig:svm-laplace-bias}. The GCRN coupling produces the most informative upper bound on the bias of the Laplace approximation in Wasserstein distance, which is also remarkably sharp compared to the lower bound. For context, the upper bound implies the relative error bound $\lVert \hat\mu - \mu\rVert / \lVert \mu\rVert \le 10\%$. Compared to the Wasserstein distance, the total variation distance is a much more restrictive metric in higher dimensions, and as a consequence non-trivial bounds on $\tvd(\hat\pi, \pi)$ are harder to obtain \citep{biswas2023bounding}. Nonetheless, by adapting the coalescive two-scale GCRefl coupling, we were able to obtain the bound $\tvd(\hat\pi, \pi) \le 0.964 \, (\pm 0.003)$ to two standard errors.

\subsection{Coupling the Hug and Hop algorithm} \label{sec:hh}

\begin{figure}[tb]
    \centering
    \includegraphics[width=\textwidth]{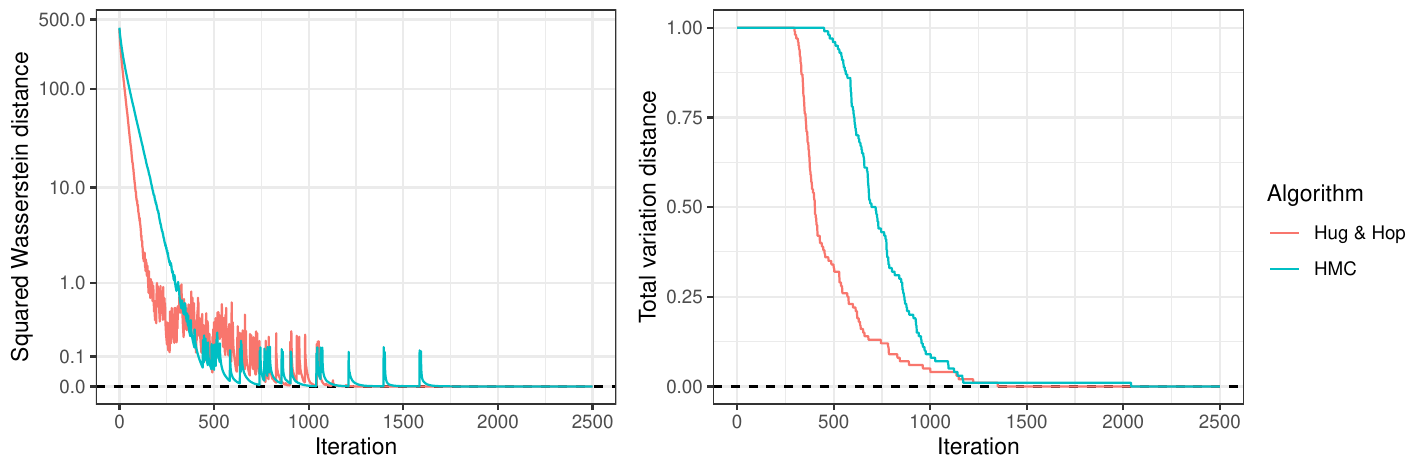}
    \caption{Rate of convergence of the H\&H and HMC algorithms targeting the SVM, when started from the prior. See Section~\ref{sec:hh} of the main text for details on these upper bound estimates.}
    \label{fig:hh}
\end{figure}

The stochastic volatility model considered in the previous sections provided a challenging test case for the RWM and its couplings. The RWM was not necessarily a practical algorithm for the problem; one would have preferred a more scalable gradient-based algorithm such as Hamiltonian Monte Carlo \citep[HMC;][]{duane1987hybrid,neal2011} or the recently-proposed Hug and Hop \citep[H\&H;][]{ludkin2022hug}. Our goal in this section is two-fold: to study the rate of convergence of H\&H with couplings and to demonstrate that the coupling ideas we developed for the RWM extend to other algorithms.

H\&H alternates between the skew-reversible Hug kernel, which proposes large moves by approximately traversing the same level set of the log-target, and the Hop kernel, which crosses between level sets by encouraging large moves in the gradient direction. Hop uses Gaussian proposals centered at the current state; the projection of the proposal onto the subspace orthogonal to the gradient is isotropic, and the projection onto the gradient uses a much larger scaling than each coordinate of its orthogonal counterpart. Due to its similarities to the RWM kernel and its use of gradient information, GCRN-like couplings are natural for the Hop kernel; one might also expect such couplings to be contractive. To couple Hop proposals, we must couple Gaussians with different covariance matrices; whereas GCRN can be straightforwardly extended to this case, extending the reflection-maximal coupling appears significantly more challenging \citep{corenflos2022coupled-rejection}.

Our coupling of H\&H kernels is inspired by contractive couplings of HMC and RWM kernels and does not require additional gradient evaluations compared to, say, simulating two H\&H chains independently. We synchronize the momenta in Hug proposals; such a coupling contracts HMC proposals \citep[e.g.][]{heng2019unbiased}, and due to the parallels of Hug and HMC we expect the contractivity to carry over to Hug. For Hop, we use a two-scale coupling which aims to contract the chains when far apart and allow for exact meetings when close together. When $\lVert X_t - Y_t \rVert^2 \ge \delta$, we couple Hop proposals according to the GCRN coupling
\begin{equation*}
\begin{aligned}
    X_p &= X_t + \frac{\lambda}{\gamma_x} Z_\nabla n_x + \frac{(\lambda\kappa)^{1/2}}{\gamma_x} \big\{Z - (Z^\top n_x) n_x\big\},\\
    Y_p &= Y_t + \frac{\lambda}{\gamma_y} Z_\nabla n_y + \frac{(\lambda\kappa)^{1/2}}{\gamma_y} \big\{Z - (Z^\top n_y) n_y\big\},
\end{aligned}
\end{equation*}
where: $\gamma_x = \lVert \nabla\log\pi(X_t) \rVert$, $n_x = \nabla\log\pi(X_t)/\gamma_x$ and similarly $\gamma_y$ and $n_y$; $Z\sim\mathcal{N}_d(0_d,I_d)$ and $Z_\nabla\sim\mathcal{N}_1(0,1)$ are independent. When $\lVert X_t - Y_t \rVert^2 < \delta$, we sample the proposals from a maximal coupling with independent residuals \citep[Algorithm~2]{jacob2020unbiased}. We encourage simultaneous acceptance in both both the Hug and Hop kernels by synchronizing the uniform acceptance variates.

We also compare H\&H with HMC, following \cite{heng2019unbiased}: we use a synchronous coupling of HMC kernels and mix in coupled RWM kernels to allow the chains to meet. As HMC is particularly sensitive to its tuning parameters, we perform extensive experimentation to optimize the contractivity of this algorithm; to ensure a fair comparison, we also do this for Hug. We find that HMC suffers from a sharp phase transition, with long integration times adversely affecting the contractivity of the coupling, whereas Hug does not have this issue. For both HMC and Hug, the acceptance rates at the optimally contractive parameters are higher than would be optimal for mixing \citep{beskos2013optimal,ludkin2022hug}; this is linked to contraction being adversely affected by acceptance in one chain but rejection in the other, as discussed in Section~\ref{sec:in-practice}.

Figure~\ref{fig:hh} shows estimated upper bounds on the rate of convergence of H\&H and HMC, computed with a lag of $L = 2500$ and with $100$ replicates each. H\&H coalesces quicker than HMC here, while also being more robust with respect to tuning. Both algorithms coalesce two orders of magnitude quicker than the RWM, suggesting that they would be suitable for suitable for unbiased MCMC in this setting.

\subsection{Comparison of the RWM and MALA algorithms} \label{sec:rwm-vs-mala}

We consider a logistic regression posterior $\pi$ on the UCI Sonar dataset with $(n,d) = (208,61)$ observations and regressors (covariates and intercept). Following standard practice \citep{gelman2008weakly}, we standardize the covariates to scale 0.5, then place a spherical Gaussian prior $\mathcal{N}_{d}(0_d, 25 I_{d})$ jointly on the regressors. The RWM and the Metropolis-adjusted Langevin algorithm (MALA; e.g. \citealp{roberts1998optimal}) perform well in such problems of moderate dimension and with relatively few observations \citep{chopin2017leave}. Our goal here is to compare these algorithms in a realistic setting, with a view towards unbiased MCMC. We are therefore interested both in their time to coalescence, as well as their performance at stationarity.

To ensure a fair comparison between algorithms, we emulate lagged couplings with large lag parameters $L$ by always simulating couplings between a stationary $X$-chain and a $Y$-chain which is started at the posterior mean. Alongside meeting times, this set-up allows us to obtain \emph{unbiased} estimates of the asymptotic variance of the MCMC algorithms using the ``EPAVE" estimator of \cite{douc2023solving}, see Appendix~\ref{app:binreg}. As is commonly done in practice, we precondition the RWM and MALA proposals to, respectively:
\begin{equation*}
X_t + (h P) Z_x, \quad X_t + \big(h^2PP^\top/2\big)\nabla\log\pi(X_t) + (h P) Z_x,
\end{equation*}
where $Z_x\sim\mathcal{N}_d(0_d, I_d)$ and where $P\in\mathbb{R}^{d \times d}$ is estimated from a preliminary run. The preconditioning requires the change $n_x = \norm(P^T \nabla\log\pi(X_t))$ in our gradient-based couplings and $e = \norm(P^{-1} (X_t-Y_t))$ in our reflective couplings, see Appendix~\ref{app:precondition}. We consider two choices for the preconditioner $P$: a diagonal matrix corresponding to the standard deviations of the target marginals and the ``full" Cholesky factor of the target covariance matrix.

In preliminary experiments (see Appendix~\ref{app:binreg}) we found reflective couplings to be particularly useful for this problem. For the RWM, we select a two-scale coupling which applies GCRefl when $\lVert P^{-1}(X_t-Y_t) \rVert^2 \ge \delta$, and otherwise applies a reflection-maximal coupling, with thresholds $\delta_\textup{diag} = 10h^2$ and $\delta_\textup{full} = 1$ depending on the preconditioner $P$. For the diagonal preconditioner, GCRefl was the only practical coupling; other couplings suffered disproportionately from the fact that the preconditioned target was highly skewed, with the largest eigenvalue of the preconditioned covariance being nearly three-tenths of the trace. For the full preconditioner, the reflection coupling was practical, but we found GCRefl to outperform it even after accounting for the additional gradient evaluations. For MALA, we use a reflection-maximal coupling. We emphasize that both the RWM and MALA couplings require access to the gradient oracle.

\begin{figure}[tb]
    \centering
    \includegraphics[width=\textwidth]{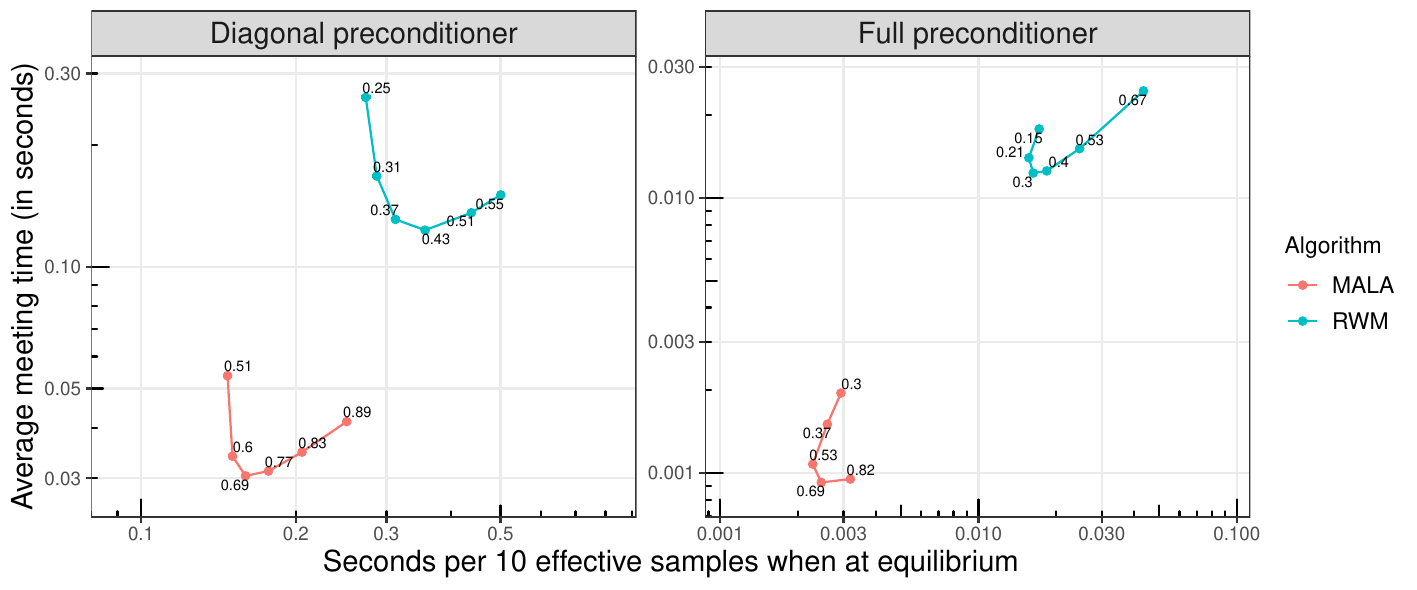}
    \caption{Comparison of RWM and MALA as in Section~\ref{sec:rwm-vs-mala}, varying the step size parameter~$h$. All estimates are shown with two standard errors. The acceptance rate for each step size is overlaid.}
    \label{fig:binreg}
\end{figure}

Figure~\ref{fig:binreg} compares the RWM and MALA algorithms in terms of wall-time. (See Appendix~\ref{app:binreg} for a figure which is not scaled by wall-time.) To measure mixing at stationarity, we consider the problem of estimating the regression coefficients and we report values corresponding to the smallest effective sample size across all coordinates. When using a diagonal preconditioner and optimizing for mixing/meeting, the RWM is $2\times$/$4\times$ slower than MALA. This may seem surprisingly close; it stems from the fact that MALA is more sensitive to the eccentricity of the target than the RWM \citep{livingstone2022barker}, so that both samplers must use similar step sizes in our setting. The efficiency gap widens when using a full preconditioner; optimizing for mixing/meeting, the RWM is $7\times$/$13\times$ slower than MALA. Translating these relative efficiencies to the unbiased MCMC setting, mixing efficiency is more relevant when a small number of long chains are run, and coupling efficiency is more relevant when a large number of short chains are run; we expect MALA to more clearly outperform the RWM in the latter regime. In passing, there is a trade-off between coupling and mixing efficiency when using the diagonal preconditioner, whereas using the full preconditioner mitigates this issue.

\section{Discussion} \label{sec:discussion}

The main takeaway from this paper is the following: in order to design effective couplings of the RWM algorithm, one should pay careful attention to the coupling of the acceptance steps. We have shown that, by making judicious use of gradient information so as to synchronize acceptance events with high probability, one can design contractive couplings which remain effective even in high-dimensional regimes. We have demonstrated the effectiveness of our proposed GCRN and GCRefl couplings at estimating the rate of convergence and the asymptotic variance of the RWM algorithm. We have also demonstrated how the contractivity of the GCRN coupling can be leveraged to estimate the bias of an approximate sampling method.

The utility of our proposed couplings of the RWM for unbiased MCMC is less clear. They demand access to a gradient oracle; it is unclear to us how one could devise similarly scalable couplings of the RWM algorithm without such information, yet access to this oracle also enables the use of more scalable gradient-based algorithms altogether, such as MALA and HMC. We have exhibited one moderate-dimensional multi-scale setting where the gap between the RWM and MALA is not as insurmountable as one might expect a priori. Random walk proposals are also a competitive alternative in the pseudo-marginal setting \citep{andrieu2009pseudo-marginal}. \cite{middleton2020unbiased} demonstrate the use of simple maximal couplings for unbiased pseudo-marginal MCMC; devising improved couplings will however require additional design considerations to ours, as one must also explicitly account for the noise induced by the unbiased estimator of the target density.

Our work points to several avenues of investigation. Methodologically, we expect the framework of asymptotic optimality to yield improved couplings for other MCMC algorithms. Indeed, we have demonstrated how a straightforward extension of our ideas yielded an effective coupling of the Hug and Hop algorithm that is well-suited to unbiased MCMC, as its per-iteration cost is no larger than that of two independent chains. For MALA, preliminary results (not included) suggest that, as for the RWM, higher-order derivative information than that available in the proposal should be incorporated in an asymptotically optimal coupling. Theoretically, contractive couplings have been used to obtain quantitative convergence rates for MCMC algorithms, e.g. MALA \citep{eberle2014error} and HMC \citep{bou-rabee2020coupling}. Our scaling limits suggest that, by adapting the couplings proposed in this paper, a sharp $\bigO(d)$ dimensional dependence for the RWM \citep{andrieu2024explicit} could be recovered. Finally, extensions of our scaling limits beyond the Gaussian case may be possible \citep{kuntz2019diffusion}, and a formal explanation of the success of the proposed couplings could be obtained by establishing deeper connections between the Langevin diffusion and the RWM.

\section*{Acknowledgements}

This paper is based on work completed while Tam\'as P. Papp was part of the EPSRC-funded STOR-i Centre for Doctoral Training (EP/S022252/1). Tam\'as P. Papp was also supported by the EPSRC grant CoSInES (EP/R034710/1). We acknowledge insightful discussions with Sam Power and Pierre E. Jacob. We thank the associate editor and the three anonymous reviewers for their comments, which substantially improved this manuscript.

\putbib
\end{bibunit}

\appendix

\begin{bibunit}

\section{Unbiased MCMC with couplings} \label{app:unbiased-mcmc}

The lagged coupling framework recalled in Section~\ref{sec:background} can also be used for the unbiased estimation of expectations of test functions $h(\cdot)$ with respect to the target $\pi$ \citep{jacob2020unbiased,biswas2019estimating,douc2023solving}. We require additional conditions on the tail decay of the meeting time $\tau$ and the moments of the function of interest $h(\cdot)$, namely:
\begin{enumerate}
	\item There exist $\delta \in (0,1)$ and $C>0$ such that $\mathbb{P}(\tau > t) \le C\delta^t$ for all $t \ge 0$.
	\item It holds that $\lim_{t \to \infty}\mathbb{E}[h(Y_t)] = \mathbb{E}_\pi[h(Y)]$. Additionally, there exist $\eta,C > 0$ such that $\mathbb{E}[h(Y_t)^{2+\eta}] \le C$ for all $t \ge 0$.
\end{enumerate}
Then, for any integer $m \ge k \ge 0$, the following is an unbiased estimator of $\mathbb{E}_\pi[h(X)]$:
\begin{equation*}
    H_{k:m} = \frac{1}{m+k-1}\sum_{t=k}^{m} h(X_{t-L}) + \sum_{t=k}^{\tau-1} \frac{\gamma(t;k,m,L)}{m-k+1} \{h(X_t) - h(Y_t)\} 
\end{equation*}
where $\gamma(t;k,m,L) = 1 + \lfloor (t-k)/L \rfloor - \lceil 0 \lor (t-m)/L \rceil$. The assumptions also ensure that $H_{k:m}$ has finite variance, so that an average of i.i.d. copies of $H_{k:m}$ is guaranteed to converge at the Monte Carlo rate. In turn, this enables principled parallel MCMC, through averages of estimators computed through pairs $(X,Y)$ simulated in parallel. 

It is clear that the meeting time $\tau$ imposes a lower bound on the length of the simulation. Numerical evidence also indicates that the variance of the estimator grows with the meeting time. It is therefore important to design couplings which ensure that the meeting times stay small, say, on average.

See \citet[Appendix~B]{douc2023solving} for a derivation of the estimator $H_{k:m}$; the form originally given in the rejoinder of \cite{jacob2020unbiased} is unfortunately incorrect. The geometric tail decay condition for the meeting time $\tau$ was relaxed to polynomial in \cite{middleton2020unbiased}, however the moment condition on $h(\cdot)$ appears crucial. We briefly discuss the use of couplings for the unbiased estimation of the asymptotic variance of an MCMC algorithm \citep{douc2023solving} in Appendix~\ref{app:rwm-vs-mala}.

\section{Additional discussion on couplings of the RWM} \label{app:extra-couplings}

\subsection{Asymptotically optimal Markovian coupling} \label{app:optimal-markovian}

We recall the approximations that motivated the GCRN coupling. Let $V(x) = \log \pi(x)$. Using firstly a Taylor expansion, the acceptance indicator at $X_t = x$ is
\begin{align*}
    B_x 
	&\approx \mathbbm{1}\big\{\log U_x \le h Z_x^\top \nabla V(x) + h^2 Z_x^\top \nabla^2 V(x) Z_x \big\},\\
	&\approx \mathbbm{1}\big\{\log U_x \le g(x) Z_{\nabla x} + c(x) \big\},
\end{align*}
where: $Z_{\nabla x} = Z_x^\top n_x \sim \mathcal{N}_1(0, 1)$ with $n_x = \nabla V(x) / \lVert \nabla V(x)\rVert$; $g(x) = h \lVert \nabla V(x)\rVert$ and $c(x) = h^2 \tr(\nabla^2 V(x))$ are constants. Scaling the step size as $h = \ell d^{-1/2}$ with the dimension $d$ and assuming standard asymptotics as $d \to \infty$ (as in e.g. \citealp{roberts1997weak}), the above approximations are sharp in the limit.

These observations motivated the GCRN coupling and were used to prove that it was asymptotically optimal for contraction over the class of product couplings~$\mathcal{P}$. However, by considering the random variables $\xi_x =  g(x) Z_{\nabla x} - \log U_x$ and $\xi_y =  g(y) Z_{\nabla y} - \log U_y$ directly, it is possible to construct a coupling that is asymptotically optimal over the entire class of Markovian couplings~$\mathcal{M}$.

\subsubsection{An implementable optimal coupling} 

We seek a coupling of RWM kernels which is optimally contractive with respect to the squared Euclidean distance, so equivalently we seek to maximize $\mathbb{E}[ h^2 Z_x^{\top} Z_y B_x B_y \mid (X_t, Y_t) = (x,y)]$. To derive an implementable optimal coupling, we use the following upper bound on the objective:
\begin{equation}\label{eqn:opt-mark-obj-bound}
\begin{aligned}
\mathbb{E}\left[ h^2 Z_x^{\top} Z_y B_x B_y \mid (X_t, Y_t) = (x,y) \right]
	&\lessapprox \ell^2 \mathbb{E}[ \mathbbm{1}\{0 \le \xi_x  + c(x) \} \mathbbm{1}\{0 \le \xi_y  + c(y) \}],\\
	&\le \ell^2 \mathbb{P}(0 \le \xi_x + c(x) ) \land  \mathbb{P}(0 \le \xi_y + c(y)),
\end{aligned}
\end{equation}
where the first approximate inequality is asymptotically sharp as $d\to \infty$, and the second inequality is trivial. The first inequality is satisfied if $Z_x \approx Z_y$ up to a low-rank perturbation. Using a standard optimal transport argument \citep[Remark~2.19]{villani2003topics}, one can show that an optimal coupling of $(\xi_x, \xi_y)$ which attains the second inequality is $\xi_x = F^{-1}_x(U)$ and $\xi_y = F^{-1}_y(U)$, where $U \sim \unif(0,1)$ and $F_{x,y}$ are the respective CDFs of $\xi_{x,y}$. Since $\xi_{x}$ only constrains one coordinate of $Z_{x}$ (and since $\xi_{y}$ similarly constrains $Z_{y}$), we can construct a coupling which renders both inequalities of Equation~\eqref{eqn:opt-mark-obj-bound} asymptotically tight.

\begin{algorithm}
  \caption{Asymptotically optimal Markovian coupling \label{alg:opt-mark-coup}}
  \begin{algorithmic}[1]
    \Require{Target density $\pi:\mathbb{R}^d \to \mathbb{R}$, score $\nabla\log\pi:\mathbb{R}^d \to \mathbb{R}^d$, step size $h >0$, $g(\cdot) = h\lVert \nabla\log\pi(\cdot) \rVert$ and $n(\cdot) = h\nabla\log\pi(\cdot) /g(\cdot)$, current state $x \in \mathbb{R}^d$.}
    \Require{Inverse CDF function $F^{-1}(\cdot \mid \mu,\sigma^2, \lambda)$ of $\textup{EMG}(\mu,\sigma^2, \lambda)$ distribution.}

\Statex
\State Sample $U\sim\textup{Unif}(0,1)$.
\State Set $\xi_x = F^{-1}(U \mid 0, g(x)^2, 1)$ and $\xi_y = F^{-1}(U \mid 0, g(y)^2, 1)$. \Comment{Optimal coupling}

\Statex
\State Sample $Z_{\nabla x} \sim \overline{\mathcal{N}}(g(x), 1 \mid \xi_x /g(x))$ and $Z_{\nabla y} \sim \overline{\mathcal{N}}(g(y), 1 \mid \xi_x /g(y))$. \Comment{Used Lemma~\ref{lemma:emg-conditionals}}
\State Set $\log U_{\nabla x} =  g(x)Z_{\nabla x} - \xi_x$ and $\log U_{\nabla y} = g(y)Z_{\nabla y} - \xi_y$.

\Statex
\State Sample $Z \sim \mathcal{N}_d(0_d, I_d)$.
\State Set $Z_x = Z + \{Z_{\nabla_x} - n(x)^\top Z \} n(x)$ and $Z_y = Z + \{Z_{\nabla y} - n(y)^\top Z \} n(y)$. \Comment{As in GCRN}

\Statex
\State \Return{$(Z_x, Z_y)$ and $(U_x, U_y)$}
  \end{algorithmic}
\end{algorithm} 

Algorithm~\ref{alg:opt-mark-coup} describes our proposed modification to the GCRN coupling.\footnote{The coupling in Line~3 of Algorithm~\ref{alg:opt-mark-coup} is arbitrary.} 
The algorithm induces a coupling of RWM kernels which is asymptotically optimally contractive coupling over $\mathcal{M}$ in the regimes considered in this paper.\footnote{We omit the proof, which follows similar lines to that of Theorem~\ref{thm:gcrn-opt-elliptgauss}.} 

To construct Algorithm~\ref{alg:opt-mark-coup}, we exploited that $\xi_x \sim \textup{EMG}(0,g^2(x),1)$, where $\textup{EMG}(\mu,\sigma^2, \lambda)$ denotes the \emph{exponentially modified Gaussian} distribution \citep[EMG;][]{grushka1972emg}, the distribution of the convolution of a Gaussian $\mathcal{N}_1(\mu,\sigma^2)$ variate and an exponential variate with rate $\lambda$. The correctness of Algorithm~\ref{alg:opt-mark-coup} stems from the conditional distributions of EMG random variables, see Lemma~\ref{lemma:emg-conditionals} below. Algorithm~\ref{alg:opt-mark-coup} is implementable using numerical inversion, as the cumulative distribution function $F(\cdot \mid \mu, \sigma^2, \lambda)$ of an $\textup{EMG}(\mu,\sigma^2,\lambda)$ variable  has the tractable expression,
\begin{equation*}
	F(x \mid \mu, \sigma^2, \lambda) = \Phi(x \mid \mu,\sigma^2) -\frac{1}{2}\exp\left(\frac{\lambda}{2}(2\mu+\lambda \sigma^2 - 2x)\right)\textup{erfc}\left( \frac{\mu + \lambda \sigma^2 - x}{\sqrt{2}\sigma} \right),
\end{equation*}
where $\Phi(\cdot \mid \mu,\sigma^2)$ is the CDF of a $\mathcal{N}_1(\mu,\sigma^2)$ variate and $\textup{erfc}(\cdot)$ is the complementary error function.

\begin{lemma}\label{lemma:emg-conditionals}
    Let $\xi = Z + E \sim \textup{EMG}(\mu,\sigma^2, \lambda)$, that is independently $Z \sim \mathcal{N}_1(\mu, \sigma^2)$ and $E \sim \textup{Exp}(\lambda)$. Then,
    \begin{equation*}
	\begin{aligned}
        Z \mid \xi &\sim \overline{\mathcal{N}}(\mu +\lambda \sigma^2, \sigma^2 \mid \xi),\\
        E \mid \xi &\sim \underline{\mathcal{N}}(\xi - \mu - \lambda\sigma^2, \sigma^2\mid 0),
	\end{aligned}
    \end{equation*}
    where $\overline{\mathcal{N}}(\mu,\sigma^2\mid u)$ denotes the normal distribution $\mathcal{N}_1(\mu,\sigma^2)$ truncated above at $u$, and analogously $\underline{\mathcal{N}}(\mu,\sigma^2 \mid l)$ denotes $\mathcal{N}_1(\mu,\sigma^2)$ truncated below at $l$.
\end{lemma}
\begin{proof}
The relevant density functions are
\begin{equation*}
    p_Z(x) \propto \exp\{-(x - \mu)^2/(2\sigma^2)\}, \quad p_E(x) \propto \exp (-\lambda x) \mathbbm{1}\{0 \le x\}.
\end{equation*}
Since $E + Z$ is a convolution, we have that
\begin{align*}
    p(Z = x \mid E + Z =\xi) &\propto p_Z(x) p_E(\xi - x) \propto \exp\{-(x - \mu - \lambda\sigma^2)^2/(2\sigma^2)\}  \mathbbm{1}\{x \le \xi\},\\
    p(E = x \mid  E + Z =\xi) &\propto p_Z(\xi - x) p_E(x) \propto \exp\{-(x -\xi + \mu + \lambda\sigma^2)^2/(2\sigma^2)\}  \mathbbm{1}\{x\ge 0 \}.
\end{align*}
These are truncated normal distributions, which concludes the proof.
\end{proof}

\subsubsection{ODE limit}

When the target is $\pi^{(d)} = \mathcal{N}_d(0_d, I_d)$ the proposed coupling satisfies a high-dimensional ODE limit as in Theorem~\ref{thm:ode-limit-stdgauss}. The drift in Proposition~\ref{prop:drift-stdgauss} requires the change
\begin{align*}
    b_{\textup{opt}}(x,y,v) &= \ell^2\mathbb{E}\big[1\land e^{\ell x^{1/2}Z - \ell^2/2}\big] \land \mathbb{E} \big[1\land e^{\ell  y^{1/2}Z - \ell ^2/2}\big] - \ell^2 v \left\{ q_\ell (x) + q_\ell (y)  \right\} \\
    &=: p_\ell(x) \land p_\ell(x) - v\left\{q_\ell (x) + q_\ell (y)\right\},
\end{align*}
where $Z \sim \mathcal{N}_1(0,1)$,
\begin{align*}
	q_\ell (x) &= \ell^2 e^{\ell^2(x-1)/2} \Phi\left(\frac{\ell}{2x^{1/2}} - \ell x^{1/2}\right), \tag{Lemma~\ref{lemma:gaussian-integrals}}\\
	p_\ell(x) &= q_\ell (x) + \ell^2\Phi\left(-\frac{\ell}{2x^{1/2}}\right). \tag{Proposition~\ref{prop:drift-stdgauss}}
\end{align*}
By construction, the drift $b_{\textup{opt}}(\cdot)$ is point-wise optimal among all couplings in $\mathcal{M}$. This drift is for the inner-product process; the corresponding drift for the squared-distance process is $\bar{b}_{\textup{opt}}(x,y,s) = a(x) + a(y) - 2b_{\textup{opt}}(x,y,(x+y-s)/2)$.

\subsubsection{Optimal scaling and relative efficiency}
\begin{figure}
    \centering
    \includegraphics[width = \textwidth]{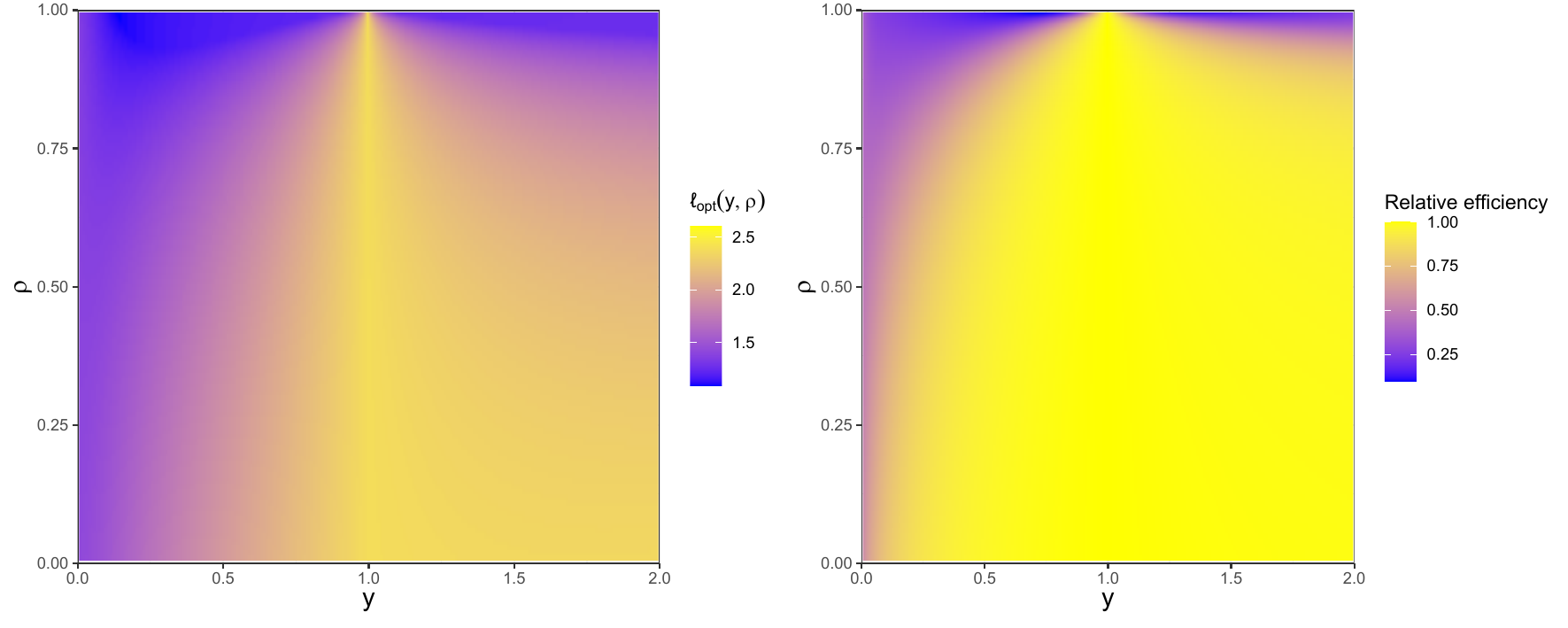}
    \caption{\textbf{Left:} Heatmap of optimal step size $\ell_\textup{opt}(y,\rho)$ for the optimal Markovian coupling. \textbf{Right:} Efficiency of the GCRN coupling relative to that of the optimal Markovian coupling, at the point-wise optimal step sizes.}
    \label{fig:scaling-optimal}
\end{figure}

We consider the optimal scaling of the optimal Markovian coupling. As in Section~\ref{sec:dim-scaling-stdgauss}, we optimize for contraction in the drift $\bar{b}(\cdot)$ corresponding to the squared Euclidean distance between the chains, parametrizing the state as $(x,y,\rho)$ where $\rho \in [-1,1]$ denotes cosine similarity. Figure~\ref{fig:scaling-optimal} (left) shows the point-wise optimal scaling $\ell_\textup{opt}(y, \rho)$, having fixed $x = 1$ so as to emulate a lagged coupling with a large lag parameter $L$.

Figure~\ref{fig:scaling-optimal} (right) shows the relative efficiency of the GCRN coupling against that of the optimal Markovian one. For each coupling, we chose the step size optimally at each $(y,\rho)$; the relative efficiency was then computed as the ratio of the drifts  $\bar{b}_\textup{gcrn}(\cdot) / \bar{b}_\textup{opt}(\cdot)$ at these optimal step sizes. We see that the GCRN coupling loses very little efficiency over most of the range.

\subsection{Preconditioning} \label{app:precondition}

Linear preconditioning can often speed up computations in practical MCMC. For a RWM algorithm with Gaussian proposals, preconditioning is equivalent to letting the proposals be of the form $x + L Z_x \sim \mathcal{N}(x, \Sigma)$ with $Z_x \sim \mathcal{N}_d(0_d,I_d)$ and $LL^\top = \Sigma$. (For instance, $L$ could be the Cholesky factor of $\Sigma$.) We describe here how couplings of proposals $(x + L Z_x, y + L Z_y)$ should be implemented. 

\paragraph{GCRN coupling.} The natural extension to the GCRN coupling of the RWM should account for first-order variation in the log acceptance ratio of, whose Taylor expansion is
\begin{equation*}
    V(x+LZ) - V(x) = Z^\top \{L^\top \nabla V(x)\} + Z^\top  \{L^\top \nabla^2 V(\bar x) L\} Z, 
\end{equation*}
where $V = \log \pi$, and $\bar x$ is on the segment from $x$ to $x+LZ$. It follows that the GCRN coupling should be
\begin{equation*}
    Z_x = Z - (n_x^{\top}Z)n_x + Z_1n_x, \quad Z_y = Z - (n_y^{\top}Z)n_y + Z_1n_y,
\end{equation*}
when preconditioning is used, where: $n_x = \norm\{L^\top \nabla V(x)\}$ and $n_y = \norm\{L^\top \nabla V(x)\}$; $\norm(x) = x/ \lVert x \rVert$; $Z_1\sim \mathcal{N}_1(0,1)$ and $Z\sim \mathcal{N}_d(0_d,I_d)$ are independent. In effect, this amounts to preconditioning the logarithmic gradient by $L^\top$.

One appeal of the GCRN coupling is that it can be straightforwardly adapted to chains which use position-dependent preconditioning $L(x)$, for instance by changing $L \leftarrow L(x)$ in the coupling above. This can lead to effective coupling strategies, as we have demonstrated in the case of the Hug kernel in Section~\ref{sec:hh}.

\paragraph{Reflection (-maximal) coupling.} The reflection coupling should maximize the variation of $L(Z_x - Z_y)$ in the direction of $(x-y)$, while minimizing the variation in all other directions. This can be achieved through the coupling
\begin{equation} \label{eqn:refl-precondition}
    Z_y = Z_x - 2(e^{\top}Z_x)e,
\end{equation}
where $e = \norm(L^{-1}(x-y))$. Following \citet[Section~4.1]{jacob2020unbiased}, one appeal of this coupling is that can be modified to allow for the proposals to be identical with maximal probability. This, of course, allows the chains to coalesce. Let $s(\cdot)$ be the density function of a $\mathcal{N}_d(0_d,I_d)$ variate and let $z = L^{-1}(x-y)$. The resulting \emph{reflection-maximal} coupling sets $Z_y = Z_x - z$ with probability $s(x-z)/s(x)$, and otherwise employs the reflection move~\eqref{eqn:refl-precondition}.

The reflection-maximal coupling furthermore applies to any pair of Gaussians $\mathcal{N}(x, \Sigma)$ and $\mathcal{N}(y, \Sigma)$ with the same covariance matrix. We use it to couple MALA proposals in the experiment of Section~\ref{sec:rwm-vs-mala}.

\section{Further details on the numerical experiments} \label{app:computation}

All computations were carried out in R \citep{r}; runtime-critical components were written in C++. Code to reproduce the numerical experiments can be found at \url{https://github.com/tamaspapp/rwmcouplings}. Part of the binary regression experiments were run on up to 112 processors of a computing cluster. All other experiments were run on a 2019-era Lenovo T490s laptop with 8 processors. (The processor counts include hyper-threading.)

\subsection{Experiments with standard Gaussian targets}\label{sec:ode-numerical}

\paragraph{Solving the ODEs} We solve the limiting ODEs numerically using \texttt{deSolve::ode} \citep{deSolve}. To calculate the drifts in Proposition~\ref{prop:drift-stdgauss}, we must compute the expectation
\begin{equation*}
    g(x,y,\rho) =  \mathbb{E}_{(Z_1, Z_2) \sim \bvn(\rho)} \left[1\land e^{\ell x^{1/2}Z_1 - \ell^2/2}\land e^{\ell y^{1/2}Z_2 - \ell^2/2}\right].
\end{equation*}
For $\rho\in(0,1)$, we evaluate $g(\cdot)$ in terms of numerically tractable quantities in Lemma~\ref{lemma:numerical-integral} below. We require the bivariate normal probabilities
\begin{align*}
	\overline{\bvn}(p,q \mid \rho) &:= \mathbb{P}_{(Z_1, Z_2) \sim \bvn(\rho)}(Z_1 \le p, Z_2 \le q),\\
	\underline{\bvn}(p,q\mid\rho) &:= \mathbb{P}_{(Z_1, Z_2) \sim \bvn(\rho)}(Z_1 \ge p, Z_2 \ge q),
\end{align*}
which we compute using \texttt{mvtnorm::pmvnorm} \citep{mvtnorm}. For $\rho = 1$, we evaluate $g(\cdot)$ using the expression in Lemma~\ref{lemma:gaussian-integrals} below instead.

\begin{lemma}\label{lemma:numerical-integral}
When $\rho <1$ it holds that
\begin{equation*}
    g(x,y,\rho) = \underline{\bvn}\left(\frac{\ell}{2x^{1/2}}, \frac{\ell}{2y^{1/2}} \Bigm\lvert \rho \right) + f(x,y,\rho) + f(y,x,\rho),
\end{equation*}
where:
\begin{equation*}
\begin{aligned}
	f(x,y,\rho) &= e^{\ell^2(x-1)/2}\overline{\bvn}\left(\frac{b \ell x^{1/2}}{\sqrt{1+b^2}}, U \Bigm\lvert -\frac{b}{\sqrt{1+b^2}}\right),\\
	b &= -\frac{(x/y)^{1/2} -  \rho}{\sqrt{1- \rho^2}}, \quad U = \frac{\ell}{2x^{1/2}} - \ell x^{1/2}.
\end{aligned}
\end{equation*}
\end{lemma}

\paragraph{Faster simulation of high-dimensional chains} In our experiments (Figure~\ref{fig:spherical-squaredist}), we simulated the full coupled chains $(X_t,Y_t)_{t \ge 0}$ as this was fast even in our considered dimension $d = 2{,}000$. We however note that it is possible to reduce the computation time by a factor $\bigO(d)$ by simulating the Markov process $(\lVert X_t \rVert^2, \lVert Y_t \rVert^2, \lVert X_t - Y_t\rVert^2)_{t \ge 0}$ directly.

\subsubsection{Proof of Lemma~\ref{lemma:numerical-integral}}
Write $g(x,y,\rho) = \mathbb{E} \left[\exp(0\land A \land B)\right]$, where $A = \ell x^{1/2}Z_1 - \ell^2/2$ and $B = \ell y^{1/2}Z_2 -\ell^2/2$. Partition the expectation according to the events:
\begin{equation*}
\{A \ge 0, B \ge 0\}, \quad \{A<0, A<B\}, \quad \{A<0, A=B\}, \quad \{A<0, A>B\}.
\end{equation*}
Using that $\mathbb{P}(A=B) = 0$, we have that
\begin{align*}
    g(x,y,\rho) 
    &= \mathbb{P}(A\ge 0, B \ge 0) + \mathbb{E}\left[ \mathbbm{1}_{\{A<0\}} \mathbbm{1}_{\{A<B\}} \exp(A) \right] + \mathbb{E}\left[ \mathbbm{1}_{\{B<0\}} \mathbbm{1}_{\{B<A\}} \exp(B)\right]\\
    &= \mathbb{P}\left(Z_1 \ge \frac{\ell}{2x^{1/2}}, Z_2 \ge \frac{\ell}{2y^{1/2}}\right) + f(x,y,\rho) + f(y,x,\rho),
\end{align*}
where $f(\cdot)$ is defined as
\begin{equation*}
    f(x,y,\rho) := \mathbb{E}\left[\mathbbm{1}_{\{\ell x^{1/2}Z_1 < \ell ^2/2\}} \mathbbm{1}_{\{\ell x^{1/2}Z_1 < \ell y^{1/2}Z_2\}} e^{\ell x^{1/2}Z_1 - \ell^2/2}\right].
\end{equation*}

To express $f(\cdot)$ in terms of tractable quantities, we expand $Z_2 = \rho Z_1 + \sqrt{1- \rho^2}Z_*$ where $Z_* \sim \mathcal{N}_1(0,1)$ is independent of $Z_1$. Collecting all terms in the integrand of $f(\cdot)$ that depend on $Z_*$ and integrating them out, we have the expression
\begin{equation*}
\mathbb{E}_{Z_*} \left[ \mathbbm{1}_{\{\ell x^{1/2}Z_1 < \ell y^{1/2}Z_2\}} \right] = \mathbb{P}_{Z_*}\left((x/y)^{1/2}Z_1 < \rho Z_1 + \sqrt{1- \rho^2}Z_*\right) =: \Phi (bZ_1),
\end{equation*}
where $b := - \{(x/y)^{1/2} -  \rho\}/ \sqrt{1- \rho^2}$. It immediately follows that
\begin{align*}
    f(x,y,\rho) 
    &= \mathbb{E}\left[\mathbbm{1}_{\{\ell x^{1/2}Z_1 < \ell^2/2\}}  \Phi (bZ_1) e^{\ell x^{1/2}Z_1 - \ell^2/2} \right]\\
    &= e^{\ell^2(x-1)/2} \mathbb{E}\left[\mathbbm{1}_{\{Z_1 < \ell/(2x^{1/2})\}}  \Phi (bZ_1) e^{
\ell x^{1/2}Z_1 - \ell^2x/2} \right]\\
    &= e^{\ell^2(x-1)/2} \int_{-\infty}^{\ell/(2x^{1/2})} \Phi(bz) e^{\ell x^{1/2}z - \ell^2x/2}\phi(z)\mathrm{d}z\\
    &= e^{\ell^2(x-1)/2} \int_{-\infty}^{\ell/(2x^{1/2})}\Phi(bz) \phi \left(z - \ell x^{1/2}\right) \mathrm{d}z\\
    &= e^{\ell^2(x-1)/2} \int_{-\infty}^{\ell/(2x^{1/2})- \ell x^{1/2}} \Phi\left(b\left(z + \ell x^{1/2}\right)\right) \phi(z) \mathrm{d}z\\
    &=: e^{\ell^2(x-1)/2} \int_{-\infty}^{U} \Phi(a + bz) \phi(z) \mathrm{d}z \\
    &= e^{\ell^2(x-1)/2} \overline{\bvn}\left( \frac{a}{\sqrt{1+b^2}}, U \Bigm\lvert -\frac{b}{\sqrt{1+b^2}} \right), \tag{\citealp[Eqn. (10,010.1)]{owen1980table}}
\end{align*}
where: (i) we have used the identity $e^{\ell x^{1/2}z - \ell^2x/2}\phi(z) = \phi(z - \ell x^{1/2})$ to obtain the fourth line; (ii) we have defined 
\begin{equation*}
	a := b\ell x^{1/2}, \quad U := \frac{\ell}{2x^{1/2}} - \ell x^{1/2}.
\end{equation*}
Substituing this expression into $g(\cdot)$ completes the proof.

\subsection{Experiments with elliptical Gaussian targets}

We evaluate the eccentricties of targets used in Section~\ref{sec:elliptgauss}.

\paragraph{AR(1) process} An AR(1) process with unit noise increments and correlation $\rho$ has a covariance with entries  $\Sigma^{(d)}_{ij} = \rho^{|i-j|}$ for all $(i,j)$. This is a Kac-Murdock-Szeg\H{o} matrix \citep[see e.g.][]{trench1999asymptotic}. It holds that
\begin{equation*}
\frac{1}{d}\tr\left(\Sigma^{(d)}\right) = 1, \quad \lim_{d\to\infty}\frac{1}{d}\tr\left(\big(\Sigma^{(d)}\big)^{-1}\right) = \frac{1+\rho^2}{1-\rho^2},
\end{equation*}
so that the limiting eccentricity is $\ecc = (1+\rho^2)/(1-\rho^2)$.

\paragraph{Chi-square eigenvalues} Let $\lambda \sim \chi^2_\nu$. It holds that $\mathbb{E}[\lambda] = \nu$ and (if $\nu>2$) that $\mathbb{E}\left[ \lambda^{-1} \right] = 1/(\nu -2)$. If the eigenvalues $\lambda_1,\dots,\lambda_d$ of the covariance matrix $\Sigma^{(d)}$ are sampled i.i.d from $\chi^2_\nu$, it holds that 
\begin{equation*}
\lim_{d \to \infty}\frac{1}{d}\tr\left(\Sigma^{(d)}\right) = \mathbb{E}[\lambda] = \nu, \quad \lim_{d\to\infty}\frac{1}{d}\tr\left(\big(\Sigma^{(d)}\big)^{-1}\right) =\mathbb{E}\left[ \lambda^{-1} \right] = \frac{1}{\nu - 2},
\end{equation*}
so that the limiting eccentricity is $\ecc = \nu/(\nu - 2)$.

\subsection{Experiments with stochastic volatility model} \label{app:svm}

\paragraph{Posterior log-density and score} The posterior log-density of the stochastic volatility model is
\begin{equation*}
    \log\pi(x_{1:d} \mid y_{1:d}) = -\frac{1}{2}\left(\sum_{t=1}^d x_t + \frac{1}{\beta^2} \sum_{t=1}^d y_t^2 \exp(-x_t) + \frac{1}{\sigma^2}\sum_{t = 1}^{d - 1}(\varphi x_t - x_{t + 1})^2 + \frac{1 - \varphi^2}{\sigma^2} x_1^2 \right) + \textup{const},
\end{equation*}
where ``const" is an offset constant in $x_{1:d}$. The score has entries
\begin{equation*}
\frac{\partial \log \pi}{\partial x_t} = -\frac{1}{2} + \frac{1}{2 \beta^2}y_t^2 \exp(-x_t) - \frac{\varphi}{\sigma^2}(\varphi x_t - x_{t + 1}) \mathbbm{1}_{\{t \ne d\}} - \frac{1}{\sigma^2}(x_t - \varphi x_{t-1}) \mathbbm{1}_{\{t \ne 1\}} - \frac{1 - \varphi^2}{ \sigma^2}x_1 \mathbbm{1}_{\{t = 1\}}
\end{equation*}
for all $t \in \{1,2,\dots,d\}$.

\paragraph{Laplace approximation} We use the LBFGS optimizer of the R function \texttt{optim} to compute the Laplace approximation. The optimization is initialized at a single draw from the prior.

\subsubsection{Rate of convergence of the RWM}

\begin{figure}
    \centering
    \includegraphics[width=0.8\textwidth]{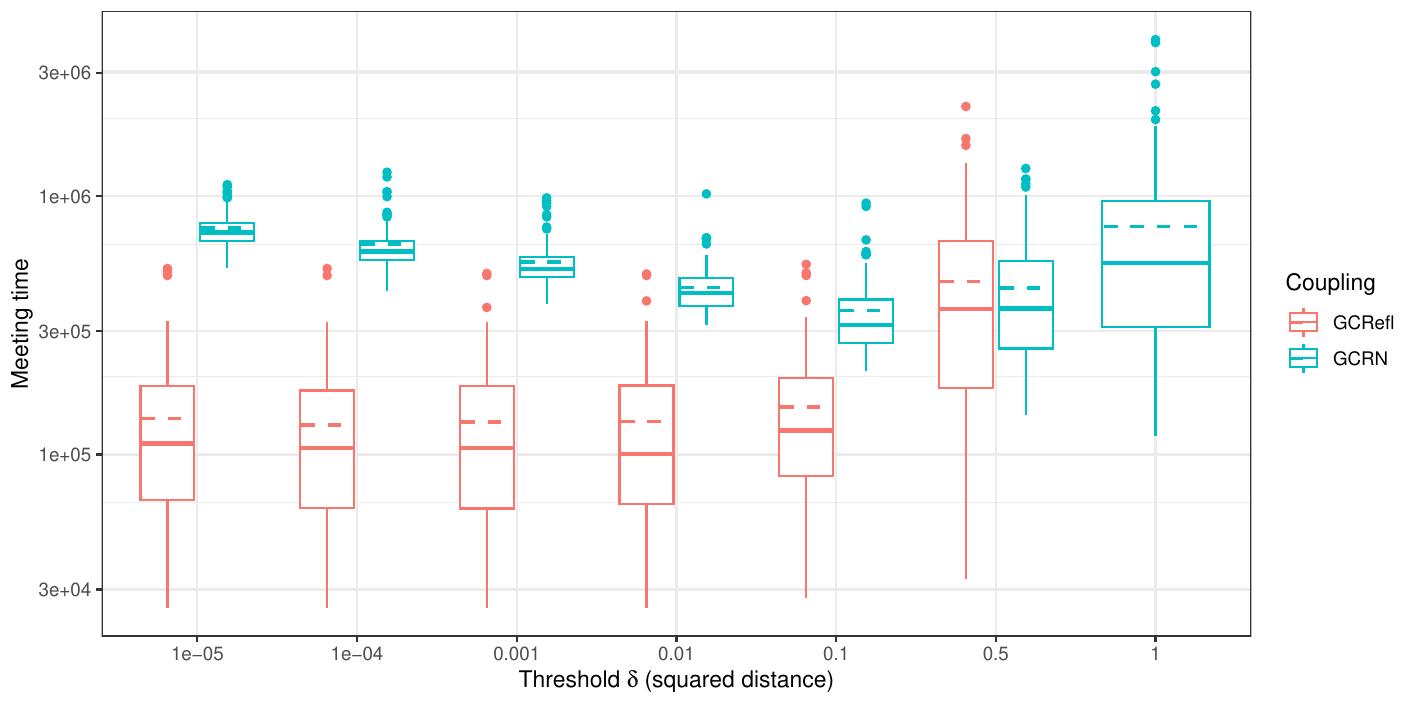}
    \caption{\textbf{Stochastic volatility model.} Box plots of $R = 100$ meeting times for various thresholds $\delta$ and two-scale RWM couplings. The dashed lines denote the sample means.}
    \label{fig:app-rwm-threshold}
\end{figure}

\paragraph[Choice of threshold delta in two-scale couplings]{Choice of threshold $\delta$ in two-scale couplings} We search for a sensible threshold $\lVert X_t - Y_t \rVert^2 = \delta$ over a coarse a grid $\{10^{-4},10^{-3},10^{-2}, 10^{-1}, 0.5, 1\}$. For each coupling and choice of threshold, we measured the meeting time between chains initialized at independent draws from the Laplace approximation (with no lag, i.e. $L=0$) and repeated this 100 times. The results are displayed in Figure~\ref{fig:app-rwm-threshold}. Compared to GCRN, GCRefl prefers a smaller $\delta$, is less sensitive to values smaller than optimal but is much more sensitive to values larger than optimal. (We omitted $\delta=1$ for the GCRefl coupling as we found the meeting times in preliminary runs to be exceedingly large.)

With the above grid of values and when using a maximal coupling, the one-step probability of coalescing the proposals at $\lVert X_t - Y_t \rVert^2 = \delta$ would be $\{0.81,0.45,0.02,\dots\}$, where dots denote probabilities under $10^{-13}$. The fact that $\delta = 0.1$ performs best for the GCRN coupling, even though the chance of meeting is extremely small at this threshold, indicates that the extra variability induced by the reflection coupling contracts the chains, on average, faster than exponential rate $\bigO(h^2)$ of GCRN.  

\subsubsection{Bias of Laplace approximation} \label{sec:svm-bias}

\paragraph{Lower bound of \cite{gelbrich1990on}} We have the explicit expression
\begin{equation*}
\wass{2}^2(\mathcal{N}_d(\hat\mu, \hat\Sigma), \mathcal{N}_d(\mu, \Sigma)) = \lVert \hat\mu - \mu\rVert^2 
+ \tr(\Sigma) + \tr(\hat\Sigma)- 2\tr \Big ( \big (\Sigma^{1/2} \hat\Sigma \Sigma^{1/2} \big)^{1/2}\Big).
\end{equation*}
We estimated the posterior mean and covariance $(\mu,\Sigma)$ by averaging over $R=50$ independent Hug and Hop chains (see Appendix~\ref{app:hh}); each chain was run for $50{,}000$ iterations and was warm-started from an independent draw from the Laplace approximation. We followed the guidelines in \cite{ludkin2022hug} and tuned Hug to $(T,B)=(0.5,10)$ integration time and bounce count (for an acceptance rate of $79\%$) and Hop to $(\lambda,\kappa)=(20,1)$ for an acceptance rate of $40\%$.

Jackknife \citep{efron1981jackknife} bias and standard error estimates suggested that the mean-squared error of our estimate of $\wass{2}^2(\mathcal{N}_d(\hat\mu, \hat\Sigma), \mathcal{N}_d(\mu, \Sigma))$ was small. We note that the bootstrap is known to be consistent in our case, see \citet[Section 2.3]{rippl2016limit}.

\paragraph{Total variation distance bound} The total variation distance bound was computed with a two-scale GCRefl coupling with the same parameter settings $(h, \delta)$ 
 as in the experiment of Section~\ref{sec:numeric-conv}.

\subsubsection{Convergence of Hug and Hop} \label{app:hh}

\begin{algorithm}
  \caption{Hug kernel \label{alg:hug}}
  \begin{algorithmic}[1]
    \Require{Target density $\pi:\mathbb{R}^d \to \mathbb{R}$, score $\nabla\log\pi:\mathbb{R}^d \to \mathbb{R}^d$, current state $x \in \mathbb{R}^d$.}
    \Require{Step size $\delta >0$ and bounce count $B\in \mathbb{N}$.}

\Statex
\Function{Bounce}{$x,z,\delta,B$}
      \For{$i = 1,\dots,B$}
	 \State $x \leftarrow x  + (\delta/2) z$
	 \State $z \leftarrow z - 2  (g(x)^\top z) z$ 
        \State $x \leftarrow x  + (\delta/2) z$
    \EndFor
	\State \Return $(x,z)$
\EndFunction

\Statex
    \State Sample $z \sim \mathcal{N}_d(0_d, I_d)$. \Comment{Generate proposal}
    \State Propose $(X,Z) \leftarrow \textsc{Bounce}(x,z,\delta,B)$. 
    \State Sample $U \sim \textup{Unif}(0,1)$.  \Comment{Metropolis filter}
     \If{$\log U < \left( \log \pi(X) - \log \pi(x) \right)$} 
          \State $x \leftarrow X$
    \EndIf
	\State \Return{$x$}
  \end{algorithmic}
\end{algorithm} 

\begin{algorithm}
  \caption{Hamiltonian Monte Carlo kernel \label{alg:hmc}}
  \begin{algorithmic}[1]
    \Require{Target density $\pi:\mathbb{R}^d \to \mathbb{R}$ and score $\nabla\log\pi:\mathbb{R}^d \to \mathbb{R}^d$, current state $x \in \mathbb{R}^d$.}
    \Require{Leapfrog step size $\delta >0$ and iteration count $B\in \mathbb{N}$.}

\Statex
\Function{Leapfrog}{$x,z,\delta,B$}
      \For{$i = 1,\dots,B$}
	 \State $z \leftarrow z  + (\delta/2) \nabla\log\pi(x)$
	 \State $x \leftarrow z + \delta z$
        \State $z \leftarrow z + (\delta/2) \nabla\log\pi(x)$
    \EndFor
	\State \Return $(x,z)$
\EndFunction

\Statex
    \State Sample $z \sim \mathcal{N}_d(0_d, I_d)$.  \Comment{Generate proposal}
    \State Propose $(X,Z) \leftarrow \textsc{Leapfrog}(x,z,\delta,B)$. 
    \State Sample $U \sim \textup{Unif}(0,1)$. \Comment{Metropolis filter}
     \If{$\log U < \left( \log \pi(X) - \log \pi(x) - \lVert Z \rVert^2/2 + \lVert z \rVert^2/2 \right)$} 
          \State $x \leftarrow X$
    \EndIf
	\State \Return{$x$}
  \end{algorithmic}
\end{algorithm} 

\paragraph{Additional algorithmic descriptions} The Hug kernel is described in Algorithm~\ref{alg:hug} (see also \citealp[Algorithm 1]{ludkin2022hug}) and the Hamiltonian Monte Carlo (HMC) kernel is described in Algorithm~\ref{alg:hmc}. The synchronous coupling of Hug and HMC proceeds by using identical initial momenta (Line 9) and identical acceptance uniforms (Line 11).

\paragraph{Tuning Hop} We follow the guidelines of \cite{ludkin2022hug} and first tune Hop independently of Hug, choosing $(\lambda, \kappa) = (20, 1)$ for an acceptance rate of $40\%$.

\paragraph{Tuning HMC and Hug for contractivity} We assess the contractivity of coupled HMC and coupled Hug and Hop (H\&H) as follows: we independently initialize a pair of coupled chains from the Laplace approximation and track the squared distance between the chains for a fixed number of iterations. We make no attempt to coalesce the chains: HMC is not mixed with the RWM, and for Hop we fix $\delta_\textup{hop} = 0$ so that only the GCRN coupling is applied.

\begin{figure}[tb]
    \centering
    \includegraphics[width=\textwidth]{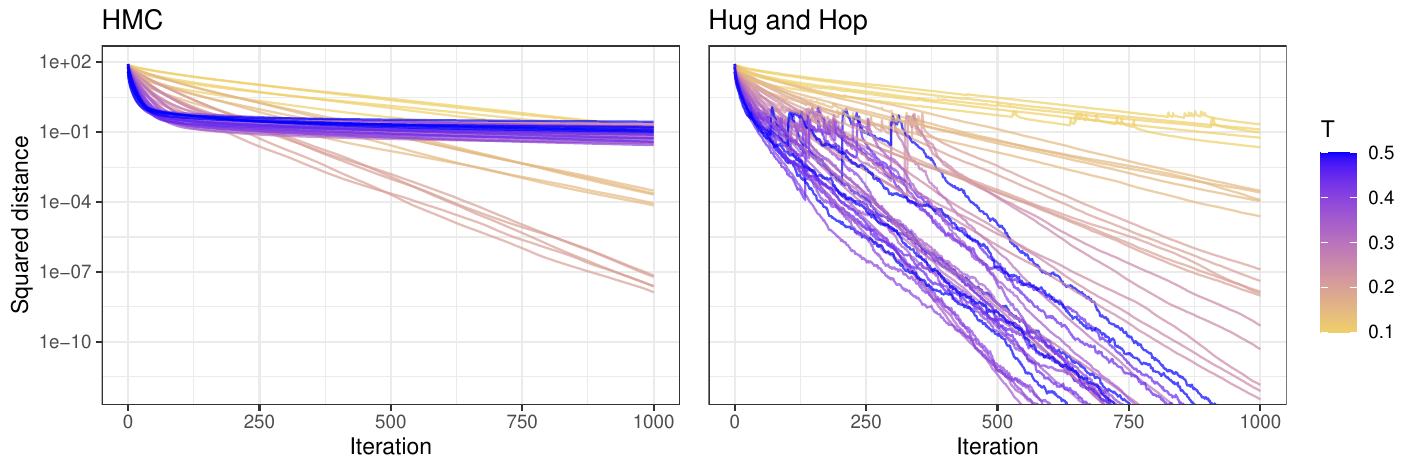}
    \caption{\textbf{Stochastic volatility model.} Contractivity of HMC, and Hug and Hop, varying the integration time $T$ and using a fine integration grid.}
    \label{fig:hug-hmc-integration-time}
\end{figure}

HMC and Hug require the tuning of two parameters $(T,B)$, where $T = \ecc B$ represents the integration time, $B$ is either the number of leafrog steps (for HMC) or the number of ``bounces" (for Hug), and $\ecc$ represents the step size used to generate the respective discretized proposals. First, we look at the impact of the integration time on the contractivity, fixing a small $\ecc = 10^{-3}$ and varying $T\in\{0.1, 0.15, \dots, 0.5\}$. Figure~\ref{fig:hug-hmc-integration-time} displays trace plots of 5 replicates for each algorithm and parameter setting. HMC requires a short integration time in order to be contractive, suffering a sharp phase transition from $T = 0.2$ to $T=0.25$. In contrast, Hug is more robust with respect to this tuning parameter and benefits from longer integration times than HMC.

\begin{figure}[tb]
    \centering
    \includegraphics[width=\textwidth]{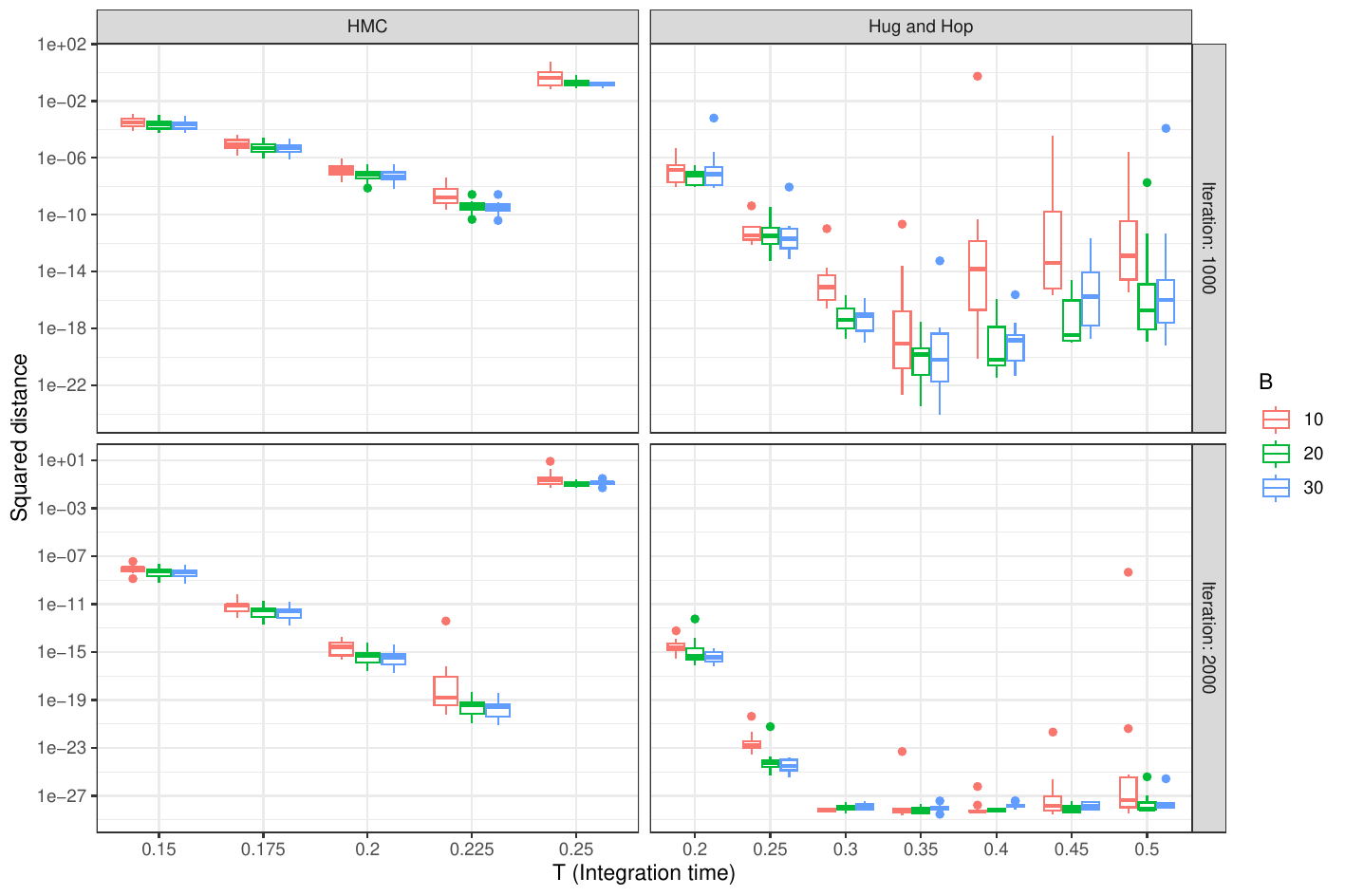}
    \caption{\textbf{Stochastic volatility model.} Contractivity of HMC, and Hug and Hop, varying the integration time $T$ and leapfrog steps/bounces $B$.}
    \label{fig:hug-hmc-paramtuning}
\end{figure}

Next, we assess which configurations are contractive among the grid of parameters
\begin{equation*}
	T_\textup{hug} \in \{0.2,0.25,\dots,0.5\},\quad T_\textup{hmc} \in \{0.15,0.175,\dots,0.25\}, \quad B \in \{10, 20, 30\}.
\end{equation*}
Figure~\ref{fig:hug-hmc-paramtuning} displays box plots of the squared distance between chains after 1{,}000 and 2{,}000 iterations, from 20 replicates each. We select from among the more efficient contractive configurations $(T_\textup{hug},B_\textup{hug}) = (0.35,10)$ and $(T_\textup{hmc},B_\textup{hmc}) = (0.225,10)$, for approximately equal cost per iteration for both H\&H and HMC. The configurations correspond to acceptance rates of $\alpha_\textup{hug} = 90\%$ and $\alpha_\textup{hmc} = 78\%$.

\paragraph{RWM and HMC mixture} With probability $\gamma = 0.05$ (the default in \citealp{heng2019unbiased}) we switch from coupled HMC to coupled RWM kernels in order to allow the chains to meet. For the RWM, we use the parameter tuning and two-scale GCRN coupling of Section~\ref{sec:numeric-conv}, which we know perform well. We expect these settings to be close to optimal as \cite{heng2019unbiased} demonstrate that the performance of the overall algorithm is insensitive to the mixture probability $\gamma$ and to the tuning of the RWM.

\begin{figure}[tb]
    \centering
    \includegraphics[width=0.7\textwidth]{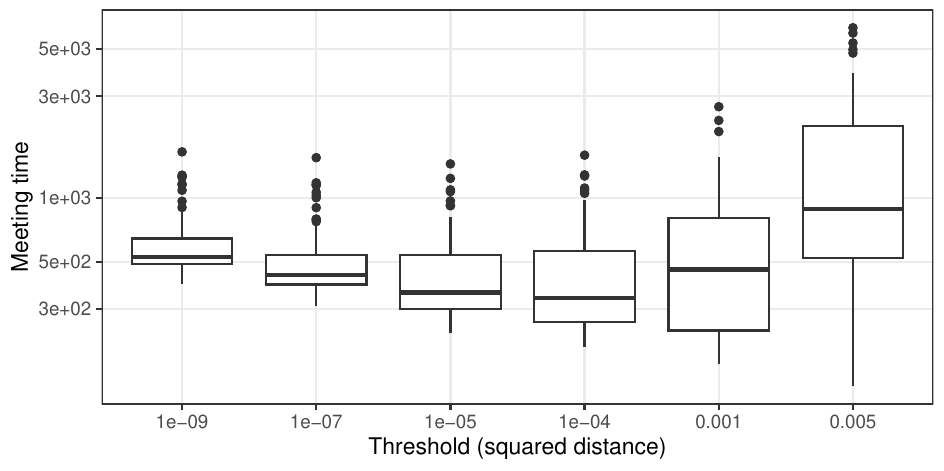}
    \caption{\textbf{Stochastic volatility model.} Box plots of the meeting times for various choices of the threshold $\delta$ in the two-scale Hug and Hop coupling.}
    \label{fig:hh-thresh}
\end{figure}

\paragraph{Choosing $\delta_\textup{hop}$ for two-scale Hop coupling} We sweep over a grid $\delta\in[10^{-9}, 5 \times 10^{-3}]$ in search of a sensible threshold for the two-scale Hop coupling. Figure~\ref{fig:hh-thresh} displays box plots of 100 replicates for each $\delta$. The meeting times are insensitive to the theshold as long as the chains have a high probability of meeting when $\lVert X_t - Y_t \rVert^2 < \delta$; we select $\delta_\textup{hop} = 10^{-5}$.

\begin{figure}[tb]
    \centering
    \includegraphics[width=0.4\textwidth]{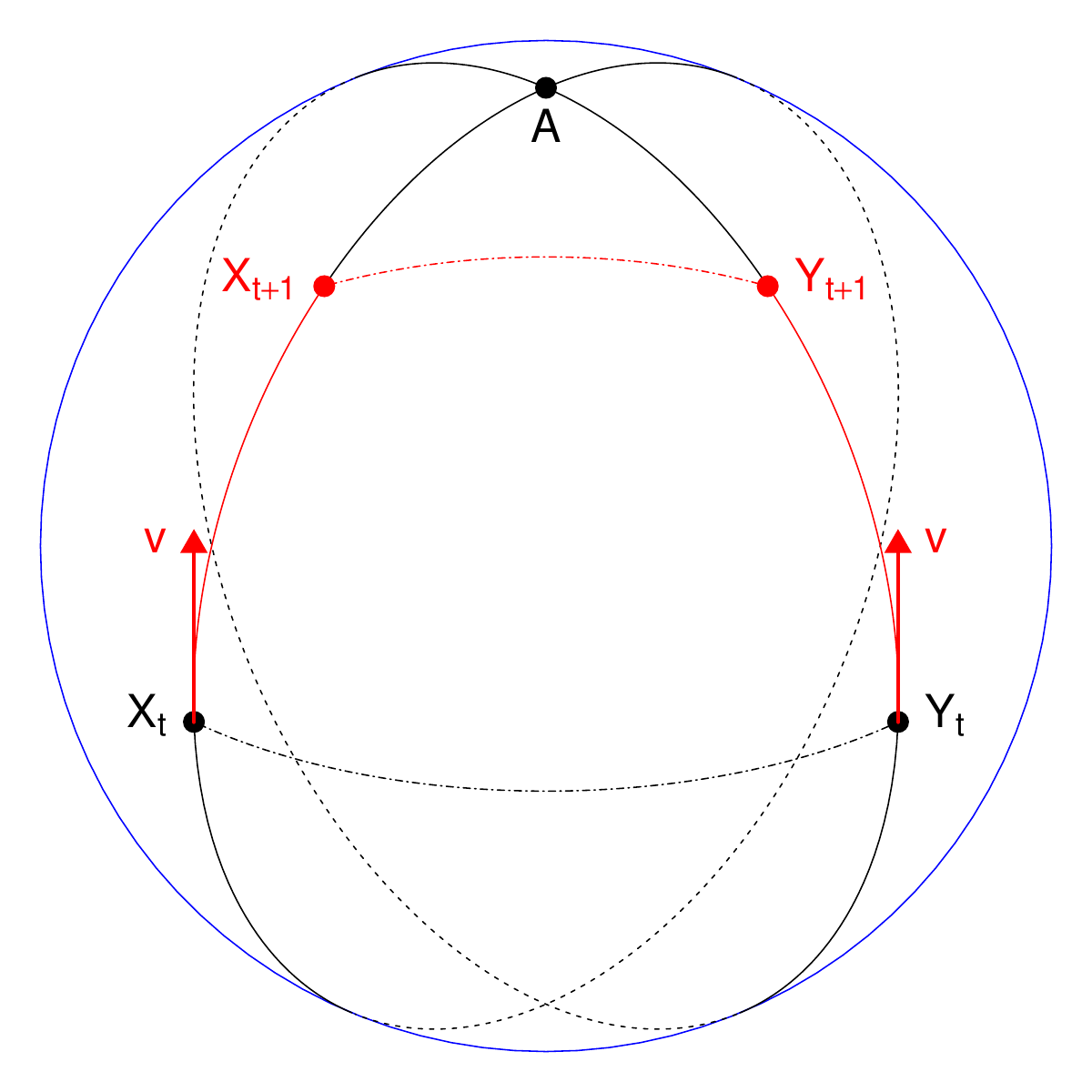}
    \caption{Visualisation of the contractive behaviour of the synchronous Hug coupling. For a spherical target of increasing dimension, the pictured scenario occurs with probability approaching 1.}
    \label{fig:hug-crn-intuition}
\end{figure}

\paragraph{Contractivity of synchonous Hug coupling in high dimensions} We illustrate our intuition as to why the CRN coupling of Hug chains is contractive in high dimensions in Figure~\ref{fig:hug-crn-intuition}. For clarity of exposition, we assume that the target is spherically symmetric (so its level sets are hyperspheres), that the Hug dynamics are exact (so they traverse great circles with constant speed), and that the two CRN-coupled Hug chains start at $(X_t,Y_t)$ which lie on the same level set. In high dimensions, the velocity $v$ is essentially always orthogonal to the great circle traversing $(X_t,Y_t)$. As a consequence, essentially almost sure contraction is achieved by synchonous movement towards one of the two intersections of the two great circles which support the coupled Hug paths (in Figure~\ref{fig:hug-crn-intuition}, movement is towards the point A or its antipode).

Intriguingly, as the coupled paths are farthest away at $(X_t,Y_t)$, the chains contract irrespectively of the length of the integration time. This is consistent with our empirical observation that the synchonous Hug coupling is contractive even if the integration time is long (see Figure~\ref{fig:hug-hmc-integration-time}).

\paragraph{Alternative coupling strategies for Hop} We have also experimented with replacing the GCRN coupling of Hop with a CRN coupling. The CRN coupling of Hop was significantly less contractive, and produced significantly larger meeting times. This was primarily due to its inability to synchronize acceptance events with the same frequency as GCRN.

\subsection{Experiments with binary regression} \label{app:binreg}

\paragraph{Posterior log-density and score}
Let $\mathbf{X} \in \mathbb{R}^{n\times d}$ be the design matrix with rows $\{\mathbf{x}_1,\dots,\mathbf{x}_n\} \in \mathbb{R}^{d}$  and let $\mathbf{y} \in \{0,1\}^n$ be the response. For all $x \in \mathbb{R}^d$, the posterior log-density of a binary regression model with a $\mathcal{N}_d(0_d,\lambda^2 I_d)$ prior is
\begin{equation*}
    \log\pi(x \mid \mathbf{X}, \mathbf{y}) = \sum_{i=1}^n \log F \left( (2\mathbf{y} _i - 1) \mathbf{x}_i^\top x \right) - \frac{1}{2\lambda^2}\lVert x \rVert^2 + \textup{const},
\end{equation*}
where: $F:\mathbb{R} \to (0,1)$ is a cumulative distribution function; ``const" is an offset constant in $x$. The score is
\begin{equation*}
\nabla \log \pi(x \mid \mathbf{X}, \mathbf{y}) = \sum_{i=1}^n (\log F)' \left( (2\mathbf{y} _i - 1) \mathbf{x}_i^\top x \right) (2\mathbf{y} _i - 1) \mathbf{x}_i - \frac{1}{\lambda^2}x.
\end{equation*}
We consider the logistic case, in which case $\log F(z) = - \log(1+ e^{-z})$ is the logistic log-CDF.

\subsubsection{Choice of two-scale couplings} \label{app:binreg-ablation}

We experimented with two-scale couplings that swapped from a contractive coupling to a reflection-maximal coupling when the chains were close enough. As a default, when coupling proposals $\mathcal{N}(x, PP^\top)$ and $\mathcal{N}(y, PP^\top)$ we specified the swapping rule as $\| P^{-1} (x - y) \|^2 \le \delta^2$ and we varied the value of $\delta > 0$. However, for the RWM using a diagonal preconditioner and the GCRefl coupling, we instead specified the swapping rule as $\| P^{-1} (X-Y) \|^2 \le \delta^2 h^2$, which is adapted to the local scale of the proposal, as we found the performance of the resulting two-scale coupling to be less sensitive to the choice of $\delta$.

\begin{figure}[tb]
    \centering
    \includegraphics[width=\textwidth]{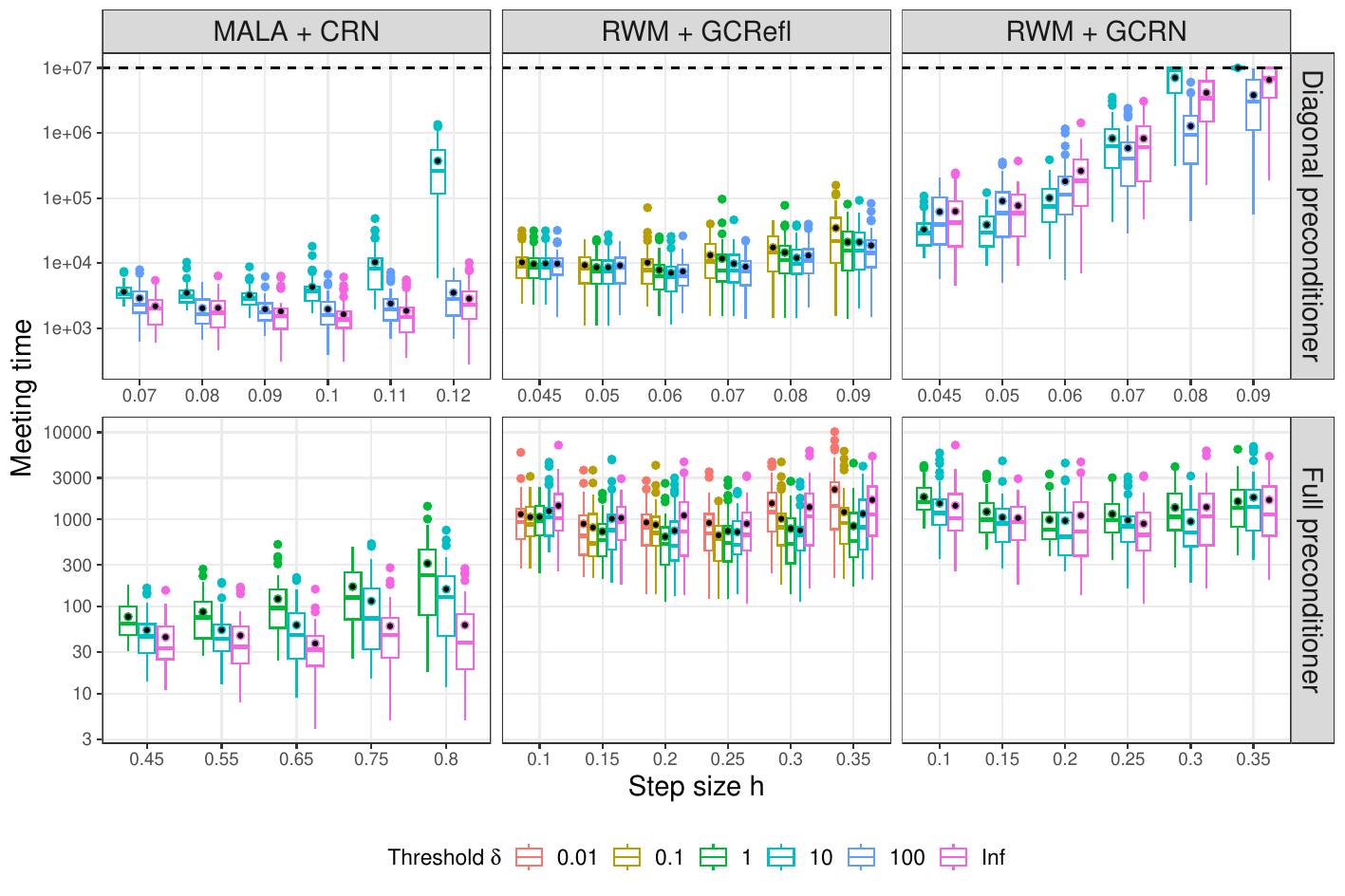}
    \caption{\textbf{Binary regression.} Box plots of meeting times for various algorithms and two-scale couplings, see Appendix~\ref{app:binreg-ablation} for details. Black dots indicate sample means. In the top row, the meeting times were truncated at $T = 10^7$ as indicated by the black dashed line. Note that the definition of $\delta$ may vary between plots.}
    \label{fig:binreg-ablation}
\end{figure}

Figure~\ref{fig:binreg-ablation} displays box plots of $R = 50$ replicates for several configurations; we measured meeting times between coupled chains started independently from a Gaussian approximation of the target. For the RWM, the two-scale GCRefl coupling was consistently the most effective: (1) when using the diagonal preconditioner, this was the only practical coupling out of the ones considered; (2) when using the full preconditioner, GCRefl outperformed both GCRN and the reflection coupling, particularly for larger step sizes. For MALA, the reflection-maximal coupling on its own consistently performed the best, across both types of preconditioning.

\subsubsection{Comparison of RWM and MALA} \label{app:rwm-vs-mala}

\paragraph{Parameters for the main experiment} 

We first approximately sampled from the target using $R = 112{,}000$ optimally-scaled MALA chains, warm-started from a Gaussian approximation to the posterior $\pi$. Discarding burn-in, we used these chains to estimate the posterior mean and covariance to a very low degree of error; these quantities are relevant for the initialization of our experiment and for the asymptotic variance estimator below.

For the diagonal preconditioner, we used the grid of step sizes 
\begin{equation*}
	\begin{aligned}
	h_\textup{rwm} &= \{0.045, 0.05, 0.06, 0.07, 0.08, 0.09\},\\
	h_\textup{mala} &= \{0.07, 0.08, 0.09, 0.1, 0.11, 0.12\},
	\end{aligned}
\end{equation*}
and we ran $R = 112{,}000$ coupled chains to estimate the relevant efficiency metrics for the RWM and MALA algorithms. For the full preconditioner, we used the grid of step sizes 
\begin{equation*}
	\begin{aligned}
	h_\textup{rwm} &= \{0.1, 0.15, 0.2, 0.25, 0.3, 0.35\},\\
	h_\textup{mala} &= \{0.45, 0.55, 0.65, 0.75, 0.8\},
	\end{aligned}
\end{equation*}
and we ran $R = 11{,}200$ coupled chains to estimate the relevant efficiency metrics for the RWM and MALA algorithms.

\begin{figure}[tb]
    \centering
    \includegraphics[width=\textwidth]{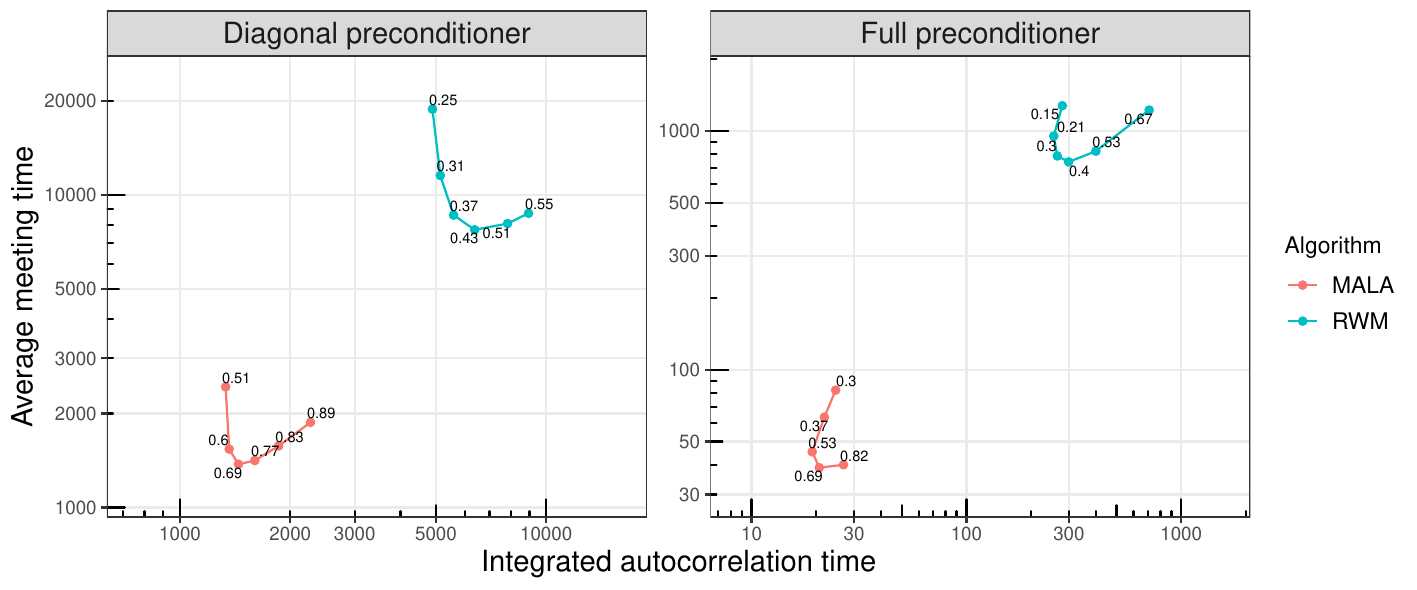}
    \caption{\textbf{Binary regression.} Efficiency comparison of RWM and MALA, varying the step sizes. All estimates are shown with two standard errors. The acceptance rate for each step size is overlaid.}
    \label{fig:binreg-iter}
\end{figure}

\paragraph{Measuring wall-time}

The experiment in the main text was performed on a shared computing cluster; as a consequence, even when using identical random seeds, we observed large variations in wall-time between different runs.
To obtain more reliable measurements, in Figure~\ref{fig:binreg} we instead scaled the meeting times according to the per-iteration cost of a marginal chain.
To account for the gradient evaluations performed by the GCRefl coupling of the RWM, we used the fact that the GCRefl coupling only updates the gradient when a marginal chain accepts, and scaled values according to the the acceptance rate of a single stationary RWM chain.
Note that this is a slight overestimate of the actual computing cost for the RWM, since (1) we combined GCRefl with a reflection-maximal coupling, which does not need gradient evaluations, and (2) one of the coupled chains was initialized near the target mode, so its acceptance rate was somewhat smaller than at stationarity.

\paragraph{Efficiency metrics} We measure the stationary efficiency of the Markov kernel $K$ by the integrated autocorrelation time (IACT) of a test function $f(\cdot)$, 
\begin{equation*}
	\textup{IACT}(K,f) = \frac{\textup{Var}(K,f)}{\textup{Var}_\pi(f(X))},
\end{equation*}
where the denominator is the target variance and where the numerator is the asymptotic variance,
\begin{equation*}
	\textup{Var}(K,f) = \textup{Var}_\pi(f(X_0)) + 2\sum_{t=1}^\infty \textup{Cov}_\pi(f(X_0), f(X_t)),
\end{equation*}
with the subscript $\pi$ indicating a stationary chain ($X_0 \sim \pi$ and $X_{t+1} \sim K(X_t, \cdot)$ for all $t \ge0$).

Figure~\ref{fig:binreg-iter} shows the efficiency comparison of the RWM and MALA in terms of worst IACT (when estimating the regression coefficients; we report the highest coordinate-wise sample average) and average meeting time. We account for wall-time in the main text: we define the time per effective sample as the worst IACT multiplied by the wall-time of one iteration.

\paragraph{On the computing cost of the RWM and MALA}

With our hardware and hand-optimized C++ implementation, a gradient evaluation was as expensive as a log-density evaluation, so one iteration of MALA took roughly twice as much as one of the RWM. We found the log-density and gradient computations to dominate the cost of the MCMC iteration, even when dense preconditioners were used.

\paragraph{Unbiased estimation of the asymptotic variance} Consider the setup of Section~\ref{sec:background} at no lag ($L=0$): we run two coupled chains $(X,Y)$ which marginally evolve under the same Markov kernel $K$, and which meet at some time $\tau$. Suppose furthermore that the chains are initialized \emph{independently} at $X_0 \sim \pi$ and $Y_0 \sim \pi_0$. \cite{douc2023solving} show that \begin{equation*}
V(X,Y) = - \textup{Var}_\pi(f(X)) + 2 \left\{ f(X_0) - \mathbb{E}_\pi[f(X)] \right\} \sum_{t=0}^{\tau - 1}\{f(X_t) - f(Y_t)\}
\end{equation*}
is an unbiased estimator of the asymptotic variance $\textup{Var}(K,f)$. This estimator coincides with ``EPAVE" \citep[Equation~3.2]{douc2023solving} with $t=0$, and is unbiased in our case as the $X$-chain is stationary.

In our experiments, we estimate $\textup{IACT}(K,f)$ near-unbiasedly by $V(X,Y)/\textup{Var}_\pi(f(X))$, up to negligible MCMC errors in computing $\mathbb{E}_\pi[f(X)]$ and $\textup{Var}_\pi(f(X))$. The near-unbiased estimator allows us to exploit parallelism and to investigate the efficiency of the marginal kernel $K$ without running long chains. Furthermore, it does not suffer from the underestimation bias of standard (spectral variance or batch means) estimators of the asymptotic variance (see \cite{vatsjones2022lugsail} and references therein), a bias which we have observed in preliminary experiments with the packages \texttt{coda} \citep{coda} and \texttt{mcmcse} \citep{mcmcse}.

\section{Proofs} \label{sec:proofs}

\subsection{Notation} 
We make the following conventions in the proofs below. For simplicity of notation, we omit most subscripts and superscripts relating to the dimension $d$. The Euclidean norm is denoted by $\lVert \cdot \rVert$. The standard normal probability density function is denoted by $\phi$, the cumulative density function is $\Phi$. $\bvn(\sigma_1^2, \sigma_2^2; \eta)$ denotes the bivariate normal distribution with coordinate-wise variances $(\sigma_1^2, \sigma_2^2)$ and covariance $\eta \in [-\sigma_1\sigma_2,\sigma_1\sigma_2]$. We write $\bvn(\eta)$ instead when both variances are unit ($\sigma_1^2 = \sigma_2^2 = 1$). We use $\implies$ to denote the weak convergence of random variables and stochastic processes.

\subsection{Auxiliary results} \label{app:proofs-aux}

We shall first require a few auxiliary results. Lemma~\ref{lemma:acc-fun-lipchitz} below recalls the Lipschitz property of the function $f(x) = 1 \land \exp(x)$ which appears in the Metropolis-Hastings acceptance ratio. Lemma~\ref{lemma:cross-term-bound} below extends Lemma~\ref{lemma:acc-fun-lipchitz}.

\begin{lemma}[{\citealp[Proposition 2.2]{roberts1997weak}}] \label{lemma:acc-fun-lipchitz} 
Let $a \in \mathbb{R}$. For all $x,y \in \mathbb{R}$, it holds that
\begin{equation*}
|1 \land \exp(x) - 1 \land \exp(y)| \le 1 \land |x-y|.
\end{equation*}
\end{lemma}

\begin{lemma} \label{lemma:cross-term-bound}
For all $a,b,x,y,u,v \in \mathbb{R}$, it holds that
\begin{equation*}
|a 1 \land \exp(x)\land \exp(u) - b 1 \land \exp(y)\land \exp(v)| \le |a-b| + |b| (1 \land |x-y| + 1 \land |u-v|).
\end{equation*}
\end{lemma}
\begin{proof}
Firstly, $a 1 \land e^{x\land u} - b 1 \land e^{y\land v} = (a - b) 1 \land e^{x\land u} + b(1 \land e^{x\land u} - 1 \land e^{y\land v})$. By the triangle inequality, it follows that
\begin{equation*}
	|a 1 \land e^{x\land u} - b 1 \land e^{y\land v}| \le |a - b| +b |1 \land e^{x\land u} - 1 \land e^{y\land v}|.
\end{equation*}

Now, by adding and subtracting the cross-term $1 \land e^{y\land u}$, then using the triangle inequality, we have that
\begin{align*}
|1 \land e^{x\land u} - 1 \land e^{y\land v}| 
&= |1 \land e^{x\land u} - 1 \land e^{y\land u}| + |1 \land e^{y\land u} - 1 \land e^{y\land v}| \\
&\le 1 \land |x\land u - y\land u| + 1 \land |y\land u - y\land v|, \tag{Lemma~\ref{lemma:acc-fun-lipchitz}}\\
&\le 1 \land |x-y| + 1 \land |u-v|.
\end{align*}
where in the final line we used that the function $f_z(x) = z \land x$ is 1-Lipschitz, for all $z \in \mathbb{R}.$ This concludes the proof.
\end{proof}

Lemmas~\ref{lemma:conv-1} and~\ref{lemma:conv-2} below are used to express the limiting drifts in Propositions~\ref{prop:drift-stdgauss} and~\ref{prop:drift-elliptgauss}.

\begin{lemma} \label{lemma:conv-1}
Let $(Z_1, Z_2) \sim \bvn(\sigma_1^2,\sigma_2^2,\eta)$, let $A$ be a constant and let $A_d \to A$ in probability as $d \to \infty$. Then, it holds that
\begin{equation*}
    \lim_{d \to \infty} \mathbb{E}\left[ Z_2 1\land \exp(Z_1 - A_d)\right] = \mathbb{E}\left[ Z_2 1\land \exp( Z_1 - A)\right].
\end{equation*}
Furthermore, the convergence is uniform over $(\sigma_1^2, \sigma_2^2, \eta)$ in a compact set provided that $A_d \to A$ uniformly over the same set.
\end{lemma}
\begin{proof}
We shall actually show uniform convergence in $L_2$. Firstly by the Cauchy-Schwarz inequality, we have that
\begin{align*}
\mathbb{E}^2 \left[Z_2 \left(1\land e^{Z_1 - A_d } - 1\land e^{Z_1 - A}\right) \right]
    &\le \mathbb{E} \left[Z_2^2\right] \mathbb{E}\left[\left(1\land e^{Z_1 - A_d} - 1\land e^{Z_1 - A}\right)^2\right] \\
    &\le \mathbb{E} \left[Z_2^2\right] \mathbb{E}\left[1\land(A_d - A)^2 \right] \tag{Lemma~\ref{lemma:acc-fun-lipchitz}}\\
    &= \sigma_2^2 \mathbb{E}\left[1\land(A_d - A)^2 \right] \\
    &\to 0 \quad\textup{as $d\to \infty$,}
\end{align*}
by the weak convergence $(A_d-A)\implies 0$ as $d\to \infty$, because the function $f(x) = 1 \land x^2$ is continuous and bounded and $f(0) = 0$. The squeeze theorem takes us to the claimed convergence. The unformity of the convergence follows from the bound.
\end{proof}

\begin{lemma} \label{lemma:conv-2}
Let $(Z_1, Z_2) \sim \bvn(\sigma_1^2,\sigma_2^2,\eta)$, let $(A,B,C)$ be constants and let $\{A_d,B_d\} \to \{A,B\}$ in probability and $C_d \to C$ in $L_1$ as $d \to \infty$. Then, it holds that
\begin{equation*}
    \lim_{d \to \infty} \mathbb{E}\left[ C_d 1\land \exp(Z_1 - A_d) \land \exp(Z_2 - B_d)\right] = C \mathbb{E}\left[ 1 \land \exp( Z_1 - A) \land \exp( Z_2 - B)\right].
\end{equation*}
Furthermore, the convergence is uniform over $(\sigma_1^2, \sigma_2^2, \eta)$ in a compact set provided that $\{A_d,B_d,C_d\} \to \{A,B,C\}$ uniformly over the same set.
\end{lemma}
\begin{proof}
We shall actually show uniform convergence in $L_1$. By Lemma~\ref{lemma:cross-term-bound}, we have that
\begin{align*}
	\mathbb{E}\left[ \left|C_d 1 \land e^{(Z_1 - A_d)  \land (Z_2 - B_d)} - C 1 \land e^{(Z_1 - A)  \land (Z_2 - B)} \right|\right] &\le \mathbb{E}\left[\left| C_d - C \right| \right] + C \mathbb{E}\left[1 \land \left| A_d - A \right| \right] + \\
	&\quad+ C \mathbb{E}\left[1 \land \left| B_d - B \right| \right]\\
	&\to 0 \quad\textup{as $d\to \infty$,}
\end{align*}
since $C_d \to C$ in $L_1$ and $\{A_d,B_d\} \to \{A,B\}$ in probability as $d \to \infty$. The squeeze theorem takes us to the claimed convergence; the uniformity of the convergence follows from the bound.
\end{proof}

Lemma~\ref{lemma:gauss-projections} below is used to prove Proposition~\ref{prop:corr-grad-proj}; it recalls a standard fact concerning the joint distribution of projections a multivariate Gaussian.
\begin{lemma} \label{lemma:gauss-projections}
    Let $Z \sim \mathcal{N}_d(0_d,I_d)$ be independent of $u,v\in\mathbb{R}^d$. Then,
    \begin{equation*}
        \left(u^\top Z, v^\top Z\right) \bigm\lvert \left\{\|u\|^2, \|v\|^2, u^\top v\right\} \sim \bvn \left( \|u\|^2, \|v\|^2; u^\top v \right).
    \end{equation*}
\end{lemma}

Lemma~\ref{lemma:coupling-min} below is used to prove the optimality of GCRN in Theorems~\ref{thm:gcrn-opt-stdgauss} and~\ref{thm:gcrn-opt-elliptgauss}; it characterizes the optimal coupling under the utility function $u(x,y) = x \land y$. 

\begin{lemma} \label{lemma:coupling-min}
	Let $\mu,\nu$ be real-valued distributions with finite first moments and let $\Gamma(\mu,\nu)$ be the set of all couplings of $(\mu,\nu)$. Then,
	\begin{equation*}
		\max_{(X,Y)\in\Gamma(\mu,\nu)}\mathbb{E}[X \land Y] = \min_{(X,Y)\in\Gamma(\mu,\nu)}\mathbb{E}[ |X - Y| ].
	\end{equation*}
	Furthermore, the optimal coupling is $X = F_\mu^{-1}(U)$, $Y = F_\nu^{-1}(U)$, where $U\sim\unif(0,1)$ and where $(F_{\mu},F_{\nu})$ denote the cumulative distribution functions of $(\mu,\nu)$ respectively.
\end{lemma}
\begin{proof}
	We have that $|X - Y| = X+Y - 2 X \land Y.$ Since $\mathbb{E}[X+Y]$ is independent of the coupling, the first claim immediately follows. The second claim follows by \citet[Remark 2.19(iii)]{villani2003topics}.
\end{proof}

Lemma~\ref{lemma:sublinear-growth-max-iid} states that the expected maximum of positive i.i.d. random variables grows strictly sub-linearly; it is used to show the limiting drift in Proposition~\ref{prop:drift-elliptgauss}. See \cite{correa2021expected-maximum-iid} for a simple proof.

\begin{lemma}[{\citealp[Theorem~6]{downey1990expected-maximum-iid}}]\label{lemma:sublinear-growth-max-iid}
Let $\mu$ be a positive real-valued distribution such that $\mathbb{E}_\mu[X] < \infty$ and let $(X_i)_{i=1}^d \iid \mu$. Then,
\begin{equation*}
    \lim_{d \to \infty}\mathbb{E}\left[ \max\left\{X_1,\dots,X_d\right\} \right] /d = 0.
\end{equation*}
\end{lemma}

Lemma~\ref{lemma:gaussian-integrals} evaluates some Gaussian integrals that appear in the limiting drifts of Propositions~\ref{prop:drift-stdgauss} and~\ref{prop:drift-elliptgauss}.

\begin{lemma}\label{lemma:gaussian-integrals}
Let $\alpha, \beta, \ell > 0$ and let $Z \sim \mathcal{N}_1(0,1)$. Then, it holds that
\begin{align*}
    \mathbb{E}\left[Z 1\land e^{-\ell \alpha Z - \ell^2/2}\right] &= -\ell\alpha e^{\ell^2(\alpha^2 - 1)/2} \Phi \left(\frac{\ell}{2\alpha}-\ell\alpha \right),\\
    \mathbb{E}\left[1 \land e^{- \ell \alpha Z - \ell^2/2} \land e^{-\ell \beta Z - \ell^2/2} \right] &= \Phi\left(-\frac{\ell}{2m}\right) + e^{\ell^2(m^2-1)/2} \left\{ \Phi \left(\frac{\ell}{2m} - lm \right) - \Phi(-\ell m) \right\} \\
&\quad + e^{\ell^2(M^2 - 1)/2} \Phi(-\ell M),
\end{align*}
where $m = \alpha \land \beta$ and $M = \alpha \lor \beta$.
\end{lemma}

Lemma~\ref{lemma:h-of-rho} below establishes some properties of the function $h:[-1,1] \times (0,\infty) \to (0, \infty)$,
\begin{equation*}
    h(\rho; \ell) = \mathbb{E}_{(Z_1,Z_2) \sim \bvn(\rho)} \left[1\land e^{\ell Z_1 - \ell^2/2}\land e^{\ell Z_2 - \ell^2/2}\right],
\end{equation*}
and is the key to our fixed-point analyses in Propositions~\ref{prop:fixedpoint-stdgauss} and ~\ref{prop:fixedpoint-elliptgauss}.

\begin{lemma}[Properties of $h$] \label{lemma:h-of-rho}
The function $h(\cdot)$ has the following properties:
\begin{enumerate}
\item $h(1; \ell) = 2 \Phi(-\ell/2)$.
\item $\partial_\rho h(\rho;\ell) > 0$ for all $\rho \in (-1,1)$ and $\lim_{\rho \nearrow 1}\partial_\rho h(\rho;\ell) = \infty$. 
\item $\partial^2_\rho h(\rho;\ell) > 0$ for all $\rho \in [0,1)$.
\end{enumerate}
\end{lemma}

The proofs of Lemmas~\ref{lemma:gaussian-integrals} and~\ref{lemma:h-of-rho} rely on repeated applications of elementary calculus. As they are not instructive, we postpone them to Appendix~\ref{app:proofs-postponed} at the end of this section.

\subsection{Proof of Proposition~\ref{prop:corr-grad-proj}}

This is a consequence of Lemma~\ref{lemma:gauss-projections}. It is clear that $(n_x^{\top} Z_x, n_y^{\top} Z_y)$ is bivariate Gaussian under each coupling, and that its coordinates have unit variance. The covariance is coupling-specific:
\begin{itemize}
	\item For CRN, $Z_y = Z_x$ and therefore $\cov(n_x^{\top} Z_x, n_y^{\top} Z_y) = n_x^{\top}n_y.$
	\item For reflection, $Z_y = Z_x - 2 (e^{\top} Z_x) e$ and so 
		\begin{equation*}
		    \cov(n_x^{\top} Z_x, n_y^{\top} Z_y) = \cov(n_x^{\top} Z_x, n_y^{\top} Z_x - 2(n_y^{\top}e)e^{\top}Z_x) = n_x^{\top}n_y - 2(n_y^{\top}e)(n_x^{\top}e).
		\end{equation*}
	\item For GCRN, $n_x^{\top} Z_x = n_y^{\top} Z_y$ and so trivially $\cov(n_x^{\top} Z_x, n_y^{\top} Z_y) = 1.$
\end{itemize}
This concludes the proof.

\subsection{The process $W$ is Markov in the standard Gaussian case} \label{sec:3d-process-markov}

Suppose that the target is standard Gaussian $\pi = \mathcal{N}_d(0_d,I_d)$. We assume that $\lVert X_t \rVert, \lVert Y_t \rVert \ne 0$; dealing with the case when one or both are null is straightforward and the proof is omitted. Let $\widehat{X}_t = X_t / \lVert X_t \rVert$ and $\widehat{Y}_t = Y_t / \lVert Y_t \rVert$.

In order to show that $W$  is Markov, it is sufficient to show that $\{(\lVert X_t \rVert^2, \lVert Y_t \rVert^2, X_t^{\top}Y_t)\}_{t \ge 0}$ is Markov. We have the following expressions, shared by all three couplings:
\begin{equation*}
\begin{aligned}
\lVert X_{t+1} \rVert^2 &= \lVert X_t \rVert^2 + \big(2h \lVert X_t \rVert \widehat{X}_t^{\top}Z_x + h^2 \lVert Z_x \rVert^2 \big) B_x,\\
\lVert Y_{t+1} \rVert^2 &= \lVert Y_t \rVert^2 + \big(2h \lVert Y_t \rVert \widehat{Y}_t^{\top}Z_y + h^2 \lVert Z_y \rVert^2 \big) B_y,\\
X_{t+1}^{\top} Y_{t+1} &= X_t^{\top} Y_t + h \lVert Y_t \rVert \widehat{Y}_t^\top Z_x B_x + h\lVert X_t \rVert \widehat{X}_t^{\top}Z_y B_y + h^2 Z_x^{\top}  Z_y B_x B_y,
\end{aligned}
\end{equation*}
where $B_x = \mathbbm{1}\big\{\log U \le - h \lVert X_t \rVert \widehat{X}_t^{\top}Z_x - h^2\lVert Z_x \rVert^2/2\big\}$ and analogously for $B_y$. Since $U$ is independent of the remaining randomness, $\big( \widehat{X}_t^{\top}Z_x, \widehat{Y}_t^{\top}Z_x, \widehat{X}_t^{\top}Z_y, \widehat{Y}_t^{\top}Z_y, \lVert Z_x \rVert^2, \lVert Z_y \rVert^2 \big),$ it suffices to show that the joint distribution of these six random variables is uniquely determined by the triplet $(\lVert X_t \rVert^2, \lVert Y_t \rVert^2, X_t^{\top}Y_t)$.

Consider the projections of $Z_x,Z_y$ onto $\textup{span}\{X_t,Y_t\}$ and its orthogonal complement separately. Under the couplings of Section~\ref{sec:coupling-list} (CRN, reflection, or GCRN), the joint randomness reduces to a $\chi^2_{d-2}$ random variable and an independent 6-dimensional multivariate normal with zero mean and a covariance matrix that is uniquely determined by $(\lVert X_t \rVert^2, \lVert Y_t \rVert^2, X_t^{\top}Y_t)$. (The calculations themselves are straightforward, but tedious, and are omitted.) It follows that the process of interest $W$ is Markov, which concludes the proof.

\subsection{Proof of Proposition~\ref{prop:drift-stdgauss}}

Throughout this proof, we condition on $W_{t/d} = \left(\lVert X_{t}\rVert^2, \lVert Y_{t}\rVert^2, X_{t}^\top Y_{t} \right)/d = (x,y,v) \in \mathcal{S}.$
% let $\mathbb{E}$ denote the expectation conditional on 

Firstly, we expand the drift to 
\begin{equation*}
d(W_{(t+1)/d} - W_{t/d}) = \left(\lVert X_{t+1}\rVert^2 - \lVert X_{t}\rVert^2, \lVert Y_{t+1}\rVert^2 - \lVert Y_{t}\rVert^2,  X_{t+1}^\top Y_{t+1} - X_{t}^\top Y_{t} \right).
\end{equation*}
Further expanding the first and last terms, we have that
\begin{equation}\label{eqn:pf-drift-stdgauss-1}
\begin{aligned}
	\lVert X_{t+1}\rVert^2 - \lVert X_{t}\rVert^2 &= (2h X_t^\top Z_x +  h^2 \lVert Z_x \rVert^2 ) B_x, \\
 	X_{t+1}^{\top} Y_{t+1} - X_t^{\top} Y_t &= h Y_t^{\top} Z_x B_x + h X_t^{\top} Z_y B_y + h^2 Z_x^{\top} Z_y B_x B_y,
\end{aligned}
\end{equation}
where $B_x = \mathbbm{1}\{ U \le \exp( - h X_t^\top Z_x -  h^2 \lVert  Z_x \rVert^2 /2 ) \}$ and similarly for $B_y$ using the same $U\sim\unif(0,1)$. Except for the jump concordance $h^2 Z_x^{\top} Z_y B_x B_y$, all terms are coupling-independent, so we deal with these first.

\subsubsection{Coupling-independent terms}

Integrating over $U\sim\unif(0,1)$, we have that 
\begin{align*}
	\mathbb{E}\left[\lVert X_{t+1}\rVert^2 - \lVert X_{t}\rVert^2\right] 
&= \mathbb{E}\left [ (-2Z_1 + h^2 \lVert Z_x \rVert^2)1 \land \exp \left( Z_1 -  h^2 \lVert  Z_x \rVert^2 /2 \right) \right],\\
\mathbb{E}\left[h Y_t^{\top} Z_x B_x\right] &= -\mathbb{E}\left [ Z_2 1 \land \exp \left( Z_1 -  h^2 \lVert  Z_x \rVert^2 /2 \right) \right],
\end{align*}
where $(Z_1,Z_2) = (- h X_t^\top Z_x, - h Y_t ^\top  Z_x )$. By Lemma~\ref{lemma:gauss-projections}, we have that 
$$(Z_1,Z_2) \sim \bvn (h^2\lVert  X_t \rVert^2, h^2\lVert  X_t \rVert^2; h^2 X_{t}^\top Y_{t}) = \bvn(\ell^2 x, \ell^2 y; \ell^2 v).$$

We have that $\lim_{d \to \infty} h^2 \lVert  Z_x \rVert^2 = \ell^2$ in $L_1$ uniformly over $(x,y,v) \in \mathcal{S}$, for any compact $\mathcal{S}$. By Lemmas~\ref{lemma:conv-1} and~\ref{lemma:conv-2} it follows that
\begin{align*}
	\lim_{d \to \infty}\mathbb{E}\left[\lVert X_{t+1}\rVert^2 - \lVert X_{t}\rVert^2\right] 
		&= \mathbb{E}\left [ (-2Z_1 + \ell^2) 1 \land \exp \left( Z_1 -  \ell^2 /2 \right) \right],\\
	\lim_{d \to \infty}\mathbb{E}\left[h Y_t^{\top} Z_x B_x\right] 
		&= -\mathbb{E}\left [ Z_2 1 \land \exp \left( Z_1 -  \ell^2/2 \right) \right],
\end{align*}
uniformly over the same set. The desired formulae follow by Lemma~\ref{lemma:gaussian-integrals}; the first limit is also derived in \cite{christensen2005scaling}. The quantities $\lim_{d \to \infty}\mathbb{E}\left[\lVert Y_{t+1}\rVert^2 - \lVert Y_{t}\rVert^2\right]$ and  $\lim_{d \to \infty}\mathbb{E}\left[ h X_t^{\top} Z_y B_y\right]$ follow by symmetry.

This completes the calculations relating to the coupling-independent terms. We now turn to the coupling-dependent term, the jump concordance. 

\subsubsection{Correlation of projections}

To evaluate the limiting expected jump concordance, we must express $ \rho = \cov(n_x^\top Z_x, n_y^\top Z_y)$ as a function of $(x,y,v)$. Recall that $e = (X_t - Y_t) / \lVert X_t - Y_t \rVert$; the normalized gradient at $X_t$ is $n_x = -X_t / \lVert X_t \rVert$. We turn to Proposition~\ref{prop:corr-grad-proj} and compute 
\begin{equation*}
	n_x^\top n_y = \frac{X_t^\top Y_t}{\lVert X_t \rVert \lVert Y_t \rVert} = \frac{v}{xy^{1/2}}, \quad 
	n_x^\top e = \frac{X_t^\top (Y_t - X_t)}{\lVert X_t \rVert \lVert X_t - Y_t \rVert} = \frac{v-x}{x^{1/2}(x+y-2v)^{1/2}},
\end{equation*}
and also $n_y^\top e  = (v-y) \{y(x+y-2v)\}^{-1/2}$ by symmetry. Plugging these into Proposition~\ref{prop:corr-grad-proj}, we have that
\begin{equation*}
    \rho_\textup{crn}= \frac{v}{(xy)^{1/2}}, \quad  \rho_\textup{refl} = \frac{2xy-(x+y)v}{(xy)^{1/2}(x+y-2v)}, \quad \rho_\textup{gcrn} = 1.
\end{equation*}

\subsubsection{Coupling-dependent term}

Integrating over $U\sim\unif(0,1)$, we have that 
\begin{align*}
	\mathbb{E}\left[h^2 Z_x^{\top} Z_y B_x B_y\right] 
	&= \mathbb{E}\left[ h^2  Z_x^\top Z_y 1 \land \exp \left(Z_3 -  h^2 \lVert  Z_x \rVert^2 /2 \right) \land \exp \left(Z_4 -  h^2 \lVert  Z_y \rVert^2 /2 \right) \right],
\end{align*}
where $(Z_3,Z_4) = (- h X_t^\top Z_x, - h Y_t ^\top  Z_y) = (\ell x^{1/2} n_x^\top Z_x,\ell y^{1/2}n_y^\top Z_y)$. By Proposition~\ref{prop:corr-grad-proj}, we have that 
$$ (Z_3,Z_4) \sim \bvn(\ell^2 x, \ell^2 y; \ell^2 (xy)^{1/2} \rho),$$
where $\rho$ is coupling-specific and evaluated above.

We have the following limits in $L_1$ as $d\to\infty$: $h^2 \lVert  Z_x \rVert^2 \to \ell^2,$ $h^2 \lVert  Z_y \rVert^2 \to \ell^2$ and $h^2  Z_x^\top Z_y \to \ell^2$. The latter limit holds for all considered couplings, as they are low-rank perturbations of the CRN coupling. Moreover, the limits hold uniformly over $(x,y,v) \in \mathcal{S}$ for any compact $\mathcal{S}$. By Lemma~\ref{lemma:conv-2} it follows that
\begin{align*}
	\lim_{d \to \infty}\mathbb{E}\left[h^2 Z_x^{\top} Z_y B_x B_y\right] 
	&= \ell^2 \mathbb{E}\left[ 1 \land \exp \left(Z_3 -  \ell^2 /2 \right) \land \exp \left(Z_4 -  \ell^2 /2 \right) \right],
\end{align*}
and this limit is uniform the same set. Putting the limits together concludes the proof.

\subsection{Proof of Proposition~\ref{prop:fluctuations-stdgauss}}

Repeat the proof of Proposition~\ref{prop:drift-stdgauss} up to and including the expansions~\eqref{eqn:pf-drift-stdgauss-1}. We bound the second moment of each term in these expansions; all bounds hold due to $B_{x,y} \in [0,1]$:
\begin{align*}
	\mathbb{E}[(2h X_t^\top Z_x  B_x)^2] &\le \mathbb{E}[(2h X_t^\top Z_x)^2] = \mathbb{E}_{Z \sim \mathcal{N}(0,1)}[4h^2 \lVert X_t \rVert^2 Z^2] = 4 \ell^2 x, \\
	\mathbb{E}[(h^2 \lVert Z_x \rVert^2 B_x)^2] &\le \mathbb{E}[h^4 \lVert Z_x \rVert^4 ] = \ell^4 (d+2)/d \le 3\ell^4,\\
	\mathbb{E}[(h^2  Z_x^\top Z_y B_x B_y)^2] &\le \mathbb{E}[ h^4  (Z_x^\top Z_y)^2 ] \le  \mathbb{E}[  h^4 (\lVert Z_x \rVert^4 + \lVert Z_y \rVert^4)/2] = \ell^4 (d+2)/d \le 3\ell^4,
\end{align*}
where for the last two bounds we used the moments of $\lVert Z_x \rVert^2 \sim \chi^2_d$ and that $d \ge 1$. By symmetry, $\mathbb{E}[(2h Y_t^\top Z_y  B_y)^2] \le 4 \ell^2 y$, $\mathbb{E}[(h Y_t^\top Z_x  B_x)^2] \le \ell^2 y$ and $\mathbb{E}[(h X_t^\top Z_y  B_y)^2] \le \ell^2 x$. These bounds are independent of the coupling and the dimension $d$.

Now, by Cauchy-Schwarz, $(x+y+z)^2 \le 3(x^2 + y^2 + z^2)$ for all $x,y,z\in\mathbb{R}$. Using this inequality twice, we obtain that $\mathbb{E} [d^2 \lVert W_{(t+1)/d} - W_{t/d}\rVert^2]\le f_\ell(x,y,v)$ for all $d \ge 1$ and for all couplings, where $f_\ell(x,y,v)$ is a linear combination of the moment bounds obtained above; $f_\ell(\cdot)$ is therefore linear in all of its arguments and is independent of $d$. It follows that, for any compact set $\mathcal{S} \subset [0, \infty)^3$,
\begin{equation*}
    \lim_{d \to \infty}\sup_{(x,y,v) \in \mathcal{S}} \mathbb{E} \left[ d^2\lVert W_{(t+1)/d} - W_{t/d} \rVert^2 \right] \le \lim_{d \to \infty}\sup_{(x,y,v) \in \mathcal{S}} f_\ell(x,y,v) < \infty.
\end{equation*}
This concludes the proof.

\subsection{Proof of Theorem~\ref{thm:ode-limit-stdgauss}}

We show here that the infinitesimal generator of the process $W^{(d)}$ converges to that of the ODE $\dot{w(t)}=c_\ell(w(t))$ as $d \to \infty$. Textbook results concerning the convergence of stochastic processes then allow us to conclude. This proof mirrors that of \citet[Theorem 1]{christensen2005scaling} and is identical for each coupling of Section~\ref{sec:coupling-list}.

\subsubsection{Technical preliminaries}

We first recall some standard technical results. Let $\Sbar=\{(x,y,v): x\ge 0, y \ge 0, v\le \sqrt{xy} \}$ and let $\mathcal{C}^\infty$ be the set of all infinitely differentiable functions $h:\Sbar \to \mathbb{R}^3$ with compact support. Let $c:\Sbar \to \mathbb{R}^3$ be as in Proposition~\ref{prop:drift-stdgauss} and let $w:[0, \infty) \to \Sbar$ be the solution to the ordinary differential equation (ODE) $\dot{w}(t) = c(w(t))$. Let $G$ be the infinitesimal generator of $\dot{w}(t) = c(w(t))$, which we recall satisfies $Gh(x) = \nabla h(x)^{\top} c(x)$ for all $h \in \mathcal{C}^\infty$ and $x \in \mathcal{S}$ \citep[e.g.][Theorem 7.3.3]{oksendal1998stochastic}. Moreover, the set $\mathcal{C}^\infty$ is a core for the infinitesimal generator of the ODE $\dot{w}(t) = c(w(t))$ \citep[e.g.][Theorem 31.5, as an ODE is a L\'evy process]{sato1999levy}.

\subsubsection{Convergence of generator}

We now proceed to the main body of the proof, showing the convergence of the discrete time generator to the continuous time generator.

Let $G^{(d)}$ be the discrete time infinitesimal generator of the process $W = W^{(d)}$ and let $h \in \mathcal{C}^{\infty}$. For $w \in\Sbar$, a Taylor expansion gives
\begin{align*}
    G^{(d)}h(w) &= \mathbb{E}\left[ h(W_{(t+1)/d}) - h(W_{t/d}) \mid W_{t/d} = w \right]d\\
    &= \mathbb{E}\left[\nabla h(w)^{\top}(W_{(t+1)/d} - w) \mid W_{t/d} = w \right] d \\
    &\quad+ \frac{1}{2} \mathbb{E}\left[(W_{(t+1)/d} - w)^{\top} \nabla^2 h(w^*) (W_{(t+1)/d} - w) \mid W_{t/d} = w \right]d,
\end{align*}
for some $w^*$ on the segment from $W_{(t+1)/d}$ to $w$, where $\nabla^2 h$ denotes the Hessian of $h$.

Recall that the support of $h$ is compact. Proposition~\ref{prop:drift-stdgauss} therefore implies that the first term above converges to $Gh(w) = \nabla h(w)^\top c(w)$, uniformly over $w \in \Sbar$. We claim that the second term converges to zero uniformly over the same set. To see this, since $h \in \mathcal{C}^{\infty}$, it follows that there exists an $M < \infty$ such that $\sup_{x \in \Sbar}\lVert \nabla^2 h(x) \rVert_\infty \le M$, where $\lVert \cdot \rVert_\infty$ is the sup-norm. It follows that 
\begin{equation*}
(W_{(t+1)/d} - w)^{\top} \nabla^2 h(w^*) (W_{(t+1)/d} - w) \le M \lVert W_{(t+1)/d} - w \rVert^2.
\end{equation*}
Proposition~\ref{prop:fluctuations-stdgauss} then takes us to the claimed convergence.

Altogether, we have that
\begin{equation}\label{eqn:thm1-pf}
    \lim_{d \to \infty} \sup_{w \in \Sbar} \left|G^{(d)}h(w) - Gh(w) \right| = 0,
\end{equation}
that is the convergence of the infinitesimal generators. The limit \eqref{eqn:thm1-pf} is analogous to \citet[Eqn.~7]{christensen2005scaling}.

\subsubsection{Convergence of stochastic process}

The final part of the proof promotes the convergence of the infintesimals $G^{(d)} \to G$ to the weak convergence of the stochastic processes $W^{(d)} \implies w$. Since $\mathcal{C}^{\infty}$ is a core for the generator $G$, the limit~\eqref{eqn:thm1-pf} is equivalent to point (i) of \citet[Theorem 17.28]{kallenberg2021foundations}. Point (iv) of the same theorem concludes the proof.

\subsection{Proof of Proposition~\ref{prop:fixedpoint-stdgauss}}

Recall that the drift is $c_\ell(x,y,v) = (a_\ell(x), a_\ell(y), b_\ell(x,y,v))$, where
\begin{align*}
	a_\ell(x) ={}& \ell^2(1-2x) e^{\ell^2(x-1)/2} \Phi\Big(\frac{\ell}{2x^{1/2}} - \ell x^{1/2}\Big) +\ell^2\Phi\Big(-\frac{\ell}{2x^{1/2}}\Big), \\
	\begin{split}
	b_\ell(x,y,v) ={}& \ell^2\mathbb{E}_{(Z_1,Z_2)\sim \bvn(\rho(x,y,v))} \left[1\land e^{\ell x^{1/2}Z_1 - \ell^2/2}\land e^{\ell y^{1/2}Z_2 - \ell^2/2}\right] \\
	&- \ell^2v\Big[ e^{\ell^2(x-1)/2} \Phi\Big(\frac{\ell}{2x^{1/2}} - \ell x^{1/2}\Big) + e^{\ell^2(y-1)/2} \Phi\Big(\frac{\ell}{2y^{1/2}} - \ell y^{1/2}\Big) \Big],
	\end{split}
\end{align*}
and where $\rho(\cdot)$ is coupling-specific:
\begin{equation*}
	\rho_\textup{crn}(x,y,v) = \frac{v}{(xy)^{1/2}}, \quad  \rho_\textup{refl}(x,y,v) = \frac{2xy-(x+y)v}{(xy)^{1/2}(x+y-2v)}, \quad \rho_\textup{gcrn}(x,y,v) = 1.
\end{equation*}

The fixed points are the solutions of $a_\ell(x) = a_\ell(y) = b_\ell(x,y,v) = 0$. Since all but one of the coordinates of the ODE are autonomous, the fixed points are stable if and only if $\partial_x a_\ell(x)<0$, $\partial_y a_\ell(y)<0$  and $\partial_v b_\ell(x,y,v) <0$. The remainder of the proof is entirely elementary calculus.

\subsubsection{Fixed points and their stability}

We start with the fixed points and their linear stability analysis. Firstly, by \citet[Lemma~4.1]{kuntz2019diffusion}: $a_\ell(x) > 0$ for $x \in [0,1)$; $a_\ell(1) = 0$;  $a_\ell(x) < 0$ for $x \in (1,\infty)$. It follows that the unique solution of $a_\ell(x) = 0$ is $x^* = 1$, and furthermore this solution must be stable.

The fixed points are therefore of the form $(1,1,v^*)$, with stability in the first two coordinates and where $v^* \in[-1,1]$ is a root of $b_\ell(1,1,v) = 0$. By Lemma~\ref{lemma:h-of-rho}, we can re-write this equation to
\begin{equation*}
	h_\ell(\rho(v)) - vh_\ell(1) = 0,
\end{equation*}
where $h_\ell(\rho) = \mathbb{E}_{(Z_1,Z_2)\sim \bvn(\rho)} \big[1\land e^{\ell Z_1 - \ell^2/2}\land e^{\ell Z_2 - \ell^2/2}\big]$ and $\rho(\cdot)$ is coupling-specific and takes the values 
\begin{equation*}
\rho_\textup{crn}(v) = v, \quad \rho_\textup{refl}(v) = \rho_\textup{gcrn}(v) = 1.
\end{equation*}
For both the GCRN and reflection couplings, since $h_\ell(1) > 0$ it follows that $v^* = 1$ is the unique fixed point, and it is stable since $\partial_v g_\ell(1)  = -h_\ell(1) < 0$. For the CRN coupling, let $g_\ell(v) = h_\ell(v) - vh_\ell(1)$. By  Lemma~\ref{lemma:h-of-rho}, $g_\ell$ is convex on $[0,1]$ and satisfies $g_\ell(0) > 0$, $g_\ell(1) = 0$ and $\lim_{v \to 1}\partial_v g_\ell(v) = \infty$. It follows that there are two fixed points: $v^*_\textup{u} = 1$, which is unstable and $v^*_\textup{crn} \in(0, 1)$, which is stable as we must have $g_\ell(v^*_{\textup{crn}}) < 0$ due to the convexity of $g_\ell$. This concludes the proof.

\subsection{Proof of Theorem~\ref{thm:gcrn-opt-stdgauss}}

Condition on the same quantities as in the proof of Proposition~\ref{prop:drift-stdgauss}.

\subsubsection{GCRN attains the claimed upper bound}

From the proof of Proposition~\ref{prop:drift-stdgauss}, under the GCRN coupling it holds that
\begin{equation*}
\lim_{d\to\infty} \mathbb{E}\left[ h^2 Z_x^{\top} Z_y B_x B_y \right] = \ell^2\mathbb{E}_{Z\sim \mathcal{N}(0,1)} \left[1\land e^{\ell x^{1/2}Z - \ell^2/2}\land e^{\ell y^{1/2}Z - \ell^2/2}\right].
\end{equation*}
This coincides with the claimed limit supremum. 

To complete the proof, it is sufficient to show that the quantity on the right-hand side is an upper upper bound on the limit supremum over $\mathcal{P}$ of the left-hand side. We now show this by an argument similar to the proof of Lemma~\ref{lemma:conv-2}.

\subsubsection{Claimed upper bound on the limit supremum}

We have the sequence of bounds
\begin{align*}
    \mathbb{E} [h^2 Z_x^{\top}  Z_y B_x B_y] 
    &= \ell^2 \mathbb{E} [B_x B_y] + \ell^2 \mathbb{E} \left[\left(\frac{1}{d} Z_x^\top Z_y - 1\right) B_x B_y \right] \\ 
    & \le \ell^2 \mathbb{E} [B_x B_y]  + \frac{\ell^2}{2} \mathbb{E}\left[\Big(\frac{1}{d} \|Z_x\|^2 + \frac{1}{d}\|Z_y\|^2 -2 \Big) B_x B_y \right] \\
	& \le \ell^2 \mathbb{E} [B_x B_y]  + \frac{\ell^2}{2} \mathbb{E}\left[\Big|\frac{1}{d} \|Z_x\|^2 + \frac{1}{d}\|Z_y\|^2 -2 \Big| \right] \\
    & \le \ell^2 \mathbb{E} [B_x B_y]  + \ell^2 \mathbb{E}\left[\Big|\frac{1}{d} \|Z_x\|^2 -1 \Big| \right],
\end{align*}
where in the second line we have used that $B_x B_y \ge 0$ and that $x^\top y \le (\|x\|^2 + \|y\|^2)/2$ for all $x,y\in \mathbb{R}^d$; in the third line that $B_x B_y \le 1$; in the final line the triangle inequality and that $\{Z_x, Z_y\}$ are equal in distribution. Since $\lim_{d \to \infty}\|Z_x\|^2 /d =1$ in $L_1$, the second term in the bound converges to zero; moreover, the convergence is uniform over the kernel coupling $\bar{K}$ used.

We now bound the first term above. Let $(Z_1,Z_2) = (n_x^\top Z_x, n_y^\top Z_y)$ so that 
\begin{align*}
	\mathbb{E} [B_x B_y] 
		&= \mathbb{E} \left[ \mathbbm{1}\{U_x \le  e^{\ell x^{1/2} Z_1 - (\ell^2/2)  \|Z_x\|^2 /d } \}  \mathbbm{1}\{U_y \le  e^{\ell y^{1/2} Z_2 - (\ell^2/2)  \|Z_y\|^2 /d }  \} \right].
\end{align*}
Since we are in the class $\mathcal{P}$ of product couplings, we have that $(U_x,U_y)$ are independent of $\{Z_x,Z_y,Z_1,Z_2\}$. Taking expectations with respect to the uniforms first, it follows that
\begin{align*}
	\mathbb{E} [B_x B_y] 
		&\le \mathbb{E} \left[ 1 \land e^{\ell x^{1/2} Z_1 - (\ell^2/2)  \|Z_x\|^2 /d } \land  e^{\ell y^{1/2} Z_2 - (\ell^2/2)  \|Z_y\|^2 /d }\right]\\
		&\le \mathbb{E} \left[ 1 \land e^{\ell x^{1/2} Z_1 - \ell^2/2 } \land  e^{\ell y^{1/2} Z_2 - \ell^2/2 }\right] +\ell^2 \mathbb{E}\left[1 \land \Big| \frac{1}{d}\|Z_x\|^2 -1 \Big|\right],
\end{align*}
where finally we have used Lemma~\ref{lemma:cross-term-bound} and that $\{Z_x, Z_y\}$ are equal in distribution. The second term converges to zero as $d \to \infty$, uniformly over the kernel coupling $\bar{K}$ used. For the first term, Lemma~\ref{lemma:coupling-min} implies that
\begin{equation*}
	\sup_{\bar{K} \in \mathcal{P}}\mathbb{E} \left[ 1 \land e^{\ell x^{1/2} Z_1 - \ell^2/2 } \land  e^{\ell y^{1/2} Z_2 - \ell^2/2 }\right] = \mathbb{E}_{Z \sim \mathcal{N}(0,1)} \left[ 1 \land e^{\ell x^{1/2} Z - \ell^2/2 } \land  e^{\ell y^{1/2} Z - \ell^2/2 }\right],
\end{equation*}
since $Z_1,Z_2 \sim \mathcal{N}(0,1)$ and since $f_a(z) = 1 \land \exp(az-\ell^2/2)$ is increasing for all $a\ge0$. 

Putting all bounds together, we have that
\begin{equation*}
	\adjustlimits\lim_{d\to\infty}\sup_{\bar{K} \in \mathcal{P}} \mathbb{E}\left[ h^2 Z_x^{\top} Z_y B_x B_y \right] \le \ell^2\mathbb{E}_{Z\sim \mathcal{N}(0,1)} \left[1\land e^{\ell x^{1/2}Z - \ell^2/2}\land e^{\ell y^{1/2}Z - \ell^2/2}\right].
\end{equation*}
This concludes the proof.

\subsection{Proof of Theorem~\ref{thm:gcrn-opt-elliptgauss}}

This proof is virtually identical to that of Theorem~\ref{thm:gcrn-opt-stdgauss}. Recall that the target is $\pi= \mathcal{N}(0, \Omega^{-1})$ and that we have defined the inner product $\langle x,y\rangle_{[k]} = x^{\top} \Omega^k y$ and the squared norm $\lVert x \rVert_{[k]}^2 = x^{\top} \Omega^k x$. Throughout the proof, we condition on
$$\left(\lVert X_t \rVert^2_{[2]}, \lVert Y_t \rVert^2_{[2]},\langle X_{t},Y_{t} \rangle_{[2]}\right)/(z_1 d) = (x_2,y_2,v_2).$$
We have that $B_x = \mathbbm{1}\{U \le \exp(- h\langle X_t, Z_x \rangle_{[1]} - h^2 \lVert Z_x \rVert_{[1]}^2/2)\}$ is the acceptance step, with the analogous expression for $B_y$ with the same $U\sim\unif(0,1)$.

\subsubsection{GCRN attains the claimed upper bound}

Fix the coupling to be GCRN. By Proposition~\ref{prop:corr-grad-proj}, we have that 
$$(- h\langle X_t, Z_x \rangle_{[1]}, - h\langle Y_t, Z_y \rangle_{[1]}) = (\ell^2 z_1 x_2 Z, \ell^2 z_1 y_2 Z) = (\lambda x_2 Z, \lambda y_2 Z)$$ 
where $Z \sim \mathcal{N}(0,1)$. The following limits hold as $d \to \infty$: $h^2 \lVert Z_x \rVert_{[1]}^2 \to \ell^2 z_1 = \lambda^2$ and $h^2 \lVert Z_y \rVert_{[1]}^2 \to \lambda^2$ in probability (by Assumption~\ref{assumption:ellipt-1} and the law of large numbers); $h^2 Z_x^\top Z_y \to \ell^2$ in $L_1$. By Lemma~\ref{lemma:conv-2}, it follows that
\begin{equation*}
\lim_{d\to\infty} \mathbb{E}\left[ h^2 Z_x^{\top} Z_y B_x B_y \right] = \ell^2\mathbb{E}_{Z\sim \mathcal{N}(0,1)} \left[1\land e^{\lambda x_2^{1/2}Z - \lambda^2/2}\land e^{\lambda y_2^{1/2}Z - \lambda^2/2}\right],
\end{equation*}
so the GCRN coupling attains the upper bound claimed in Theorem~\ref{thm:gcrn-opt-elliptgauss}.

To complete the proof, we show that the above quantity is indeed an upper upper bound on the limit supremum over $\mathcal{P}$ of the left-hand side.

\subsubsection{Upper bound on limit supremum}

As in the proof of Theorem~\ref{thm:gcrn-opt-stdgauss}, we have the coupling-independent bound
\begin{equation*}
	\begin{aligned}
	\sup_{\bar{K} \in \mathcal{P}}\mathbb{E} [h^2 Z_x^{\top}  Z_y B_x B_y] 
	&\le \ell^2\mathbb{E}_{Z \sim \mathcal{N}(0,1)} \left[ 1 \land e^{\lambda x_1^{1/2} Z - \lambda^2/2 } \land  e^{\lambda y^{1/2} Z - \lambda^2/2 }\right] \\
	&\quad + \ell^2 \mathbb{E}\left[\Big|\frac{1}{d} \|Z_x\|^2 -1 \Big| \right] + \ell^2 \mathbb{E}\left[1 \land \Big| h^2\|Z_x\|_{[1]}^2 - \lambda^2 \Big|\right].
	\end{aligned}
\end{equation*}
The second term tends to zero as $d \to \infty$, and by Assumption~\ref{assumption:ellipt-1} so does the third. It follows that
\begin{equation*}
\adjustlimits \lim_{d\to\infty} \sup_{\bar{K}\in \mathcal{P}}\mathbb{E}\left[ h^2 Z_x^{\top} Z_y B_x B_y \right] \le \ell^2\mathbb{E}_{Z\sim \mathcal{N}(0,1)} \left[1\land e^{\lambda x_2^{1/2}Z - \lambda^2/2}\land e^{\lambda y_2^{1/2}Z - \lambda^2/2}\right],
\end{equation*}
which concludes the proof.

\subsection{Proof of Proposition~\ref{prop:drift-elliptgauss}}

This proof almost is identical to that of Proposition~\ref{prop:drift-stdgauss}, though some additional bookkeeping is required. Recall that the target is $\pi= \mathcal{N}(0, \Omega^{-1})$ and that we have defined the inner product $\langle x,y\rangle_{[k]} = x^{\top} \Omega^k y$ and the squared norm $\lVert x \rVert_{[k]}^2 = x^{\top} \Omega^k x$. Recall that we condition on 
$$\left(\lVert X_t \rVert^2_{[j]}, \lVert Y_t \rVert^2_{[j]},\langle X_{t},Y_{t} \rangle_{[j]}\right)/\left(z_{j-1} d\right) = (x_j,y_j,v_j) \quad \text{for all } j \in \{-1,0,1,2\}$$
in all expectations below.

Unless specified otherwise, the claims that follow hold over all $k \in \{-1,0,1\}$. The relevant one-step differences are
\begin{equation}\label{eqn:pf-drift-elliptgauss-1}
\begin{aligned}
    \lVert X_{t+1} \rVert_{[k]}^2 - \lVert X_t \rVert_{[k]}^2 ={}& \left\{2h \langle X_t, Z_x \rangle_{[k]} + h^2 \lVert Z_x \rVert_{[k]}^2\right\} B_x, \\
   \langle X_{t+1}, Y_{t+1}\rangle_{[k]} - \langle X_{t}, Y_{t}\rangle_{[k]} ={}& h \langle Y_t, Z_x\rangle_{[k]} B_x + h \langle X_t, Z_y\rangle_{[k]} B_y + h^2 \langle Z_x, Z_y\rangle_{[k]} B_xB_y. 
\end{aligned}
\end{equation}
where $B_x = \mathbbm{1}\{U \le \exp(- h\langle X_t, Z_x \rangle_{[1]} - h^2 \lVert Z_x \rVert_{[1]}^2/2)\}$, with the analogous expression for $B_y$ with the same $U\sim\unif(0,1)$. All terms except for $h^2 \langle Z_x, Z_y\rangle_{[k]} B_xB_y$ are coupling-independent, so we deal with them first.

\subsubsection{Coupling-independent terms}

Integrating over $U\sim\unif(0,1)$, we have that 
\begin{equation*}
\begin{aligned}
	\mathbb{E}\left[\lVert X_{t+1} \rVert_{[k]}^2 - \lVert X_t \rVert_{[k]}^2\right] 
		&= \mathbb{E}\left [ (-2Z_2 +  h^2 \lVert Z_x \rVert_{[k]}) 1 \land \exp \left(Z_1 -  h^2 \lVert Z_x \rVert_{[1]}^2/2 \right) \right],\\
	\mathbb{E}\left[h \langle Y_t, Z_x\rangle_{[k]} B_x\right] 
		&= -\mathbb{E}\left [ Z_3 1 \land \exp \left( Z_1 -  h^2 \lVert Z_x \rVert_{[1]}^2/2 \right) \right],
\end{aligned}
\end{equation*}
where $(Z_1,Z_2,Z_3) = (- h\langle X_t, Z_x \rangle_{[1]}, - h \langle X_t, Z_x \rangle_{[k]}, - h \langle Y_t, Z_x \rangle_{[k]})$. By Lemma~\ref{lemma:gauss-projections}, we have that
\begin{equation*}
\begin{aligned}
	(Z_1,Z_2) &\sim \bvn(\ell^2 z_1 x_2, \ell^2 z_{2k-1} x_{2k};\ell^2 z_{k} x_{k+1}) \\
	(Z_1,Z_3) &\sim \bvn(\ell^2 z_1 x_2, \ell^2 z_{2k-1} y_{2k};\ell^2 z_{k} v_{k+1}).
\end{aligned}
\end{equation*}
By Assumption~\ref{assumption:ellipt-2} and the law of large numbers, we have the following limits in $L_1$ as $d \to \infty$: $h^2 \lVert Z_x \rVert_{[1]}^2 \to \ell^2 z_1$ and $h^2 \lVert Z_x \rVert_{[k]}^2 \to \ell^2 z_k$. By Lemmas~\ref{lemma:conv-1} and~\ref{lemma:conv-2}, and using the short-hand $\lambda^2 = \ell^2 z_1$, it follows that 
\begin{equation*}
\begin{aligned}
	\lim_{d \to \infty}\mathbb{E}\left[\lVert X_{t+1} \rVert_{[k]}^2 - \lVert X_t \rVert_{[k]}^2\right] 
		&= \mathbb{E}\left [ \left(-2Z_2 +  \ell^2 z_k\right) 1 \land \exp \left(Z_1 -  \lambda^2 /2 \right) \right],\\
	\lim_{d \to \infty}\mathbb{E}\left[h \langle Y_t, Z_x\rangle_{[k]} B_x\right] 
		&= -\mathbb{E}\left [ Z_3 1 \land \exp \left(Z_1 - \lambda^2/2 \right) \right].
\end{aligned}
\end{equation*}
The analytical formulae for these quantities follow from Lemma~\ref{lemma:gaussian-integrals}. The expressions for 
\begin{equation*}
\lim_{d \to \infty}\mathbb{E}\left[\lVert Y_{t+1}\rVert_{[k]}^2 - \lVert Y_{t}\rVert_{[k]}^2\right], \quad \lim_{d \to \infty}\mathbb{E}\left[ h \langle X_t, Z_y\rangle_{[k]} B_y\right]
\end{equation*}
follow by symmetry.

\subsubsection{Correlation of projections}

To evaluate the term $h^2 \langle Z_x, Z_y\rangle_{[k]} B_xB_y$, we must express $\rho = \cov(n_x^\top Z_x, n_y^\top Z_y)$ in terms of the quantities we have conditioned on. Recall that $e = (X_t - Y_t) / \lVert X_t - Y_t \rVert$; the normalized gradient at $X_t$ is $n_x = -\Omega X_t / \lVert \Omega X_t \rVert$. We turn to Proposition~\ref{prop:corr-grad-proj} and compute 
\begin{equation*}
	n_x^\top n_y = \frac{\langle X_t, Y_t \rangle_{[2]} }{\lVert X_t \rVert_{[2]} \lVert Y_t \rVert_{[2]}} = \frac{v_2}{(x_2y_2)^{1/2}}, \quad 
	n_x^\top e = \frac{\langle X_t, Y_t \rangle_{[1]} -  \lVert X_t\rVert_{[1]}^2} {\lVert X_t \rVert_{[2]} \lVert X_t - Y_t \rVert_{[0]}} = \frac{v_1 - x_1}{(z_1x_2)^{1/2}\{z_{-1}(x_0 + y_0 - 2v_0)\}^{1/2}},
\end{equation*}
and also $n_y^\top e = (v_1 - y_1) (z_1 y_2)^{-1/2}\{z_{-1}(x_0 + y_0 - 2v_0)\}^{-1/2}$ by symmetry. Plugging these into Proposition~\ref{prop:corr-grad-proj}, and using that $\ecc =z_1 z_{-1}$, we have that
\begin{equation*}
    \rho_\textup{crn} = \frac{v_2}{(x_2y_2)^{1/2}}, \quad \rho_\textup{refl} = \frac{v_2}{(x_2y_2)^{1/2}} + \frac{2(x_1 - v_1)(y_1 - v_1)}{\ecc(x_2y_2)^{1/2} (x_0 + y_0 - 2 v_0)}, \quad \rho_\textup{gcrn} = 1,
\end{equation*}
as claimed.

\subsubsection{Coupling-dependent term}

Integrating over $U\sim\unif(0,1)$, we have that 
\begin{equation*}
	\mathbb{E}\left[h^2 \langle Z_x, Z_y\rangle_{[k]} B_xB_y\right] = \mathbb{E}\left[h^2 \langle Z_x, Z_y\rangle_{[k]} 1 \land e^{Z_4 -  h^2 \lVert  Z_x \rVert_{[1]}^2 /2 } \land e^{Z_5-  h^2 \lVert  Z_y \rVert_{[1]}^2 /2} \right],
\end{equation*}
where $(Z_4,Z_5) = (- h \langle X_t, Z_x \rangle_{[1]}, - h\langle Y_t, Z_y \rangle_{[1]}) = (\lambda x_2^{1/2} n_x^\top Z_x, \lambda y_2^{1/2} n_y^\top Z_y)$. By Proposition~\ref{prop:corr-grad-proj},  we have that 
$$ (Z_4,Z_5) \sim \bvn(\lambda^2 x_2, \lambda^2 y_2; \lambda^2 (x_2y_2)^{1/2} \rho),$$
where $\rho = \cov(n_x^\top Z_x, n_y^\top Z_y)$ is coupling-specific and evaluated above.

By Assumption~\ref{assumption:ellipt-2}, the following limits hold in $L_1$ as $d\to\infty$: $h^2 \lVert  Z_x \rVert^2_{[1]} \to \lambda^2$, $h^2 \lVert  Z_y \rVert^2_{[1]} \to \lambda^2$ and $ h^2 \langle Z_x, Z_y\rangle_{[k]} \to \ell^2 z_{k}$. We prove the final limit at the end; it holds since all couplings are low-rank perturbations of the CRN coupling. By Lemma~\ref{lemma:conv-2} it follows that
\begin{equation*}
	\lim_{d \to \infty}\mathbb{E}\left[h^2 Z_x^{\top} Z_y B_x B_y\right] = \ell^2 z_{k} \mathbb{E}\left[ 1 \land e^{Z_4 -  \lambda^2 /2} \land e^{Z_5 -  \lambda^2 /2} \right].
\end{equation*}
Putting the limits together concludes the proof, up to showing that $ h^2 \langle Z_x, Z_y\rangle_{[k]} \to \ell^2 z_{k}$ in $L_1$.

\subsubsection{Convergence of limiting inner product}

We finally show that, for all $k \in\{-1,0,1\}$ and under all considered couplings, it holds that: $\lim_{d \to \infty} \langle Z_x, Z_y\rangle_{[k]}/d \to z_{k}$ in $L_1$.

For the CRN coupling, $\langle Z_x, Z_y\rangle_{[k]}/d = \lVert Z_x \rVert_{[k]}^2/d$. The claimed limit follows by Assumption~\ref{assumption:ellipt-2}.

For the reflection coupling, $\langle Z_x, Z_y\rangle_{[k]} = \lVert Z_x \rVert^2_{[k]} - 2 (e^\top Z_x) \langle e, Z_x\rangle_{[k]}$, so it suffices to show that $(e^\top Z_x) \langle e, Z_x\rangle_{[k]} /d \to 0$ in $L_1$. Now, using the representation~\eqref{eqn:elliptgauss-prod-assumption} for the precision matrix, $\langle e, Z_x\rangle_{[k]}$ is mean-zero Gaussian with standard deviation at most $\max\{\omega_1^{2k}, \dots \omega_d^{2k}\}$, so by Cauchy-Schwarz we have that
\begin{align*}
    \mathbb{E}\left[|(e^\top Z_x) \langle e, Z_x\rangle_{[k]}| \right]/d \le \mathbb{E}\left[ \max\{\omega_1^{2k}, \dots \omega_d^{2k}\} \right]/d.
\end{align*}
By Assumption~\ref{assumption:ellipt-2}, Lemma~\ref{lemma:sublinear-growth-max-iid} applies and the right-hand side tends to zero as $d \to \infty$. This leads us to the claimed convergence for the reflection coupling.

For the GCRN coupling, we have that 
\begin{align*}
	\langle Z_x, Z_y\rangle_{[k]}&= \langle Z  - (n_x^\top Z_x + Z_1) n_x, Z  - (n_y^\top Z_y + Z_1) n_y \rangle_{[k]}\\
	&= \lVert Z \rVert^2_{[k]}  - (n_x^\top Z_x + Z_1) \langle n_x , Z \rangle_{[k]} - (n_y^\top Z_y + Z_1) \langle n_y , Z \rangle_{[k]}\\
	&\quad + (n_x^\top Z_x + Z_1) (n_y^\top Z_y + Z_1) \langle n_x, n_y \rangle_{[k]}.
\end{align*}
As for the reflection coupling, when scaled by $d^{-1}$, each of the final three ``residual" terms tends to zero in $L_1$. This is since we can bound the expected absolute value of each residual by $4\mathbb{E}\left[ \max\{\omega_1^{2k}, \dots \omega_d^{2k}\} \right]$. The claimed convergence follows, which completes the proof.

\subsection{Proof of Proposition~\ref{prop:fixedpoint-elliptgauss}}

We have dealt with the CRN and GCRN couplings in the proof of Proposition~\ref{prop:fixedpoint-stdgauss}.

For the reflection coupling, we seek the roots $v^*$ of
\begin{equation*}
	h_\ell(v + \ecc^{-1}(1-v)) - v h_\ell(1) = 0, \quad v \in [0,1],
\end{equation*}
where $h_\ell(\rho)= \mathbb{E}_{(Z_1,Z_2)\sim \bvn(\rho)} \big[1\land e^{\ell Z_1 - \ell^2/2}\land e^{\ell Z_2 - \ell^2/2}\big]$. An equivalent formulation is obtained by reparametrizing with $v = (w - \ecc^{-1}) / (1 - \ecc^{-1})$:
\begin{equation*}
	g_\ell(w):= h_\ell(w) - (w - \ecc^{-1}) / (1 - \ecc^{-1}) h_\ell(1) = 0.
\end{equation*}
By Lemma~\ref{lemma:h-of-rho}, $g_\ell$ is convex and satisfies $g_\ell(0) > 0$, $g_\ell(w) = 0$ and $\lim_{w \to 1}\partial_w g_\ell(w) = \infty$. It follows that there are two fixed points: $w^*_\textup{u} = 1$, which is unstable and $w^*_\textup{refl} \in (0,1)$, which is stable as we must have $g_\ell(w^*_{\textup{refl}}) < 0$ due to the convexity of $g_\ell$.

Returning to the original parametrization, $v^*_\textup{u} = 1$ is unstable and $v^*_\textup{refl} =  (w^*_\textup{refl} - \ecc^{-1}) / (1 - \ecc^{-1}) \in (0,1)$ is stable. To show that $v^*_\textup{refl} \ge v^*_\textup{crn}$, recall that
\begin{equation*}
	0 = h_\ell(v^*_\textup{refl} + \ecc^{-1}(1-v^*_\textup{refl})) - v^*_\textup{refl} h_\ell(1) \ge  h_\ell(v^*_\textup{refl}) - v^*_\textup{refl} h_\ell(1).
\end{equation*}
Given that $v^*_\textup{crn}$ is the unique solution of $h_\ell(v) - v h_\ell(1) = 0$ over $v \in (0,1)$, Lemma~\ref{lemma:h-of-rho} and the above inequality imply that $v^*_\textup{refl} \ge v^*_\textup{crn}$.

\subsubsection[Behaviour with eccentricity epsilon]{Behaviour with eccentricity $\ecc$}

Let $f_\ell(\ecc,v) = h_\ell(v + \ecc^{-1}(1-v)) - v h_\ell(1)$, so that the fixed-point equation is $f_\ell(\ecc,v^*_\textup{refl}) = 0$.  By Lemma~\ref{lemma:h-of-rho}, $f_\ell(\ecc,v)$ is decreasing in $\ecc$, and furthermore (due to the convexity of $f_\ell$ in $v$) it holds that $f_\ell(\ecc,v)$ is decreasing in $v$ near $v^*_\textup{refl}$. Since, $f_\ell(\ecc,v^*_\textup{refl}) = 0$, it follows that $v^*_\textup{refl}$ must be decreasing in $\ecc$.

We conclude by showing the limits. Since $f_\ell(1,v) = h_\ell(1) - v h_\ell(1)$, we have that $\lim_{\ecc \to 1}v^*_\textup{refl}(\ecc) = 1$. Since $f_\ell(\infty,v) = h_\ell(v) - v h_\ell(1)$, we have that $\lim_{\ecc \to \infty}v^*_\textup{refl}(\ecc) = v^*_\textup{crn}$. This concludes the proof.

\subsection{Proof of Theorem~\ref{thm:gcrn-opt-product}}

We proceed as in the proof of Theorem~\ref{thm:gcrn-opt-stdgauss}, first proving the upper bound, then showing that the GCRN coupling attains it.

\subsubsection{Showing the upper bound}

We have that 
\begin{align}
    \mathbb{E} [h^2 Z_x^{\top}  Z_y B_x B_y] 
    &= \ell^2 \mathbb{E} [B_x B_y] + \ell^2 \mathbb{E} \left[\left(\frac{1}{d} Z_x^\top Z_y - 1\right) B_x B_y \right]\label{eqn:product-target-decomposition} \\ 
    & \le \ell^2 \mathbb{E} [B_x] + \frac{\ell^2}{2} \left\{  \mathbb{E}[(\frac{1}{d} \|Z_x\|^2 + \frac{1}{d}\|Z_y\|^2 -2) B_x B_y]\right\} \nonumber\\
    & \le \ell^2 \mathbb{E} [B_x] + \frac{\ell^2}{2} \left\{  \mathbb{E}[(\frac{1}{d} \|Z_x\|^2 -1)_+] + \mathbb{E}[(\frac{1}{d} \|Z_y\|^2 -1)_+]\right\}, \nonumber \\
    & = \ell^2 \mathbb{E} [B_x] + \ell^2 \mathbb{E}\left[\left(\frac{1}{d} \|Z_x\|^2 -1\right)_+\right], \nonumber
\end{align}
where $(x)_+ = 0 \lor x$ is the positive part of $x$; we have used that $B_{y} \in [0,1]$ and $Z_x^\top Z_y \le (\|Z_x\|^2 + \|Z_y\|^2)/2$ to obtain the second line, and that $B_{x,y} \in [0,1]$ to obtain the third line. 

The bound obtained above is invariant to the coupling used. The first term converges to $2 \Phi(-\ell (bI)^{1/2}/2)$ by \citet[Theorem~5]{roberts2001optimal} (this calculation is also performed below), while the second term converges to zero since $\|Z_x\|^2/d \to 1$ in $L_2$. It follows that
\begin{equation*}
    \lim_{d\to\infty} \sup_{\bar{K} \in \mathcal{M}}\mathbb{E} [h^2 Z_x^{\top}  Z_y B_x B_y] \le 2 \ell^2 \Phi(-\ell (bI)^{1/2}/2),
\end{equation*}
which is the claimed bound.

\subsubsection{Showing that GCRN attains the bound}

To show that the GCRN coupling attains the claimed bound, we hereafter restrict to this coupling and we return to the decomposition~\eqref{eqn:product-target-decomposition}. 

To get a handle on the first term of~\eqref{eqn:product-target-decomposition}, we Taylor-expand the logarithm of the acceptance ratio for the $X$-chain \citep[as in][proof of Theorem~5]{roberts2001optimal},
\begin{align*}
    &\log \pi (X+h Z_x) - \log \pi (X) =\\
    &\quad=h \nabla \log \pi(X)^\top Z_x + \frac{h^2}{2} Z_x^\top \nabla^2 \log \pi(X) Z_x + R_x\\ 
    &\quad= \Big\{ \frac{1}{d} \sum_{i=1}^d (\ell \omega_i)^2 \left[(\log f)' (\omega_i X_i)\right]^2 \Big\}^{1/2} Z_\nabla + \frac{1}{2d} \sum_{i=1}^d (\ell \omega_i Z_{x,i})^2 (\log f)'' (\omega_i X_i) + R_x\\
    &\quad=: G_x^{1/2} Z_\nabla + H_x,
\end{align*}
where $Z_\nabla \sim \mathcal{N}_1(0,1)$ corresponds to the gradient direction, and $R_x$ is the third-order remainder term.

By the law of large numbers, the following limit holds in probability:
\begin{align*}
    \lim_{d \to \infty}G_x 
    &= \ell^2  \mathbb{E}_{X \sim \pi^{(1)}} \left[ \omega_1^2 (\log f)' (\omega_i X) ^2\right] \\
    &= \ell^2  \mathbb{E}_{Y \sim f} \left[ \omega_1^2 (\log f)' (Y) ^2\right]  \\
    &=\ell^2 b I.
\end{align*}
Our assumptions ensure that the remainder term satisfies $\lim_{d \to \infty}R_{x} = 0$ in probability \citep[see e.g.][Lemma~6]{sherlock2013optimal}. By this and the law of large numbers, the following limit holds in probability:
\begin{align*}
    \lim_{d \to \infty}H_x 
    &= \frac{1}{2} \mathbb{E}_{X \sim \pi^{(1)}} \left[ (\ell \omega_1 Z_{x,1})^2 (\log f)'' (\omega_1 X)\right]\\
    &= \frac{\ell^2  b}{2} \mathbb{E}_{Y \sim f} \left[ (\log f)'' (Y)\right]\\
    &= -\frac{\ell^2  bI}{2},
\end{align*}
where in the last equality we have used the identity $\mathbb{E}_{Y\sim f}[(\log f)'' (Y)] = - \mathbb{E}_{Y\sim f}[(\log f)' (Y)^2] = -I$, which follows by integration by parts.

The analogous Taylor expansion for the $Y$-chain is: $\log \pi (Y+h Z_y) - \log \pi (Y) = G_y^{1/2} Z_\nabla + H_y,$ where, by the definition of the GCRN coupling, the random variable $Z_\nabla \sim \mathcal{N}_1(0,1)$ corresponding to the gradient direction is \underline{identical} to the analogous one for the $X$-chain. As before, the following limits hold in probability: $\lim_{d \to \infty}G_y = \ell^2 bI$, $\lim_{d \to \infty} H_y = (-\ell^2 bI)/2$.

With the above limits in probability in hand, the first term in~\eqref{eqn:product-target-decomposition} writes as
\begin{align*}
    \ell^2 \mathbb{E} [B_x B_y] &= \ell^2 \mathbb{E} \left[1 \land \exp(G_x^{1/2} Z_\nabla + H_x) \land \exp(G_y^{1/2} Z_\nabla + H_y) \right]\\
     &= \ell^2 \mathbb{E} \left[1 \land \left[ \exp(G_x^{1/2} Z_\nabla + H_x) \left\{ 1\land \exp( \partial G \cdot Z_\nabla + \partial H) \right\} \right] \right],
\end{align*}
where both $\partial G := G_y^{1/2} - G_x^{1/2}$ and $\partial H := H_y - H_x$ go to $0$ in probability as $d \to \infty$. Several applications of Slutsky's Theorem (ST) and the Continuous Mapping Theorem (CMT) lead us to 
\begin{equation*}
    \lim_{d \to \infty} \ell^2 \mathbb{E} [B_x B_y] = 2 \ell^2 \Phi(-\ell (bI)^{1/2}/2).
\end{equation*}
\begin{itemize}
    \item By ST and CMT, $\lim_{d \to \infty} 1\land \exp( \partial G \cdot Z_\nabla + \partial H) = 1$ in probability.
    \item Again, by ST and CMT, $\exp(G_x^{1/2} Z_\nabla + H_x)$ converges weakly to a log-normal $L \sim \log\mathcal{N}(-\ell^2  bI/2, \ell^2  bI)$ as $d \to \infty$. A further application of ST, shows that the sequence of non-negative random variables $A_d = \exp(G_x^{1/2} Z_\nabla + H_x) \left[ 1\land \exp( \partial G \cdot Z_\nabla + \partial H) \right]$ converges weakly to the same limit.
    \item The function $g(x) = 1 \land x$ is bounded for $x \in [0,\infty)$. It follows that $\lim_{d \to \infty}\mathbb{E} [B_x B_y] = \lim_{d \to \infty}\mathbb{E} [1 \land A_d] = \mathbb{E}[1 \land L]$, by the definition of weak convergence.
    \item Lemma~\ref{lemma:gaussian-integrals} evaluates $\mathbb{E}[1 \land L]$ and completes the calculation.
\end{itemize}

We therefore have a limit for the first term of~\eqref{eqn:product-target-decomposition}. The second term of~\eqref{eqn:product-target-decomposition} satisfies $\lim_{d \to \infty} \ell^2 \mathbb{E} [(Z_x^\top Z_y/d - 1) B_x B_y ] = 0,$ since $Z_x^\top Z_y/d \to 1$ in $L_1$ and since $B_x B_y \in [0,1]$. Putting it all together, under the GCRN coupling we have that
\begin{equation*}
    \lim_{d\to\infty} \mathbb{E} [h^2 Z_x^{\top} Z_y B_x B_y] = 2 \ell^2 \Phi(-\ell (bI)^{1/2}/2),
\end{equation*}
which coincides with the upper bound and therefore concludes the proof of Theorem~\ref{thm:gcrn-opt-product}.

\subsection{Postponed proofs} \label{app:proofs-postponed}

\subsubsection{Proof of Lemma~\ref{lemma:gaussian-integrals}}

Let $E_1 = \mathbb{E} \left[Z 1\land e^{-\ell\alpha Z - \ell^2/2}\right]$. Using that $\mathrm{d}\phi(z) = - z \phi(z) \mathrm{d}z$, it holds that $\int_a^b z \phi(z) \mathrm{d}z = -\int_a^b\mathrm{d}\phi(z) = \phi(a) - \phi(b).$ It thus follows that
\begin{align*}
    E_1 &= \int_{-\infty}^{-\ell/(2\alpha)} z \phi(z) \mathrm{d}z + \int_{-\ell/(2\alpha)}^{\infty} z e^{-\ell\alpha z - \ell^2/2} \phi(z) \mathrm{d}z\\
    &= -\phi\left(-\frac{\ell}{2\alpha}\right) + e^{\ell^2(\alpha^2 - 1)/2}\int_{-\ell/(2\alpha)}^{\infty} z\phi(z + \ell\alpha) \mathrm{d}z\\
    &=  -\phi\left(-\frac{\ell}{2\alpha}\right) + e^{\ell^2(\alpha^2 - 1)/2} \left(\int_{-\ell/(2\alpha)}^{\infty} (z + \ell\alpha) \phi(z + \ell\alpha) \mathrm{d}z - \ell\alpha \int_{-\ell/(2\alpha)}^{\infty} \phi(z+\ell\alpha) \mathrm{d}z \right)\\
    &=  -\phi\left(-\frac{\ell}{2\alpha}\right) + e^{\ell^2(\alpha^2 - 1)/2} \left(\int_{-\ell/(2\alpha) + \ell\alpha}^{\infty} z \phi(z) \mathrm{d}z - \ell\alpha \int_{-\ell/(2\alpha) + \ell\alpha}^{\infty} \phi(z) \mathrm{d}z \right)\\
    &=  -\phi\left(-\frac{\ell}{2\alpha}\right) + e^{\ell^2(\alpha^2 - 1)/2} \phi\left(-\frac{\ell}{2\alpha} + \ell \alpha\right) - \ell \alpha e^{\ell^2(\alpha^2 - 1)/2} \Phi\left(\frac{\ell}{2\alpha} - \ell \alpha\right),
\end{align*}
where we used the identity $\phi(z) = e^{\ell^2\alpha^2/2 + \ell\alpha z}\phi(z + \ell\alpha)$ in the second line. The desired formula follows by applying the same identity with $z = -\ell/(2\alpha)$; the first two terms in the final expression cancel.

The second expectation equivalently writes as $E_2 = \mathbb{E}\left[1 \land e^{-\ell m Z - \ell^2/2} \land e^{-\ell M Z - \ell^2/2}\right],$ where $m = \alpha \land \beta$ and $M = \alpha \lor \beta$. The form of the integrand depends on the sign of $Z$. Therefore,
\begin{align*}
   E_2 
	&= \int_{-\infty}^{-\ell/(2m)} \phi(z) \mathrm{d}z + \int_{-\ell/(2m)}^{0} e^{-\ell m z - \ell^2/2}\phi(z) \mathrm{d}z + \int_{0}^{\infty} e^{-\ell M z - \ell^2/2}\phi(z) \mathrm{d}z\\
    &= \Phi\left(-\frac{\ell}{2m}\right) + e^{\ell^2(m^2 - 1)/2} \int_{-\ell/(2m)}^{0} \phi(z + \ell m)\mathrm{d}z + e^{\ell^2(M^2 - 1)/2} \int_{0}^{\infty} \phi(z + \ell M)\mathrm{d}z\\
    &= \Phi\left(-\frac{\ell}{2m}\right) + e^{\ell^2(m^2 - 1)/2} \int_{-\ell/(2m) + \ell m}^{\ell m} \phi(z)\mathrm{d}z + e^{\ell^2(M^2 - 1)/2} \int_{\ell M}^{\infty} \phi(z)\mathrm{d}z\\
   &= \Phi\left(-\frac{\ell}{2m}\right) + e^{\ell^2(m^2 - 1)/2} \left\{ \Phi\left(\frac{\ell}{2m} - \ell m\right) - \Phi(-\ell m) \right\} + e^{\ell^2(M^2 - 1)/2} \Phi(-\ell M),
\end{align*}
where we used identity $e^{-\ell\alpha z}\phi(z) = e^{\ell^2\alpha^2/2}\phi(z + \ell\alpha)$ with $\alpha\in\{m,M\}$ in the second line. This completes the proof.

\subsubsection{Proof of Lemma~\ref{lemma:h-of-rho}}

\textbf{Proof of claim 1.} This is an immediate consequence of Lemma~\ref{lemma:gaussian-integrals}.

\textbf{Proof of claim 2.} Firstly, the integral re-writes as 
\begin{equation*}
    h(\rho; \ell) = \mathbb{E}_{(Z_1,Z_2) \sim \bvn(\rho)} \left[ \exp\left\{  0 \land (\ell Z_1 - \ell^2/2) \land (\ell Z_2 - \ell^2/2) \right\} \right].
\end{equation*}
We use the reparametrization trick $Z_2 = \rho Z_1 + \sqrt{1-\rho^2}Z_*$, where $Z_*\sim\mathcal{N}_1(0,1)$ is independent of $Z_1$. This expresses $h(\cdot)$ as an integral over randomness $(Z_1,Z_*)$ which does not depend on $\rho$; thereafter, only $Z_2$ in the integrand depends on $\rho$. Differentiating and bringing the derivative inside the integral, we obtain
\begin{equation*}
    \partial_\rho h(\rho; \ell) = \mathbb{E} \left[ \mathbbm{1}\{Z_2 \le \ell/2\} \mathbbm{1}\{Z_2 \le Z_1\} \partial_\rho(\ell Z_2 - \ell^2/2) e^{\ell Z_2 - \ell^2/2}\right].
\end{equation*}
We use a second reparametrization trick: $Z_1 = \rho Z_2 + \sqrt{1-\rho^2}Z_{**}$, where $Z_{**}\sim\mathcal{N}_1(0,1)$ is independent of $Z_2$. We can now evaluate:
\begin{gather*}
    \partial_\rho Z_2 = Z_1 - \frac{\rho}{\sqrt{1-\rho^2}} Z_* = Z_1 - \frac{\rho}{\sqrt{1-\rho^2}} \frac{Z_2 - \rho Z_1}{\sqrt{1-\rho^2}}= \frac{Z_1 - \rho Z_2}{1-\rho^2} = \frac{1}{\sqrt{1-\rho^2}}Z_{**},\\
\mathbbm{1}\{Z_2 \le Z_1\} = \mathbbm{1}\left\{Z_2 \le \rho Z_2 + \sqrt{1-\rho^2}Z_{**}\right\} = \mathbbm{1}\left\{\sqrt{\frac{1-\rho}{1+\rho}}Z_2 \le Z_{**}\right\}.
\end{gather*} 
Therefore, 
\begin{align*}
    \partial_\rho h(\rho; \ell)
    &= \frac{\ell}{\sqrt{1-\rho^2}}\mathbb{E} \left[ \mathbbm{1}\{Z_2 \le \ell/2\} \mathbbm{1}\left\{\sqrt{\frac{1-\rho}{1+\rho}}Z_2 \le Z_{**}\right\} Z_{**} e^{\ell Z_2 - \ell^2/2}\right]\\
    &= \frac{\ell}{\sqrt{1-\rho^2}}\mathbb{E}\left[\mathbbm{1}\{Z_2 \le \ell/2\}  \phi\left(\sqrt{\frac{1-\rho}{1+\rho}}Z_2\right) e^{\ell Z_2 - \ell^2/2}\right],
\end{align*}
where in the last line we have used that $\mathbb{E}[Z_{**}\mathbbm{1}\{x \le Z_{**}\}]= \phi(x)$ for all $x$. It follows that $\partial_\rho h(\rho) > 0$ for all $\rho \in (-1,1)$, and that $\lim_{\rho \nearrow 1}\partial_\rho h(\rho) = \infty$, as claimed.

\textbf{Proof of claim 3.} Firstly, repeated applications of the chain rule yield:
\begin{equation*}
    \partial_\rho \left\{ \phi\left(\sqrt{\frac{1-\rho}{1+\rho}}Z_2\right) \right\} = \frac{1}{(1+\rho)^2}z^2\phi\left(\sqrt{\frac{1-\rho}{1+\rho}}z\right).
\end{equation*}
(We have used that $\partial_\rho\phi(\sqrt{f(\rho)} z) = -\partial_\rho f(\rho)z^2\phi(\sqrt{f(\rho)}z)/2$ for any non-negative differentiable $f(\cdot)$, and then that $\partial_\rho f(\rho) =  -2/(1+\rho)^2$ when $f(\rho) = (1-\rho) / (1+\rho)$.)

Now, differentiating twice, we have that
\begin{align*}
\partial_\rho^2 h(\rho; \ell) 
	&= \partial_\rho\left\{\partial_\rho h(\rho; \ell)\right\} \\
	&= \partial_\rho\left\{ \frac{\ell}{\sqrt{1-\rho^2}}\mathbb{E}\left[\mathbbm{1}\{Z_2 \le \ell/2\}  \phi\left(\sqrt{\frac{1-\rho}{1+\rho}}Z_2\right) e^{\ell Z_2 - \ell^2/2}\right] \right\}\\
	&= \frac{\rho}{1-\rho^2}\partial_\rho h(\rho; \ell) + \frac{\ell}{\sqrt{1-\rho^2}}\mathbb{E}\left[\mathbbm{1}\{Z_2 \le \ell/2\} \frac{1}{(1+\rho)^2}Z_2^2 \phi\left(\sqrt{\frac{1-\rho}{1+\rho}}Z_2\right) e^{\ell Z_2 - \ell^2/2}\right],
\end{align*}
where we have used the product rule of differentiation for the final line. By Claim 2, the first term is non-negative for $\rho\in[0,1)$; the second term is strictly positive over the same range. It follows that $\partial_\rho^2 h(\rho; \ell) > 0$ for $\rho \in[0,1)$, as claimed. This concludes the proof of Lemma~\ref{lemma:h-of-rho}.

\putbib
\end{bibunit}

\end{document}